\documentclass[a4paper,12pt]{article}
\pdfoutput=1
\usepackage{graphicx}
\usepackage{verbatim}
\usepackage[latin1]{inputenc}
\usepackage[english]{babel}
\usepackage[T1]{fontenc}
\usepackage{epsfig}
\usepackage{subfigure}
\usepackage{amsmath}
\usepackage{amssymb}
\usepackage{amsthm}
\usepackage{bbm}
\usepackage{mathrsfs}
\usepackage{enumitem}
\usepackage[longnamesfirst, sort]{natbib}
\usepackage{accents}
\usepackage{setspace}
\usepackage{xcolor} 
\usepackage[bottom]{footmisc}
\usepackage{soul}
\usepackage[hyphens]{url}
\usepackage[hidelinks]{hyperref}
\usepackage[font=footnotesize,labelfont=bf]{caption}
\usepackage{apptools}
\usepackage{array,multirow}
\AtAppendix{\counterwithin{lemma}{section}}
\AtAppendix{\counterwithin{proposition}{section}}
\AtAppendix{\counterwithin{assumption}{section}}
\AtAppendix{\counterwithin{definition}{section}}

\allowdisplaybreaks[1]

\newtheorem{corollary}{Corollary}

\newtheorem{definition}{Definition}

\newtheorem{lemma}{Lemma}

\newtheorem{proposition}{Proposition}

\newtheorem{assumption}{Assumption}

\begin{document}

\title{Research Joint Ventures:\\ The Role of Financial Constraints}

\author{Philipp Brunner, Igor Letina and Armin Schmutzler\thanks{Brunner: Department of Economics, University of Bern. Email: philipp.brunner@unibe.ch. Letina: Department of Economics, University of Bern and CEPR. Email: igor.letina@unibe.ch. Schmutzler: Department of Economics, University of Zurich and CEPR. Email: armin.schmutzler@uzh.ch. We are grateful to Andrea Calef,  Esm\'{e}e Dijk, Ernst-Ludwig von Thadden, Saish Nevrekar and Regina Seibel and to seminar audiences in Bern, Helsinki, Zurich, at Oligo (2022), Swiss Society of Economics of Statistics (2022) and Swiss IO Day (2021) for helpful comments.}}

\date{\today }

\maketitle

\begin{abstract}
This paper provides a novel theory of research joint ventures for financially constrained firms.
When firms choose R\&D portfolios, an RJV can help to coordinate research efforts, reducing investments in duplicate projects. This can free up resources, increase the variety of pursued projects and thereby increase the probability of discovering the innovation. 
RJVs improve innovation outcomes when market competition is weak or external financing conditions are bad. An RJV may increase the innovation probability and nevertheless lower total R\&D costs. RJVs that increase innovation also increase consumer surplus and tend to be profitable, but innovation-reducing RJVs also exist. Finally, we compare RJVs to innovation-enhancing mergers.
\end{abstract}

\vspace{4mm}
\begin{tabular}{ll}
\textbf{JEL Codes:}& L13, L24, O31 \\

\textbf{Keywords:}& Innovation, Research Joint Ventures, Financial Constraints,  \\  
& Mergers, Intensity of Competition, Licensing
\end{tabular}

\vspace{12mm}

\newpage
\section{Introduction}

Innovation often involves large R\&D investments. A well-known example is the pharmaceutical industry where blockbuster drugs can require high upfront R\&D expenses.\footnote{For example, see CBO's report ``Research and Development in the Pharmaceutical Industry'', available at \href{https://www.cbo.gov/publication/57126}{https://www.cbo.gov/publication/57126}.} Similarly, automobile producers have invested \textsterling 341 billion within five years on R\&D in electric vehicles.\footnote{Sources for all RJVs mentioned in this section are listed in Appendix \ref{appendix_sources}.} Even industry giants rarely attempt to take on such demanding tasks on their own. In the last few years, major players have agreed on research joint ventures (RJVs). For instance, \textit{Daimler} and \textit{Geely} jointly develop battery-driven \textit{Smart} cars. \textit{PSA} and \textit{Opel} hooked up with \textit{Saft}, a subsidiary of \textit{Total}, to develop batteries. Together with \textit{BP}, \textit{Daimler} and \textit{BMW} develop charging stations.
\textit{Renault}, \textit{Nissan} and \textit{Mitsubishi Motors} agreed on investing \$26 billion to develop common platforms for electric vehicles. Further up in the value chain, suppliers of essential inputs have also joined forces.\footnote{For instance, the German chemical firm \textit{BASF} and the Chinese firm \textit{Shanshan} jointly search for better materials to produce cathodes for batteries.} Not only are the required R\&D investments large, there is also significant uncertainty about which technology the vehicles of the future will rely on. Today, most electric vehicles are powered by lithium-ion batteries, but this technology has significant drawbacks and automotive companies are jointly exploring alternatives.\footnote{For example, \textit{Volvo} and \textit{Daimler} are collaborating on fuel-cell driven cars, while \textit{Ford} and \textit{BMW} have jointly invested in a startup developing solid-state batteries.}   In all these partnerships, at least some of the firms are competing or planning to compete in the product market.\footnote{While we will focus on such horizontal RJVs, purely vertical collaborations are common as well. For instance, \textit{Panasonic} engages in a joint venture with \textit{Toyota} to develop batteries; \textit{Volkswagen} and \textit{Stellantis} develop networks of charging stations with \textit{Enel} and \textit{ENGIE}, respectively.} 

Competition policy typically treats RJVs more leniently than other forms of horizontal cooperation. For instance, the European Union addresses RJVs either under its merger regulation or under Article 101 of the EU treaty, depending on whether it is a full-function joint venture or not. In the latter case, even if an RJV has been found to have anti-competitive object or actual or potential competition-restricting effects (Article 101(1)), it may still be justified on the basis of efficiency gains under certain conditions (Article 101(3)).\footnote{See, in particular, Commission Regulation No. 1217/2010 of 14. December 2010.} 
The legal situation in the United States is similar.\footnote{The \textit{1993 National Cooperative Research and Production Act} specifies that horizontal cooperation in RJVs is not per se illegal, but is to be evaluated under the ``rule of reason''. }

An important prerequisite for justifying a friendly approach of competition policy towards RJVs is that they foster R\&D. The literature focuses on knowledge spillovers as the main justification.\footnote{Early examples include \cite{katz1986analysis}, \cite{d1988cooperative}, \cite{kamien1992research}. See Section \ref{SecDiscussion} for a detailed literature discussion.} Our paper analyzes a different channel through which RJVs can lead to more innovation: When R\&D costs are high (so that firms are financially constrained) and there is significant uncertainty about the right way to generate the desired innovative outcome, an RJV can reduce investments in duplicate R\&D projects, thereby freeing up funds that can be invested in previously unexplored approaches. To clarify under which condition this happens, we introduce a model that combines financial constraints and uncertainty about the right way to generate the desired innovation. Contrary to previous theoretical literature on RJVs, firms choose how to spread investments over different R\&D projects rather than merely how much to invest. This feature allows us to investigate how the RJV members can benefit from reallocating scarce resources. To the best of our knowledge, we provide the first analysis of RJVs that explicitly considers project choice.

In our benchmark model, we analyze a symmetric duopoly. The firms choose in which set of R\&D projects from a continuum of alternatives to invest.
Only one of all possible projects will lead to an innovation, resulting in a positive effect on the firm's product market profits. Therefore, investing in a wider range of projects increases the likelihood of finding the right approach. Projects are identical except that some are more costly than others. 
Each firm has a fixed budget, which can be used for R\&D investments.\footnote{We can also interpret the budget as the internally available time of researchers or the laboratory's capacity, which can be expanded through (more expensive) external researchers or laboratories.} In addition, firms can borrow externally. In line with the empirical literature (see Section \ref{SecDiscussion}), we assume that firms who borrow externally have to pay a positive interest rate on the external loans. The firm chooses its investment strategy so as to maximize expected profits. We assume that the budget is sufficiently small that, in equilibrium, both firms borrow positive amounts from the financial markets.  Our analysis compares the outcome of this R\&D competition game with the alternative that the firms form an RJV in an otherwise identical setting. In the latter case, the firms combine their budgets, and the RJV chooses R\&D investments to maximize joint payoffs. Firms share the research costs equally and, if successful, both receive the innovation. After the R\&D outcomes materialize, the firms compete in the product market.

Our central results give conditions under which an RJV increases the probability of innovation. The intensity of product market competition is crucial. To see this, note that, in the absence of an RJV, an innovating firm may benefit from \textit{escaping competition}, moving ahead by being the only one who has access to a superior technology. Under an RJV, it is obviously impossible to escape competition by innovation, because firms have agreed to share the fruits of their research efforts. Instead, a successful RJV symmetrically increases the profits of both firms.
When competition is soft (e.g., price competition with sufficiently differentiated goods), so that the increase in industry profits from successful joint innovation is large relative to the benefit from escaping competition, the RJV increases the innovation probability. Interestingly, this result does \textit{not} rely on financial constraints. Moreover, like all our main results, it does not require spillovers, the driving force behind innovation-enhancing research cooperation in the literature. 

Next, we suppose that competition is not soft (for instance, homogeneous quantity competition and price competition with weakly differentiated products). In this case, the value of escaping competition is always higher than the joint profit increase from innovating together. Thus, without financial constraints, an RJV would reduce the probability of innovation. This is precisely where our modeling choices play a critical role, because they allow us to identify features of RJVs that are absent in standard models. In an RJV, the participating parties not only coordinate the decisions on how much to invest, but also in which projects. This allows them to reduce duplication and free up resources, which they can spend on further projects without having to access the capital market. When an RJV frees up enough internal funding, it can \textit{potentially} invest in a wider range of projects, compared to independent firms, using just internal funding. Whether an RJV \textit{actually} makes use of this opportunity or whether it just enjoys the cost savings from avoided duplication depends not only on the nature of competition, but also on financial constraints: When external financing conditions are sufficiently bad, then the RJV increases the innovation probability even when competition is not soft. To repeat, this result relies on the existence of financial constraints: Without them, the RJV would invest in less projects than the two independent firms together.

In the situation with relatively intense competition just described, the RJV not only increases the innovation probability, but it also reduces overall R\&D spending. Thus, total industry R\&D costs and the innovation probability do not necessarily move in the same direction. This is in stark contrast with the existing literature, which typically views an RJV as innovation-enhancing if and only if it increases total investment cost. This feature of our model underlines the importance of allowing for different R\&D projects. 

Further, under very mild assumptions, we show that any RJV that increases the probability of innovation also increases expected consumer surplus. This occurs because consumers are better off if innovation is more likely, and conditional on being discovered, if it is used by as many firms as possible. 

Overall, the results just discussed show that RJVs are helpful for inducing innovation and improving consumer welfare under a wider range of circumstances than identified by previous literature. However, in line with existing worries in EU circles, we also found circumstances under which RJVs are harmful to innovation.\footnote{See ``Guidelines on the applicability of Article 101 of the Treaty on the Functioning of the European Union to horizontal co-operation agreements.'' Official Journal of the European Union  (2011/C 11/01).} Thus, to evaluate the innovation effects of RJVs, it is decisive to understand the incentives of firms to form an RJV.
If firms only had an incentive to form RJVs that reduce innovation, then lenient policy towards them would be misguided. We thus ask: Will firms have incentives to engage in RJVs for which our analysis has shown that they enhance innovation? Or will they rather engage in RJVs that reduce innovation? We find general and widely applicable conditions under which firms benefit from forming RJVs that increase the innovation probability. In particular, this will always be true unless competition is very intense. However, we also find circumstances under which firms engage in RJVs even though they reduce overall innovation -- the cost savings in these cases suffice to make the RJVs profitable, thus giving at least some foundations to the above-mentioned concerns.

Next, we compare RJVs and mergers. Which of the two forms of cooperation is more conducive to innovation depends on the nature of product market cooperation and the stringency of financial constraints. This result relates to a recent discussion in merger control that has emphasized R\&D effects, asking whether (potentially) beneficial effects of mergers on innovation provide a justification for waving them through in spite of their well-known mark-up increasing effects. We identify a wide range of parameters for which even mergers that lead to a higher innovation probability than R\&D competition should be prohibited, as an RJV would have the same social benefits without the social costs of eliminating competition.\footnote{More broadly, authors such as \cite{farrell2000scale} have emphasized that, even if efficiency gains outweigh the competition-softening effects of a merger, competition authorities still have to ask whether the merger is actually necessary to achieve these gains.}

Moreover, we explore the link between our analysis and the more familiar rationale for RJVs that relies on knowledge spillovers.
We find that knowledge spillovers and financial constraints are complements: Financial constraints are more likely to increase the probability of innovation the stronger spillovers are, and vice versa.
Finally, we analyze the relation between licensing and RJVs. In line with previous literature, the chance to earn licensing fees increases innovation incentives under R\&D competition, so that the conditions under which an RJV yields a higher probability of innovation than R\&D competition become more restrictive. Moreover, with licensing, if an RJV increases the probability of innovation, it \textit{always} results in lower R\&D spending.

All told, our paper attempts to shed light on how the consideration of project choice and financial constraints affects the analysis of RJVs. While we ignore important aspects such as firm asymmetries, costs of RJV formation and governance issues, and we work under the debatable assumption that the RJV does not induce collusive behaviour in the product market, we are confident that our approach can be a useful input for a more comprehensive welfare analysis.\footnote{See \cite{duso2014collusion} and \cite{helland2019research} for evidence suggesting that RJVs may foster collusion. However, note that our analysis of mergers for the duopoly case can alternatively be interpreted as an RJV with full collusion in the product market.}

Section \ref{SecDiscussion} discusses our contribution in the light of existing literature. In Section \ref{SecModel}, we provide the benchmark duopoly model. Section \ref{SecEffects} analyzes the innovation effects of RJVs and identifies conditions under which they are profitable. In Section \ref{SecExtensions}, we compare RJVs and mergers. Further, we extend the analysis to the case of spillovers and to multiple firms, and we discuss licensing. Section \ref{SecConclusion} concludes.

\section{Relation to the Literature}\label{SecDiscussion}

Our paper analyzes R\&D competition between duopolists who (i) select R\&D projects and (ii) are financially constrained, comparing their choices with those of RJVs and merged firms.  Accordingly, we discuss the relation of our paper to existing treatments of R\&D project choice, financially constrained firms, RJVs and mergers and innovation.

\textbf{Innovation project choice:} Our model of R\&D competition with project choice by symmetric incumbents builds on \cite{letina2016}. \cite{letina2020start} apply that framework to study the innovation decisions of an incumbent and an entrant. These papers neither include financial constraints, nor do they address joint ventures. Contrary to these models, \cite{moraga2022mergers} allow for (two) different types of R\&D, but fix the overall spending.\footnote{Broadly related models of R\&D project choice are \cite{Gilbert2019}, \cite{bryan2017}, \cite{letina2019inducing}, \cite{bardey2016} and \cite{bavly2020}.}

\textbf{Financially constrained firms:} Authors such as \cite{hall2010financing} and \cite{kerr2015financing} have argued why external financing of R\&D is more costly than for other investments,\footnote{Examples are the riskiness of the investments and the difficulty of providing collateral, as physical assets are relatively less important than human capital.} so that internal financing plays a strong role (\cite{czarnitzki2011r}).
Several authors have found negative effects of financial constraints on R\&D.\footnote{See \cite{mohnen2008financial}, 
\cite{savignac2008impact},
\cite{mancusi2014r}, \cite{howell2017financing} \cite{krieger2017developing}, \cite{caggese2019financing}. \cite{czarnitzki2011r} also find that the availability of internal financing has a larger impact on R\&D than on capital investment and that basic research is more prone to financial constraints as it is riskier.} 
In line with our results, \cite{helland2019research} finds that capital-constrained firms are more likely to join an RJV. In contrast with the empirical literature, the theoretical literature is small.\footnote{One exception is \cite{fumagalli2022}, who consider acquisitions of startups that might be financially constrained.} We are not aware of any oligopoly model of project choice by financially constrained firms (with or without RJV formation).

\textbf{The theory of RJVs:} 
The existing theoretical literature on RJVs does not allow for different R\&D projects, and it focuses on spillovers rather than financial constraints. Without RJVs, firms invest in R\&D to gain a competitive advantage (as in our model). If knowledge spillovers to competitors are large, such gains are small and firms limit their R\&D investments: \cite{d1988cooperative} show in a two-stage Cournot duopoly that, with high spillovers, an RJV leads to higher R\&D expenditures, output and welfare than R\&D competition. With low spillovers,  R\&D investments and welfare are higher under R\&D competition than in an RJV.\footnote{For similar results with more than two firms, see \cite{suzumura1992cooperative}. \cite{amir2019spillovers} shows how subsidies can help to achieve the second-best social optimum.} As argued above, we find that, with soft competition, an RJV increases the investment probability (and hence welfare) even with low spillovers. \cite{katz1986analysis} and \cite{kamien1992research} similarly show how product market competition affects the comparison between R\&D competition and cooperation. 
Like other authors, such as \cite{amir2019spillovers}, both papers also argue that an RJV reduces wasteful effort duplication. However, firms can only choose the amount of R\&D investment, which is in a strictly positive relation with the R\&D outcome (the size or probability of an innovation). In our model with project choice, an RJV may increase the innovation probability while costs fall because duplication is eliminated. This feature is particularly relevant when there is fundamental uncertainty about the right approach to R\&D.\footnote{In broadly related work, \cite{kamien2000meet} allow firms to choose different research approaches, but approaches only differ in their spillover rates, and each approach will succeed with certainty, which is in stark contrast to our model. Other important lines of research include stochastic R\&D (\cite{choi1993cooperative}) and absorptive capacity, whereby spillovers are increasing in own R\&D (\cite{kamien2000meet}).}
The literature has also highlighted important caveats to the claim that research cooperation is socially beneficial:  RJVs foster product market collusion, which leads to dynamic inefficiency.\footnote{See \cite{grossman1986research}, \cite{martin1996r}, \cite{jacquemin1988cooperative}, \cite{caloghirou2003research} and \cite{miyagiwa2009collusion}. Conversely, \cite{vilasuso2000public} show that, if forming RJVs is costly, firms may form less RJVs than socially optimal; similarly \cite{falvey2013coordination}.}  Competition authorities who decide on such RJVs have to weigh these risks against the potential benefits, which is difficult given realistic informational constraints \citep{cassiman2000research}.

\textbf{The empirics of RJVs:} Empirical studies support various benefits of RJV identified in the theoretical literature.\footnote{See \cite{link1998case}, \cite{cassiman2002r}, \cite{becker2004r}, \cite{aschhoff2008empirical} }
\cite{veugelers1997internal} emphasize that absorptive capacity is necessary to reap these benefits. \cite{Roeller2007} show that cost-sharing motives are important for RJV formation.
 Finally, \cite{duso2014collusion} and \cite{helland2019research} find empirical evidence that RJVs among competitors are more prone to collusion, which reduces welfare. Thus, authorities should scrutinize horizontal R\&D cooperation.

\textbf{Mergers and innovation:} Several authors have studied whether incumbent mergers increase innovation.
\cite{federico2017, federico2018} and \cite{motta2021effect} identify negative effects in models with one-dimensional R\&D effort; similarly, \cite{letina2016} and \cite{Gilbert2019} obtain negative effects on R\&D diversity in models of project choice.
 \cite{denicolo2018a} find positive effects. In \cite{bourreau2018b}, both possibilities arise, where the positive effects come from allowing for horizontal rather than only vertical R\&D innovations. In our model with project choices of financially constrained firms who engage in purely vertical innovations, we similarly find that the effects of the merger can be positive or negative. Contrary to \cite{bourreau2018b}, however, the possibility of a positive effect reflects the merged entity's ability to coordinate which projects to invest in and the existence of financial constraints.

\section{Model}\label{SecModel}

Our model of R\&D with project choice builds on previous work of \cite{letina2016} and \cite{letina2020start}.\footnote{Accordingly, the model description follows those papers closely.} However, neither of these papers deals with RJVs or budget constraints. We assume that two ex-ante symmetric firms $(i\in\{1,2\})$ can invest in R\&D before they compete in the product market. There are two possible levels of technology -- current technology, which is available to both firms, and new technology, which is only available to the firms that innovate. To improve their technology level, firms can invest in multiple projects $\theta$ from the set of available projects $\Theta = [0,1)$. 
Only one of these projects is correct, that is, leads to an innovation. Let  $\hat{\theta}\in \Theta$ be the correct project. Nature chooses which of the available projects is correct, but firms are not informed about it, hence firms see the location of the correct project as a random variable. We assume that the location of the correct project is uniformly distributed on $\Theta = [0,1)$.
For each  $\theta \in [0,1)$, each firm chooses whether to invest in that research project ($r_i(\theta) =1$) or not ($r_i(\theta) =0$). If $r_i(\hat{\theta})=1$, then firm $i$ will innovate for sure and if $r_i(\hat{\theta})=0$, then firm $i$ will not innovate.\footnote{In the previous version of this paper, \cite{brunner2022research}, we allow firms to partially invest in research projects, that is, choose $r_i(\theta)\in [0,1]$. This richer model admits a symmetric equilibrium. However, all economic insights remain the same as in the current version. } We restrict the firm's choices to the set of measurable functions $r : \Theta \rightarrow \{0,1\}$, which we denote with $\mathcal{R}$. 
The cost of developing a project $\theta$ is given by $C(\theta)$, where we assume that the function $C:[0,1) \rightarrow \mathbb{R}^+$ is differentiable and strictly increasing and that $C(0)=0$ and $\lim_{\theta \rightarrow 1}C(\theta) = \infty$. Therefore, the total research costs of firm $i$ are $\int_0^1 r_i(\theta)C(\theta) d\theta$.\footnote{If this integral does not converge, we assign the value $\infty$ to it.} 

If a firm has chosen $r_i(\hat{\theta})=1$ in the investment stage, it has access to an innovation and enters the product market competition with technology state $t_i=I$. If it has not invested in $\hat{\theta}$, it does not have access to an innovation, and its technology state is $t_i=0$.

For now, we do not explicitly model product market competition. Instead, we formulate weak general assumptions that we show to hold in familiar models of product market competition in Section \ref{sec:example}. We assume that the product market profits of firm $i$ are given in reduced form by the expression $\pi_{t_it_j}$ for $j\neq i$. If both firms innovate, then they will compete with the new technology, and their market profits are given by $\pi_{II}$. 
Similarly, if both firms compete with the current technology, then each of them obtains profits $\pi_{\mathit{0}\mathit{0}}$.
If a single firm innovates, it obtains profits $\pi_{I\mathit{0}}$,
while the other firm obtains $\pi_{\mathit{0}I}$. 
We will impose the following regularity assumptions on the profit functions.
\begin{assumption}[Regularity of profit functions]\label{ass_r}
$ $
\begin{enumerate}[label=(\roman*)]
\item Profits are non-negative: $\pi_{t_it_j} \geq 0$ for all $t_i$ and $t_j$.
\item Innovation increases profits: $\pi_{II} \geq \pi_{\mathit{0}\mathit{0}}$.
 \item Competitor innovation reduces profits: $\pi_{t_i\mathit{0}} \geq \pi_{t_iI}$ for $t_i \in \{\mathit{0},I\}$.
\item Escaping competition is more valuable than catching up:\\
$\pi_{I\mathit{0}} - \pi_{\mathit{0}\mathit{0}}  \geq \pi_{II}-\pi_{\mathit{0}I}$.
\end{enumerate}
\end{assumption}

Obviously, Assumptions \ref{ass_r}(i)-(iii) are compatible with most standard oligopoly models. Furthermore, authors such as \cite{bagwell1994sensitivity}, \cite{leahy1997public}, \cite{farrell2000scale} and \cite{schmutzler2013competition} have argued that submodularity conditions like (iv) hold for many innovation games with standard models of price and quantity competition unless knowledge spillovers are strong. Intuitively, a successful innovation of the competitor reduces own equilibrium outputs and margins, which reduces the benefits from increasing margins and outputs through own innovation.

While we will always maintain that competition is sufficiently intense that Assumption \ref{ass_r}(iv) holds, we will distinguish between three different regimes according to the intensity of competition.\footnote{\cite{boone2008competition,boone2008new} similarly uses the relation between efficiency differences and profit differences in his definition of intensity of competition.}

\begin{definition}[Intensity of competition] $ $
\begin{enumerate}[label=(\roman*)]
\item Competition is \textbf{intense} if avoiding the competitor catching up is more valuable than catching up:
$\pi_{I\mathit{0}} - \pi_{II}  > \pi_{II}-\pi_{\mathit{0}I}$. 
\item Competition is \textbf{soft} if avoiding the competitor catching up is less valuable than
improving together: $\pi_{I\mathit{0}}-\pi_{II}<\pi_{II} - \pi_{\mathit{0}\mathit{0}}$.
\item Competition is \textbf{moderate} if neither of the above cases holds, so that: $\pi_{II} - \pi_{\mathit{0}I} \geq \pi_{I\mathit{0}} - \pi_{II}  \geq \pi_{II}-\pi_{\mathit{0}\mathit{0}}$.
\end{enumerate}
\end{definition}

For cost-reducing investments, competition is typically intense in a homogeneous \linebreak 
Bertrand market, but also for a homogeneous Cournot market with linear demand (see Section \ref{SecCournot}). In Section \ref{SecPriceCompetition}, we will see that all three regimes arise with differentiated price competition, depending on the degree of substitution.

Each firm has a research budget $B$. If a firm spends more than $B$, it has to borrow from the capital market at the interest rate $\rho>0$, reflecting the well-known difficulties of external financing of R\&D investments (see Section \ref{SecDiscussion}).\footnote{Although our financing assumption is simple, it captures the essence of the idea that the marginal costs of own funds, as long as they are available, are lower than the marginal costs of borrowed funds.}
We will assume (in a way which will be made precise in Assumption \ref{ass_budget}) that without an RJV the budget is binding and both firms find it optimal to borrow positive amounts from the capital market.

The expected total payoff of firm $i$, given the strategy of competitor $j$ is then
\begin{align*}
\mathbb{E}\Pi_i(r_i, r_{j})=& \int_0^1  (1-r_{j}(\theta))\left[r_i(\theta)\pi_{I\mathit{0}} +(1-r_i(\theta))\pi_{\mathit{0}\mathit{0}}\right] d\theta   \\
& +  \int_0^1  r_{j}(\theta)\left[r_i(\theta)\pi_{II} +(1-r_i(\theta))\pi_{\mathit{0}I}\right] d\theta  \\
& -\int_0^1 r_i(\theta)C(\theta) d\theta - \rho \max \left\lbrace 0,  \int_0^1 r_i(\theta)C(\theta) d\theta - B  \right\rbrace. 
\end{align*}
The first integral captures the expected payoffs when firm $j$ does not innovate. Similarly, the second integral represents the payoffs when firm $j$ innovates. The third line represents research costs, depending on whether the firm borrows from the capital market or not. Firms choose $r_i(\theta)$ and $r_{j}(\theta)$ simultaneously with the goal of maximizing $\mathbb{E}\Pi_i$ and $\mathbb{E}\Pi_j$, respectively. We will focus on pure strategy equilibria throughout.

\section{Effects of RJVs}\label{SecEffects}

\subsection{Equilibrium under R\&D Competition} \label{SecRDCompetition}

We now characterize the equilibrium strategies under R\&D competition. Given our assumptions on research costs, it is intuitive that both firms will invest in projects near $\theta=0$, whereas neither firm will invest in projects near $\theta=1$. One would thus expect equilibrium strategies to be of the following type.

\begin{definition}\label{definition2}
A \textbf{double cut-off strategy profile} is a profile $(r_i,r_j)$ of research strategies for which $\theta_L \in [0,1) $ and $\theta_H \in [\theta_L,1) $ exist such that
\begin{alignat*}{2}
    &r_i(\theta)=r_j(\theta)=1 \quad &&\text{ if } \theta<\theta_L \\
    &r_i(\theta)=r_j(\theta)=0 \quad &&\text{ if } \theta>\theta_H .
\end{alignat*}
\end{definition}

Note that the definition does not specify which firm invests for $\theta \in (\theta_L,\theta_H)$.
To find the equilibrium cut-off values, consider the equations
\begin{align*}
(1+\rho)C(\theta_1)&=\pi_{I\mathit{0}}-\pi_{\mathit{0}\mathit{0}} \\
(1+\rho)C(\theta_2)&=\pi_{II} - \pi_{\mathit{0}I}.
\end{align*}
$\theta_1$ is the most expensive project in which a firm can profitably invest using external finance, assuming that the competitor does not invest in this project. Similarly, $\theta_2$ is the most expensive project in which a firm can profitably invest using external finance, assuming that the competitor invests in this project. An immediate consequence of Assumption \ref{ass_r}(iv) is that $\theta_2 \leq \theta_1$. The following assumption guarantees that both firms will borrow positive amounts in any equilibrium.
\begin{assumption}\label{ass_budget}
$B < \int_0^{\theta_2}C(\theta) d\theta$.
\end{assumption}

Next, we characterize all equilibria of this game.\footnote{Of course, for any equilibrium strategies $r^*_i$ and $r^*_j$ there exist infinitely many equilibria which only differ on sets of measure zero. We ignore those differences and only regard strategies as distinct if they differ on sets of positive measure.}

\begin{lemma}
[Characterization of investment strategies under competition]\label{prop_equilibrium} $ $\\
\noindent (i) The research competition game has multiple equilibria. A profile of double-cut off strategies $(r^*_i,r^*_j)$ is an equilibrium if it satisfies (a) $\theta_L=\theta_2$ and $\theta_H=\theta_1$ and (b) for each $ \theta \in (\theta_2, \theta_1)$ either:
\begin{align*}
    r^*_i(\theta)&=1 \text{ and } r^*_{j}(\theta)=0 \text{ or }\\
    r^*_i(\theta)&=0 \text{ and } r^*_{j}(\theta)=1.
\end{align*}
(ii) No other pure-strategy equilibria of the research-competition game exist.
\end{lemma}

\begin{figure}
 \centering
        \includegraphics[width=0.5\textwidth]{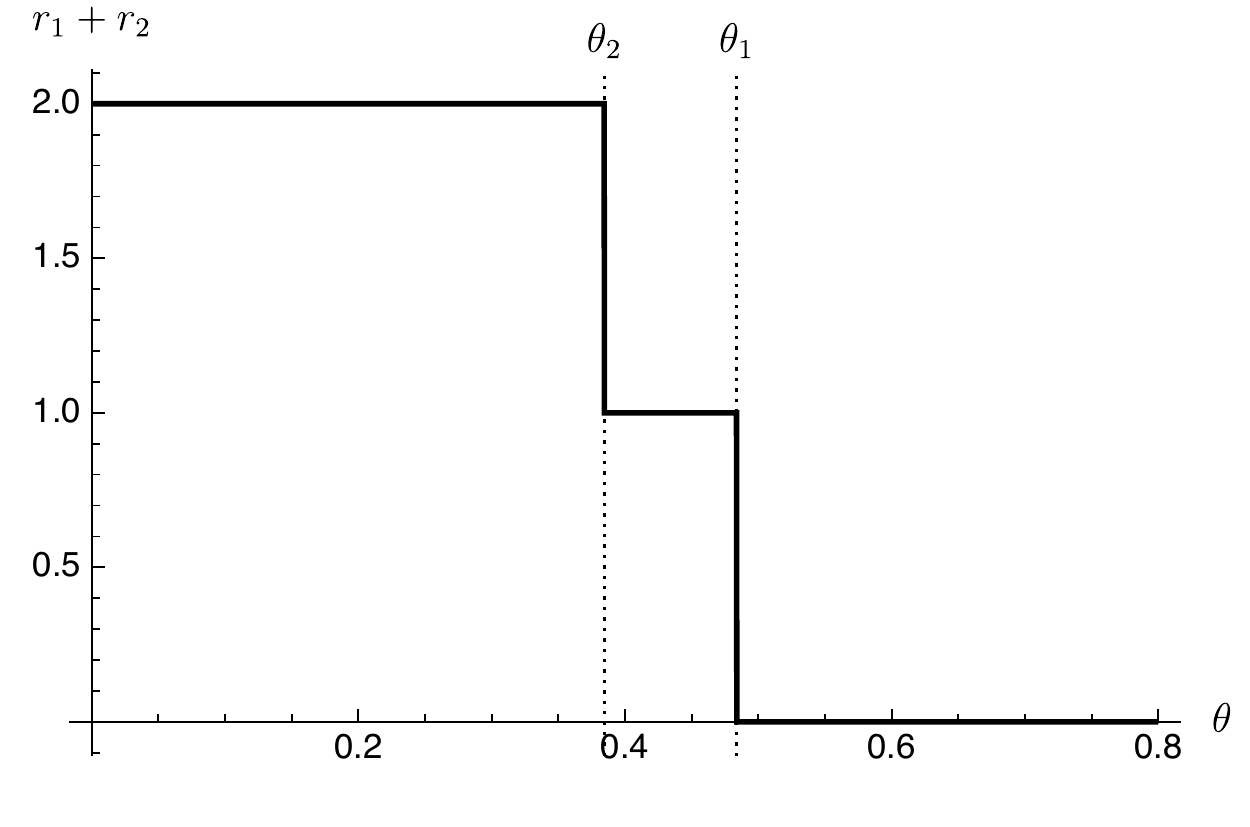} 
        \caption{Industry portfolio of research projects in any equilibrium. 
        } 
        \label{fig:AsyEquilibrium}
\end{figure}

Thus, all equilibria share the double cut-off structure, which is determined by the marginal cost of research projects and the benefit of being successful. Both firms invest in the cheap projects $\theta\in[0,\theta_2)$, while neither firm invests in the expensive projects $\theta\in(\theta_1,1]$. For intermediate projects, the marginal benefits of an innovation are higher than the marginal costs when only one firm finds the innovation,
but not when both firms are successful. Hence, for each $ \theta$ in the interval $(\theta_2, \theta_1)$,  one firm invests while the other does not invest, but the identity of the investing firm is not determined, which leads to the multiplicity of equilibria. However, all equilibria are equivalent --  in the sense that they generate the innovation with the same probability and lead to the same market structure in each state of the world. In any equilibrium, the overall innovation probability is $\theta_1$.  Furthermore, in any equilibrium, the probability of a duopoly with an innovation is $\theta_2$, the probability of a single firm with an innovation is $\theta_1-\theta_2$, while the probability of a duopoly without an innovation is $1-\theta_1$. While not necessary for our results, we can moreover show that in the presence of small cost asymmetries, the equilibrium where the low-cost firm invests in all projects in $[0,\theta_1)$ while the high-cost firm invests in all projects  in $[0,\theta_2)$ risk-dominates any other equilibrium.\footnote{We consider a difference in borrowing costs of $\epsilon>0$. See Appendix \ref{SecAppRisk} for details.}

Note that there is duplication of research efforts in equilibrium, as all projects in the interval $[0,\theta_2)$ are duplicated. 
Figure \ref{fig:AsyEquilibrium} depicts the industry portfolio of research projects in every equilibrium.

The symmetric setting of our model brings out clearly that the asymmetric outcome depicted in the figure exclusively reflects equilibrium considerations rather than exogeneous differences between firms. The value of investing in a particular project depends on the behavior of the competitor. Investing tends to be more worthwhile if the competitor does not invest than if he invests. In the former case, the resulting profit increase is 
given by the value of escaping competition ($\pi_{I\mathit{0}}-\pi_{\mathit{0}\mathit{0}}$), which by Assumption \ref{ass_r}(d) is larger than the value of catching up ($\pi_{II}-\pi_{\mathit{0}I}$), which determines the incentives for investing in projects that the competitor also invests in. The asymmetric investment behavior of firms for intermediate projects directly reflects these differences in incentives.

As the difference between the value of escaping competition and the value of catching up increases (reflecting greater intensity of competition), the area with asymmetric investment becomes larger. An increase in the value of escaping competition increases $ \theta_1$ and thus project variety and the probability of innovation. This tends to lead to more demand for external funding to finance more expensive projects. Conversely, in most standard oligopoly models, the value of catching up decreases with more intense competition, which lowers the amount of duplication and, thus, the need for external funds. Therefore, the overall effect of increased competition on external funding is ambiguous.
 Further, a higher borrowing cost $\rho$ implies a lower probability of finding the innovation and less duplication of effort since $\theta_1$ and $\theta_2$ both decrease in $\rho$. Lastly, by Assumption \ref{ass_budget}, a marginal change in the budget size B does not affect the equilibrium portfolio.

\subsection {Optimal Project Choice of an RJV}\label{SecRJVOptimum}

In our model of RJVs, the firms combine their individual budgets and invest in research together. However, the two firms still compete in the product market after the successful project has been realized.\footnote{This is the main difference to a merger, which will result in a monopolistic market in any case.} Moreover, the research costs are equally shared and both firms obtain the innovation if developed. This eliminates the possibility of an asymmetric product market structure. The firms will compete either with or without innovation. Like an individual firm, the RJV can borrow at the interest rate $\rho$ on the external market if the total budget $2B$ is insufficient. The RJV  chooses the research strategy $r_v$ to maximize the expected total payoff
\begin{align}
\mathbb{E}\Pi_v(r_v)= 2&\int_0^1 \left[ r_v(\theta)\pi_{II} +(1-r_v(\theta))\pi_{\mathit{0}\mathit{0}}\right]d\theta \nonumber \\ & - \int_0^1 r_v(\theta)C(\theta) d\theta 
  -\rho\max\left\{\int_0^1 r_v(\theta)C(\theta) d\theta - 2B,0 \right\}. \label{eq_rjv_profit}   
\end{align}

The optimal strategy will be of the following type.
\begin{definition}
A \textbf{single cut-off strategy} is a research strategy $r_v$ for which a $\theta^* \in [0,1) $ exists such that $r_v(\theta)=1$ if $\theta<\theta^* $ and $r_v(\theta)=0$ if $\theta>\theta^* $.
\end{definition}

Let $\theta^{B}$ be defined as the solution to $\int_0^{\theta^{B}} C(\theta) d\theta=2B$ if $\int_0^1 C(\theta) d\theta>2B$ and $\theta^{B}=1$ otherwise. That is, a joint venture which invests in all projects in the set $(0,\theta^{B})$ either has innovation costs equal to $2B$ or invests in all projects. Next, let $\theta^u$ and $\theta^\rho$ be the solutions to the following equations
\begin{align*}
(1+\rho)C(\theta^\rho)&=2(\pi_{II}-\pi_{\mathit{0}\mathit{0}}) \\
C(\theta^u)&=2(\pi_{II}-\pi_{\mathit{0}\mathit{0}}).
\end{align*}
Thus, $\theta^u$ is the most expensive research project in which an RJV that does not borrow from the capital market wants to invest in. Similarly, $\theta^\rho<\theta^u$ is the most expensive research project in which an RJV that has to borrow would choose to invest in. How $\theta^{B}$ relates to these two values will determine the  optimal portfolio of the RJV.

\begin{lemma}[Investment strategies of an RJV]\label{prop_RJVInvestment} $ $

The RJV chooses a single cut-off strategy with 
\begin{align*}
    \theta^* = \begin{cases} 
    \theta^\rho  \quad &\text{if} \quad \theta^{B} < \theta^\rho \\
     \theta^{B} \quad &\text{if} \quad \theta^{B} \in [\theta^\rho,\theta^u] \\
     \theta^u  \quad &\text{if} \quad \theta^{B}>\theta^u.
    \end{cases}
\end{align*}
\end{lemma}

Thus, the cut-off project always lies in the interval $[\theta^\rho,\theta^u]$. Which of the three cases in the lemma arises depends on the budget $B$, the interest rate $\rho$, on product market profits and on the cost function. If $\theta^{B}<\theta^\rho$, then the joint venture invests its entire budget $2B$ into research and, in addition, it borrows from the capital market in order to finance its research activities. In contrast with a marginal change in the cost of borrowing $\rho$, a marginal increase in the budget would not affect the investment strategy. When $\theta^{B}\in (\theta^\rho, \theta^u)$, the RJV invests the entire budget, but it does not borrow. Thus, a marginal increase in the budget would lead to an increase in investment, whereas a marginal change in $\rho$ would have no effect. Finally, when $\theta^{B} > \theta^u$, the RJV does not borrow and furthermore only invests a portion of its budget into research. Hence, neither marginal changes in $B$ nor in $\rho$ would change investment behavior, which is fully determined by product market conditions.

Note that in standard oligopoly models, the expression $\pi_{II}-\pi_{\mathit{0}\mathit{0}}$ is decreasing in standard parameterizations of the intensity of competition.\footnote{For instance, this is the case in our two examples with linear Cournot competition and differentiated price competition in Section \ref{sec:example}.} This implies that the critical cut-off projects $\theta^\rho$ and $\theta^u$ are larger when product market competition is softer. Thus, unless it is optimal to just invest the entire budget ($\theta^*=\theta^{B}$), the RJV uses more funding when competition becomes softer. 
This is in line with the findings of \cite{kamien1992research} and \cite{goyal2001r} that softer product market competition increases incentives to cooperate and leads to higher research efforts.

\subsection{R\&D Competition vs. R\&D Cooperation}\label{SecComparison}

Next, we present our central result that deals with the effect of the RJV on the probability that an innovation will be discovered. Define the interest threshold  $\bar{\rho}$ and the budget threshold $\bar{B}(\rho)$ as 
\begin{align*}
    \bar{\rho} & =
    \begin{cases}
    \dfrac{\pi_{I\mathit{0}}-\pi_{II}-(\pi_{II}-\pi_{\mathit{0}\mathit{0}})}{2(\pi_{II}-\pi_{\mathit{0}\mathit{0}})}, & \text{for } \pi_{II} > \pi_{\mathit{0}\mathit{0}} \\
    \infty, & \text{for } \pi_{II} = \pi_{\mathit{0}\mathit{0}}.
  \end{cases}\\
  \bar{B}(\rho) & = \dfrac{\int_0^{\theta_1} C(\theta)d\theta}{2} 
\end{align*}

The budget threshold depends negatively on $\rho$ because $\theta_1$ does.
The thresholds play a critical role for the effects of an RJV on innovation.

\begin{proposition}[Comparison of R\&D competition and RJV]\label{prop_RJVEffect} $ $
\begin{enumerate}[label=(\roman*)]
\item Suppose competition is soft.
Then the innovation probability is strictly larger under the RJV than under R\&D competition.
\item Suppose competition is moderate or intense. 
Then: \newline (a) The innovation probability is strictly larger under the RJV than in any equilibrium under competition if and only if $B > \bar{B}(\rho)$ and $\rho > \bar{\rho}$.\newline
(b) If the formation of the RJV strictly increases the innovation probability, then it weakly decreases total R\&D spending.
\end{enumerate}
\end{proposition}

The result reflects the subtle interplay between product market competition and financing conditions. In a model with R\&D project choice, an RJV results in efficiency gains at the investment stage --  it reduces the amount of duplication of research projects. This allows the RJV to ``cast a wider net,'' as the funds that were previously used to finance duplicate research projects can now be redirected to other projects. This duplication reduction effect of the RJV makes it less costly to sustain high innovation probabilities. However, a potential countervailing effect needs to be taken into account: 
Escaping competition can be very valuable for each individual firm. Thus, compared with an RJV, incentives for innovation may be higher for a firm that can fully appropriate the benefits from innovation as the single successful innovator under R\&D competition. If competition is soft, i.e., (i) holds, then this countervailing effect has no bite, as joint profits in an RJV are high enough that the innovation probability will be higher than under R\&D competition. As we will see in Proposition \ref{PropCompSpill}, this result does not even require the existence of financial constraints.

By contrast, Proposition \ref{prop_RJVEffect}(ii) deals with the case that product market competition is moderate or intense.
Then additional requirements are necessary for an RJV to increase innovation. Together, the condition that $\rho > \bar{\rho} $ and $B >\bar{B}(\rho)$  guarantee that the RJV will invest in more projects than both firms would in \emph{any} equilibrium without the RJV, even though product market competition is not soft.\footnote{Note that there is a tension between Assumption \ref{ass_budget} which demands that the budget is not too high and the condition in Proposition \ref{prop_RJVEffect}(ii) that $B >\bar{B}(\rho)$. The Cournot example in Section \ref{sec:example} shows that the conditions can nevertheless be satisfied together for non-degenerate parameter regions.} The advantages of the RJV in this setting come from the ability to avoid duplication and thereby finance a wider range of projects internally, thus avoiding the necessity to borrow from the capital markets. This is illustrated in Figure \ref{fig:Merged}. When either $B \leq \bar{B}(\rho)$ or $\rho \leq \bar{\rho}$, so that the conditions in (iia) are not satisfied, then RJVs (weakly) decrease the innovation probability. 

\begin{figure}
 \centering
        \includegraphics[width=0.75\textwidth]{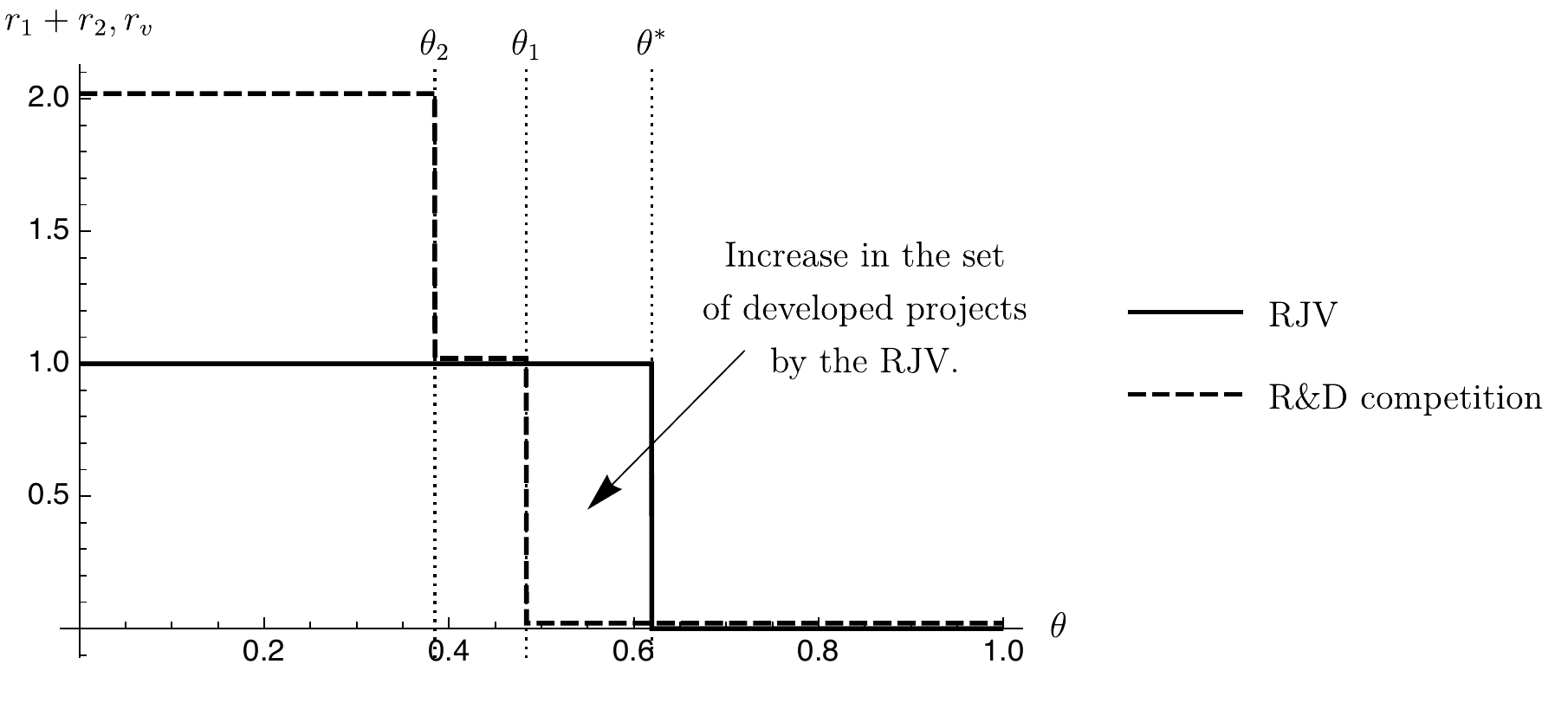}
        \caption{An example of an RJV increasing the set of developed projects.}  \label{fig:Merged}
\end{figure}

Result (iib) deserves particular emphasis. It is common in the innovation literature to use the overall amount of R\&D spending as a measure of the probability that an innovation will be discovered. Usually, a policy is said to promote innovation if it leads to more R\&D spending. The result demonstrates that this approach can be misleading: When competition is not soft and financial constraints are severe, R\&D competition leads to both higher R\&D costs and a lower innovation probability than an RJV. 
Intuitively, in any equilibrium under R\&D competition, both firms invest more than their available budget in R\&D. Therefore, the marginal R\&D project that they are willing to invest in has to be sufficiently profitable, so that incurring the higher marginal cost of borrowed funds is justified. However, whenever the conditions of Proposition \ref{prop_RJVEffect}(iia) are satisfied, an RJV optimally invests weakly less than its total budget. In spite of this reduction in R\&D costs, the probability of innovation increases as the reduction in investment corresponds to avoided duplication rather than reductions in project variety. By contrast, in case (i), we cannot rule out that the RJV spends more on R\&D than the firms under R\&D competition: While the RJV can achieve the same innovation probability as the competitive firms with lower costs, it also faces stronger investment incentives.

Our model assumes (conservatively) that the borrowing costs $\rho$ do not change if firms form an RJV. There are several reasons why an RJV should be able to borrow at a lower cost than an individual firm engaged in competition with a rival. For example, the fact that the probability of discovering the innovation goes up, or that the profit variance goes down, both suggest that an RJV should be able to borrow at a lower cost. If this were the case, the set of parameters for which an RJV increases innovation would become larger.

\subsection{Consumer Welfare}\label{consumer_welfare}

While supporting inventiveness of an industry can be a worthy goal in itself, competition policy often emphasizes consumer surplus. As our next result shows, the two are aligned under very mild conditions. 

The model we have introduced so far does not specify the impact of innovation on consumers at all. However, a natural intuition is that consumers benefit from innovations. The following assumption formalizes this intuition. Denote consumer surplus resulting from market competition when both firms have innovated with $CS_{II}$, when neither firm has innovated with $CS_{\mathit{0}\mathit{0}}$ and with $CS_{I\mathit{0}}$ when only one firm has innovated.

\begin{assumption}\label{ass_CS}
Consumers benefit from innovation: $CS_{II} >  CS_{\mathit{0}\mathit{0}}$ and $CS_{II} > CS_{I\mathit{0}}$.
\end{assumption}

When innovations are aimed at developing better products or lowering production costs, we can expect that some of the benefits will be passed on to consumers, so that Assumption \ref{ass_CS} will hold. With the addition of this assumption, we can show the following result.

\begin{proposition}[Effect of RJVs on consumer surplus]\label{Prop_CS} $ $ \newline
If an RJV strictly increases the innovation probability, then it also strictly increases expected consumer surplus.
\end{proposition}

Formation of an RJV affects consumer surplus in two ways: (i) it changes the probability that the innovation is discovered and (ii) it changes the diffusion of innovation among the competing firms. By Assumption \ref{ass_CS}, consumers benefit from both a higher innovation probability and more diffusion of innovation. Since RJVs always facilitate the diffusion of innovation (because whenever the RJV innovates, both firms can use the resulting innovation), then an RJV that increases the innovation probability clearly leads to higher consumer surplus. 

Together with Proposition \ref{prop_RJVEffect}, this result gives simple conditions for an RJV to increase expected consumer surplus. It should be noted that these conditions are sufficient, but not necessary -- because RJVs always increase the diffusion of innovation, it is possible for an RJV that slightly decreases the innovation probability to lead to a higher expected consumer surplus. 

\subsection{Profitability of RJVs}\label{Profitability}

So far, we have analyzed how the formation of an RJV affects
innovation probability and consumer welfare. 
We have not yet asked whether it is in the firms' interest to agree on an RJV. In the following, we will deal with this issue.
We ask under which conditions joint profits are higher under an RJV than under R\&D competition. If this requirement is not fulfilled, then at least one of the firms would not consent to an RJV. By contrast, if the RJV does increase joint profits and profits are symmetric under R\&D competition, then the RJV will result in a Pareto improvement from the perspective of the two firms.\footnote{A sufficient condition for equal profits under R\&D competition to emerge is that the firms coordinate on an equilibrium where they innovate with an equal probability and where their research costs are equal.} Even when profits are not symmetric ex ante, an RJV that increases joint profits could always be turned into a Pareto improvement using suitable transfers.

Then using Lemmas \ref{prop_equilibrium} and \ref{prop_RJVInvestment}, we
find that net profits with an RJV are at least as high as under competition
if and only if%
\begin{gather}
2\theta^{\ast}\pi_{II}  +2\left(  1-\theta^{\ast}\right)
\pi_{\mathit{0}\mathit{0}}  - \gamma^{rjv}  \geq
\label{WeakIC}\\
2\theta_{2}\pi_{II}  +\left(  \theta_{1}-\theta_{2}\right)
\left(  \pi_{I\mathit{0}}  +\pi_{\mathit{0}I}  \right)  +2\left(
1-\theta_{1}\right)  \pi_{\mathit{0}\mathit{0}} 
-2\gamma^{com}  \text{,}\nonumber
\end{gather}
where $\gamma^{rjv}$ and $\gamma^{com}$ capture total research cost (including the costs of external financing) incurred by the RJV and a single firm under competition, respectively.\footnote{Using Lemma \ref{prop_RJVInvestment}, $\theta^{\ast}$ can be expressed in terms of ($\theta^{B}$, $\theta^{u}$ and $\theta^{\rho}$), which, in turn, can be expressed in terms of fundamentals.}
In the following, we\ will shed more light on this condition by identifying
transparent (sufficient)\ conditions on primitives under which it holds. Define
\begin{align*}
    \Psi =
    \begin{cases}
    \dfrac{\pi_{I\mathit{0}}+\pi_{\mathit{0}I}-2\pi_{II}}{2(\pi_{II}-\pi_{\mathit{0}\mathit{0}})} & \text{for } \pi_{II} > \pi_{\mathit{0}\mathit{0}} \\
    \infty, & \text{for } \pi_{II} = \pi_{\mathit{0}\mathit{0}},
  \end{cases}
\end{align*}
and note that whenever competition is intense, $\Psi>0$.

\begin{proposition}[Profitable innovation-enhancing RJV]\label{PropICNew} $ $ 
\begin{enumerate}[label=(\roman*)]
\item If competition is soft, any RJV strictly increases net profits (as well as the innovation probability).
\item If competition is moderate, any RJV that increases the innovation probability also increases net profits. The converse statement does not hold.
\item If competition is intense and an RJV strictly increases the innovation probability, it increases net profits if, in addition, $\frac{\min\{\theta^{B},\theta^{u}\}-\theta_{1}}{\theta_{1}-\theta_{2}}>\Psi$.
\end{enumerate}
\end{proposition}

The distinction between the three cases reiterates the importance of the intensity of competition. In case (i), competition is soft, so that part (i)\ of Proposition \ref{prop_equilibrium} applies -- the RJV increases innovation. According to Proposition \ref{PropICNew}, in this case, the firms' incentives for RJV formation are fully aligned with the goal of increasing the innovation probability. 
Proposition \ref{PropICNew}(ii) shows that, in the part of the region with moderate competition where the RJV increases the innovation probability (see Proposition \ref{prop_equilibrium}(ii)), the RJV is incentive-compatible. However, with moderate competition, we cannot rule out the case that firms engage in RJVs even when they reduce the probability of innovation: In Proposition \ref{PropICBad} in the appendix, we show that this case arises close to soft competition region. In that proposition, we provide a condition where an RJV would reduce innovation slightly, but without major adverse effects on gross profits. The cost-reducing effect of an RJV will then suffice to make it profitable. 
This result shows that the concern in the European Union that RJVs may have an adverse effect on the probability of innovation may not be entirely unfounded.

Finally, Part (iii) applies when competition is intense. Contrary to soft and moderate competition, the conditions guaranteeing that an RJV increases innovation by Proposition \ref{prop_RJVEffect} do not guarantee incentive compatibility: The additional condition in Proposition \ref{PropICNew}(iii) limits the intensity of competition as captured by $\Psi$.\footnote{Note that $\Psi$ is high if the value of avoiding competition is high relative to the value of catching up.} For instance, it does \textit{not} hold with homogeneous Bertrand competition. It also requires that the budget of the RJV is sufficiently large.\footnote{At the boundary between the intense and moderate competition regime, $\Psi=0$. Thus, the second condition in (iii) reduces to $\theta^B>\theta_1$ and $\theta^u>\theta_1$, which is equivalent to the conditions $B>\bar{B}(\rho)$ and $\rho>\overline{\rho}$ in (ii).}

In most cases, the conditions in Proposition \ref{PropICNew} also guarantee that the RJV does not spend more than its total budget, so that, by Assumption 2, it does not increase total expenditures. An exception arises in the subcase of (i) where the budget is sufficiently low that $\theta^B<\theta^{\rho}$: In this case, the RJV may spend more ($\theta^{\ast}=\theta^\rho$) than the two firms would have spent under R\&D competition. Spending the same amount as before would have reduced costs without affecting innovation and thus would have already been profitable. The fact that the RJV chooses to spend more thus means that this is profitable, despite the increase in R\&D costs.

An immediate corollary of Propositions \ref{Prop_CS} and \ref{PropICNew} is that an RJV increases total welfare whenever conditions $(i)$--$(iii)$ of Proposition \ref{PropICNew} are satisfied. The reason for this is that such an RJV increases the innovation probability, so that, by Proposition \ref{Prop_CS}, it increases consumer welfare and by Proposition \ref{PropICNew}, it increases net profits, therefore increasing total welfare.

\subsection{Examples}\label{sec:example}

In this subsection, we illustrate the general analysis with two standard oligopoly models. For the first one, homogeneous linear Cournot competition, competition is moderate or intense, so that Proposition \ref{prop_RJVEffect}(ii) always applies and financial constraints are necessary for the innovation probability to be higher with an RJV than without.
In the second example, differentiated price competition, competition can be soft. When this is the case, Proposition \ref{prop_RJVEffect}(i) applies and the RJV always increases the innovation probability. In each case, we only sketch the analysis; more details are in Appendix \ref{SecAppExamples}. 

\subsubsection{Cournot Competition}\label{SecCournot}

Suppose that two firms are choosing quantities $q_1$ and $q_2$, with $Q=q_1+q_2$. Assuming an interior solution, the market price is given by $P(Q) = a - b Q$. Each firm can produce the good with some constant marginal cost $c$. The firms can invest in a potential process innovation that reduces the marginal cost of production to $c-I$ for some $I>0$. Denoting $\alpha=a-c$, we assume for simplicity that $\alpha>I$, which guarantees that innovations are non-drastic.
Calculating standard Cournot profits when firms have marginal costs $c$ or $c-I$ yields the reduced form profits  $\pi_{t_it_j}$ and it is straightforward to verify that they satisfy Assumption \ref{ass_r}.
In fact, the stricter condition that competition is not soft, as required by Proposition \ref{prop_RJVEffect}(ii), holds for all parameter values.

\begin{figure}[h]
 \centering
      \includegraphics [width=0.8\textwidth]{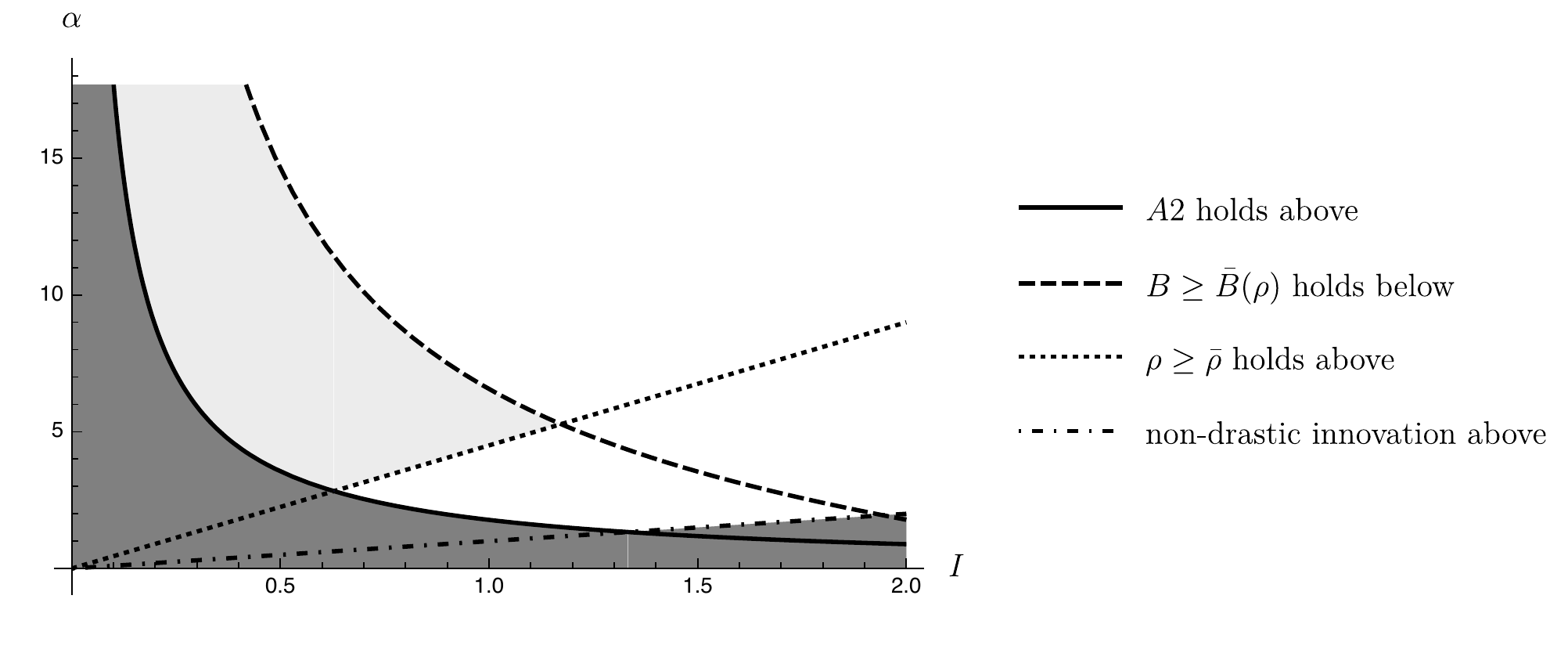}
   \caption{Comparison of R\&D competition and RJV in a Cournot example with inverse demand $P(Q)=a-bQ$, constant marginal costs $c$, $B=0.01$, $\rho=0.1$ and $C(\theta)= \frac{\theta}{1-\theta^2}$. Axes depict cost reduction $I$ and $\alpha = a-c$.}  \label{fig:CournotSpill}
\end{figure}

Thus, after calculating $\overline{\rho}=\frac{I}{2\alpha+I}$, we directly obtain:

\begin{corollary}\label{CorCournotNS}
In the linear Cournot model, the innovation
probability is strictly higher with an RJV than with R\&D competition if and only if $\rho > \bar{\rho}= \frac{I}
{2\alpha+I}$ and $B > \bar{B}(\rho)$. If these conditions both hold, then the
total R\&D expenditures of the RJV are lower than those under R\&D competition.
\end{corollary}

When competition is moderate or intense, the RJV only improves the innovation probability if the impact of pooling of resources is significant enough (as captured by the conditions on $\rho$ and B). 
Importantly, Corollary \ref{CorCournotNS} also identifies the role of the product market. A larger product market (captured by higher $\alpha$) and a smaller innovation size $I$ both increase the range of interest rates for which the RJV increases profits.

Figure \ref{fig:CournotSpill} illustrates the result for specific parameter values. Assumption \ref{ass_budget} and the focus on non-drastic innovations imply that we do not consider the darkly shaded region. 
The lightly shaded area depicts the parameter region for which the innovation probability is higher with an RJV than with R\&D competition. The existence of this region means that the requirement of Assumption \ref{ass_budget} that the budget is sufficiently small and the requirement from Proposition \ref{prop_RJVEffect}(ii) that it is sufficiently large are consistent. Note that all RJVs that increase the innovation probability compared to any equilibrium under R\&D competition are profitable in this case.\footnote{This is true because, for the given parameterization, competition is moderate. Hence, according to Proposition \ref{prop_RJVEffect}(ii), such innovation-enhancing RJVs must satisfy the conditions that are sufficient for a profitable RJV according to Proposition \ref{PropICNew}(ii).}
In the parameter region colored in white, an RJV lowers the innovation probability compared to any equilibrium under competition.

\subsubsection{Differentiated Price Competition}\label{SecPriceCompetition}

The linear homogeneous Cournot model is simple to analyze, but it restricts the possible outcomes, because competition is moderate or intense, so that Propositions \ref{prop_RJVEffect}(i) and \ref{PropICNew}(i) never apply. With differentiated goods, competition can be soft (as well as moderate or intense), so that these results become applicable. 
To see this, consider a standard model of differentiated price competition with inverse demand $p_i=1-q_i-bq_j$ for $b \in [0,1)$ and constant marginal cost $c>0$ where firms can engage in cost reductions $I \leq c$.
In the appendix, we derive the equilibrium profits.

\begin{figure}[h]
 \centering
     \includegraphics[width=0.8\textwidth]{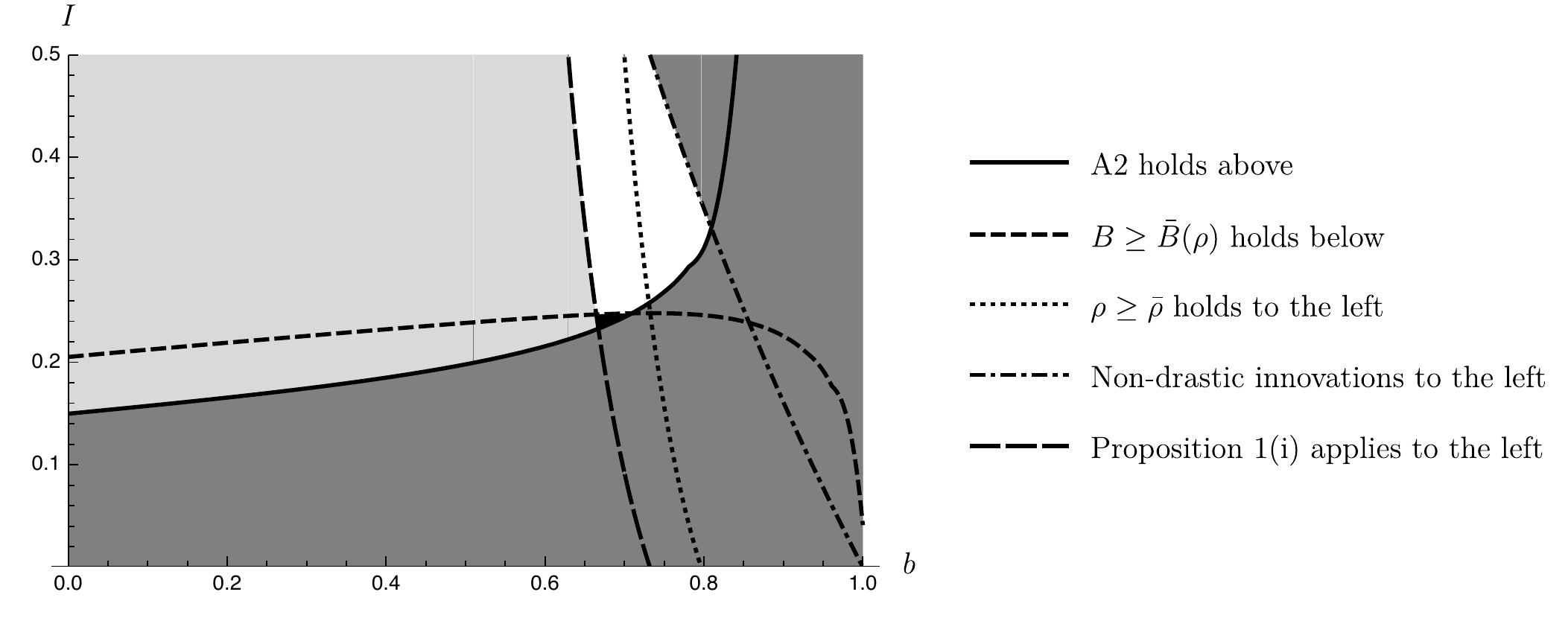}
    \caption{Comparison of R\&D competition and RJV in a differentiated Bertrand example with inverse demands $p_i=1-q_i-bq_j$ and constant marginal costs $c=0.5$. Axes depict substitution parameter $b$ and cost reduction $I$.}  \label{fig:Bertrand}
\end{figure}

Figure \ref{fig:Bertrand} illustrates the stark contrast to the homogeneous Cournot example. We exclude parameter areas where the innovation is drastic (in which case $\pi_{\mathit{0}I}$ would be negative)
and/or Assumption \ref{ass_budget} is violated (darkly shaded region). The central observation is that the RJV increases innovation and is profitable with sufficiently weak competition (in the large shaded grey area). This follows by applying Propositions \ref{prop_RJVEffect}(i) and  \ref{PropICNew}(i). By contrast, the parameter region where Proposition \ref{prop_RJVEffect}(ii) and  \ref{PropICNew}(ii) apply is very small (only the very small black area in the middle of the figure).  
Finally, note that, by Proposition \ref{PropICBad}, it will be profitable to engage in RJVs that reduce innovation (and costs) for parameter constellations near the left boundary of the white region.\footnote{Close to the left boundary of the white region, condition (i) of Proposition \ref{PropICBad} holds. Moreover, $B<\bar{B}(\rho)$, so that (ii) holds.}

\section{Further Results}\label{SecExtensions}

In this section, we provide further results. We first compare the effects of RJVs with those of mergers. Then we allow for spillovers and licensing, respectively. Finally, we consider markets with more than two firms.

\subsection{Mergers vs. RJVs}\label{SecMergers}

Competition policy usually views RJVs more favorably than full mergers as they allow the participants to reap some of the efficiency benefits that might arise in R\&D, without necessarily eliminating product market competition between the firms involved.\footnote{Note, however the empirical work suggesting that RJVs may foster collusion (\cite{duso2014collusion} and \cite{helland2019research}).} However, a precise comparison needs to take differences in the effects of RJVs and mergers on innovation into account. In the following, we therefore analyze the innovation effects of a merger between the two firms, following the above analysis of the RJV closely. Contrary to the RJV, the merged entity not only combines the research budget, but its constituent parts give up competition entirely. We denote the (monopoly) profit of the merged firm as $\pi_{t_m}$, where $t_m \in \{0,I\}$ indicates whether the firm has successfully innovated or not. In line with Assumption \ref{ass_r}(ii), we assume that innovation increases profits. 
\begin{assumption}\label{ass_merger}
$\pi_I >\pi_{\mathit{0}}$.
\end{assumption}

The analysis for the merged firm is entirely analogous to the RJV case, except that we have to replace $\pi_{II}$ with $\pi_I$ and $\pi_{\mathit{0}\mathit{0}}$ with $\pi_{\mathit{0}}$ in the expected payoff formula \eqref{eq_rjv_profit}.

Accordingly, we define critical values $\theta^u_m$ and $\theta^\rho_m<\theta^u_m$ which are analogous to $\theta^u$ and $\theta^\rho$, except that we replace $2(\pi_{II} -\pi_{\mathit{0}\mathit{0}})$ with $\pi_I -\pi_{\mathit{0}}$. It is straightforward that the merged firm optimally uses a single cut-off strategy like the  RJV, with $\theta^u_m$ and $\theta^\rho_m$ instead of $\theta^u$ and $\theta^\rho$ (see Lemma \ref{prop_mergerInvestment} in Appendix \ref{SecAppMergerPortfolio}). As a result, the comparison between investments with a merger and with R\&D competition (see Proposition \ref{prop_CompMergers} in Appendix \ref{SecAppMergerPortfolio}) is analogous to the comparison between the RJV and R\&D competition (Proposition \ref{prop_RJVEffect}), except that we again need to replace $2(\pi_{II} -\pi_{\mathit{0}\mathit{0}})$ with $\pi_I -\pi_{\mathit{0}}$, and the interest rate threshold thus becomes  
\begin{align*}
    \bar{\rho}_m & = \dfrac{\pi_{I\mathit{0}}-\pi_{\mathit{0}\mathit{0}} -(\pi_I-\pi_{\mathit{0}})}
{\pi_I-\pi_{\mathit{0}}}.
\end{align*}

The following result compares the innovation probability under a merger and under an RJV.

\begin{proposition}[Comparison of an RJV and a merger]\label{prop_comparison_m-rjv} $ $
\begin{enumerate}[label=(\roman*)]
\item If $2(\pi_{II}-\pi_{\mathit{0}\mathit{0}}) \geq \pi_I - \pi_{\mathit{0}}$, the innovation probability under an RJV is weakly higher than under a merger. The difference is strict, except when $\theta^{B}\in [\theta^\rho, \theta^u_m]$ or \newline$2(\pi_{II}-\pi_{\mathit{0}\mathit{0}})=\pi_I - \pi_{\mathit{0}}$.
\item If $2(\pi_{II}-\pi_{\mathit{0}\mathit{0}})<\pi_I - \pi_{\mathit{0}}$, the innovation probability under an RJV is weakly lower than under a merger. The difference is strict, except when $\theta^{B}\in [\theta^\rho_m, \theta^u]$.
\end{enumerate}
\end{proposition}

A merger leads to similar efficiency gains as an RJV -- in both cases duplicate projects are eliminated, and those resources can be invested into new projects. However, the total profit increase for the members of the RJV will generally differ from those for the merged firm: Whereas innovation increases the joint profit of the RJV by  $2(\pi_{II}-\pi_{\mathit{0}\mathit{0}})$, the corresponding value for the merged firm is  $\pi_I - \pi_{\mathit{0}}$. The above result confirms the intuition that the relative size of these two profit differentials determines whether an RJV or a merged firm will be more likely to generate innovation.

However, there is a subtle effect of financial constraints: Even when the total profit effects of innovation differ for RJVs and mergers, the investments and thus the innovation probability are the same for non-degenerate parameter ranges. This happens when the budgets are intermediate, that is either $\theta^{B}\in[\theta^\rho, \theta^u_m]$ in case $(i)$ or $\theta^{B}\in [\theta^\rho_m, \theta^u]$ in case $(ii)$.  In those cases, both the RJV and the merged firm invest their entire budgets, but the marginal return of additional research projects is not sufficient to justify the cost of borrowing from the capital market. Hence, the RJV and the merged entity each invest exactly the total research budget $2B$ into R\&D.

Proposition \ref{prop_comparison_m-rjv} enables us to analyze whether a merger or an RJV would be better from the consumer surplus perspective, assuming that firms would want to engage in it.\footnote{This will, for instance, be the case if Propositions \ref{PropICNew} or \ref{PropICBad} apply.} To this end, introduce a further weak assumption:

\begin{assumption}\label{ass_CS_Number}
   When technology is the same, consumer surplus is higher with two active firms than with one. 
\end{assumption}

We can use this Assumption together with Assumption \ref{ass_CS} to compare consumer surplus with an RJV and a merger.

\begin{corollary}\label{CorRJVMerger}
    Suppose that Assumptions \ref{ass_CS} and \ref{ass_CS_Number} hold. \newline
    (i) If (a) $2(\pi_{II}-\pi_{\mathit{0}\mathit{0}}) \geq \pi_I - \pi_{\mathit{0}}$ or (b) $2(\pi_{II}-\pi_{\mathit{0}\mathit{0}}) < \pi_I - \pi_{\mathit{0}}$ and $\theta^{B}\in[\theta^\rho_m, \theta^u]$, then the consumer surplus is higher with an RJV than with a merger.\newline
    (ii) If $2(\pi_{II}-\pi_{\mathit{0}\mathit{0}}) < \pi_I - \pi_{\mathit{0}}$ and $\theta^{B}\not\in[\theta^\rho_m, \theta^u]$, then the consumer surplus can be higher or lower with a merger than with an RJV.   
\end{corollary}
Part (i)(a) of Corollary \ref{CorRJVMerger} shows that, in case $(i)$ of Proposition \ref{prop_comparison_m-rjv}, the RJV unambiguously increases consumer surplus. The reason is that the RJV both weakly increases the probability that the innovation will be discovered and increases competition for any level of technology. Even in case $(ii)$ of  of Proposition \ref{prop_comparison_m-rjv}, where the profit increase from innovation is larger for the merger than for the RJV, the RJV unambiguously leads to higher consumer surplus than the merger if $\theta^{B}\in[\theta^\rho_m, \theta^u]$, as the innovation probability is the same for both forms of cooperation. However, if $\theta^{B} \not\in[\theta^\rho_m, \theta^u]$, the comparison is ambiguous: The merged firm would be more likely to discover the innovation, while the RJV would maintain the more competitive market structure. 
One consequence of this analysis is that from a consumer perspective, an RJV is preferable to a merger, except possibly when the innovation probability would be significantly higher under a merger. This suggests that firms should not only be required to show that a merger would have positive innovation effects, but also that these effects would not occur with an RJV.

\subsection{Spillovers}\label{Spillovers}
Our model differs from the previous literature on RJVs not only by its focus on financial constraints as opposed to spillovers, but also by the feature that firms can choose between different R\&D projects. To simplify the comparison with the existing literature, we first consider a variant of our project choice model without financial constraints, but with spillovers. Thereafter, we analyze the interaction between financial constraints and spillovers.
 
\subsubsection{Spillovers without financial constraints}\label{SpilloversNoFC}

We modify the setting of Section \ref{SecModel} by assuming that the firms with cost functions $C(\theta)$ choose their investment portfolio without any budget constraint. Moreover, with R\&D competition, if a firm has invested successfully in a project and the rival has not, then with probability $\sigma\in[0,1]$ the rival will obtain access to the innovation. Thus, it is now possible that a firm obtains the innovation without investing itself.

We provide the equilibrium characterization for R\&D competition in Appendix \ref{SecAppSpillNFC}. As in the benchmark model, we obtain an equilibrium in double cut-off strategies. A full description of the equilibrium is given in Lemma \ref{prop_equilibrium_spill}. The analysis with RJVs is simpler than in the case with financial constraints. The increase in joint profit from a successful innovation is $2\pi_{II}-2\pi_{\mathit{0}\mathit{0}}$. Hence, the RJV invests in all projects up to a cut-off value, which is given by
$\theta^{u}$, 
and it does not invest in the remaining ones. The following result compares investments in the RJV with those under R\&D competition.

\begin{proposition}\label{PropCompSpill}
Consider the model with spillovers, but without financial constraints. Assume that $\pi_{I\mathit{0}}>\pi_{II}$. Then the innovation probability is strictly larger under the RJV than under R\&D competition if and only if \[
\sigma > 1 - \frac{\pi_{II}-\pi_{\mathit{0}\mathit{0}}}{\pi_{I\mathit{0}}-\pi_{II}}.
\]
This condition is always satisfied if competition is soft.
\end{proposition}

The proof is in Appendix \ref{SecAppSpillNFC}. As in the case with financial constraints, for RJVs to generate a higher innovation probability than R\&D competition, it is crucial that the value of escaping competition is sufficiently small relative to the value of joint innovation. A simple, but important implication of Proposition \ref{PropCompSpill} needs to be emphasized: When competition is soft, then the RHS of the inequality in Proposition \ref{PropCompSpill} is negative and an RJV increases innovation for any level of spillovers (including $\sigma =0)$. When competition is moderate or intense, the RHS is positive, but an RJV can still increase the innovation probability if the spillovers are strong enough relative to the strength of the competition. The exception is (homogeneous) Bertrand competition, which is so intense that $\pi_{II}=\pi_{\mathit{0}\mathit{0}}=0$, so that the inequality cannot be satisfied for any $\sigma \in [0,1]$.

\subsubsection{Spillovers with financial constraints}\label{SpilloversFC}

In Appendix \ref{SecAppSpillFC}, we integrate the model with spillovers just discussed into the model with financial constraints.
Large parts of the analysis follow directly from our results in Section \ref{SecEffects}. To apply those results, one needs to define the \emph{expected} payoffs $\tilde{\pi}_{t_it_j}$ of discovering the innovation (i.e., before any spillovers happen and taking into account the possibility of a spillover) and then observe that Assumption \ref{ass_r} holds with $\pi$ replaced by $ \tilde{\pi}$. Then, the results of Section \ref{SecEffects} apply after replacing realized product market profits with expected payoffs.
Adapting Assumption \ref{ass_budget}, we assume that the research budgets of the individual firms are sufficiently small that they will borrow positive amounts in any equilibrium.
It is straightforward to show that there is an equilibrium in double cut-off strategies under R\&D competition (see Lemma \ref{prop_equilibriumSpillFC} in Appendix \ref{SecAppSpillFC} for details).
The comparison between R\&D competition and RJV is also very similar to the case without spillovers (see Proposition \ref{prop_RJVEffect_spill} in Appendix \ref{SecAppSpillFC}): When the total profit increase $2\pi_{II}-2\pi_{\mathit{0}\mathit{0}}$ from innovation is high enough, then the RJV will lead to a greater innovation probability than R\&D competition independent of financial constraints. If the total profit increase from innovation is lower, the RJV only leads to a greater innovation probability if both the interest rate $\rho$ and the RJV budget $2B$ are above a threshold; in this case, the RJV saves investment costs by avoiding duplication.  

The following differences to the benchmark model are relevant for the comparison between investments under R\&D competition and under the RJV. 
First, RJVs unconditionally increase innovation whenever $2\pi_{II}-2\pi_{\mathit{0}\mathit{0}}>\pi_{I\mathit{0}}-\pi_{\mathit{0}\mathit{0}}-\sigma(\pi_{I\mathit{0}}-\pi_{II})$, which is more likely to be satisfied when spillovers are strong (i.e., when $\sigma$ is high).
Second, when that condition is not satisfied, an increase in $\sigma$  lowers the thresholds for the budget and the interest rate which are needed to guarantee that the RJV increases the innovation probability.
The conditions under which an RJV increases the innovation probability are thus weaker with higher spillovers, just as they are with higher interest rates:

\begin{proposition}[Benefit of RJV increases in the spillover rate] \label{prop_RJV_better} $ $
Fix any $\sigma$ and $\rho$. If the innovation probability is strictly larger under the RJV than under R\&D competition, then it is also strictly larger for any $\sigma' \geq \sigma$ and $\rho' \geq \rho$.
\end{proposition}

As in the case without spillovers, an RJV results in efficiency gains at the investment stage by reducing duplication, and resources can be invested in a larger set of projects. Moreover, whereas spillover effects reduce investment under competitive R\&D, this is not the case with an RJV. Thus, the positive effect of R\&D cooperation on the innovation probability must be larger with spillovers than without, reflecting the internalization of positive spillovers by the RJV.

\subsection{Licensing}\label{sec_license}

Like RJVs, licensing agreements are an instrument for firms to share the fruits of innovation. The literature has demonstrated the possible benefits and costs of such agreements when R\&D efforts are one-dimensional. Here, we show how the possibility of licensing influences R\&D project choice in the absence of an RJV and, thereby, the effects of switching to an RJV. In particular, we will show that even when licensing of innovations is possible, RJVs can still lead to an increase in the innovation probability.

We thus extend our benchmark model to allow for licensing of innovations.\footnote{In Appendix \ref{SecAppLicense}, we describe the details of the model. Here, we sketch the main ideas.} 
We suppose that, if only one firm has innovated successfully, it can license the innovation to the competitor with a two-part tariff $(L,\eta)$, consisting of an output-independent fixed fee $L$ and a variable, output-dependent part $\eta$ (e.g., royalties).\footnote{As will become clear later, if only simpler licensing contracts were available, our analysis would still apply. See \cite{shapiro1985} for a discussion of licensing with and without royalties. \cite{fauli2003merge} analyze licensing with fixed fee, royalty and two-part tariff contracts as an alternative to mergers.} 
When the unsuccessful firm licenses the innovation, both the innovator and the licensee have the technology state $t_i=I$. However, the incentives of the licensee to compete vigorously are dampened by the variable part of the licensing contract $\eta$.\footnote{For example, royalties increase the licensee's marginal cost and, thus, soften competition. This leads to asymmetric product market competition, although the firms use equal technology.} This reduction of the intensity of competition increases total industry profits (compared to the situation when both firms independently innovate) by some amount $\Delta \geq 0$.

We assume that the innovator makes a take it or leave it offer, extracting all the rents from the licensee.
In particular, the innovator sets the fixed fee $L$ such that the unsuccessful firm earns its outside option $\pi_{\mathit{0}I}$ and, thus, is indifferent between accepting the contract or not. Therefore, the innovator is willing to license the innovation if her profits with licensing, $2\pi_{II}+\Delta-\pi_{\mathit{0}I}$, are at least as high as her profits without, $\pi_{I\mathit{0}}$.
Licensing always happens if competition is soft or moderate and sometimes when it is intense.\footnote{This is related to the result of \cite{katzshapiro1985} that, in a Cournot setting, a successful innovator will license small innovations, but not large or drastic innovations.}  
As in the analysis of spillovers in Section \ref{SpilloversFC}, after replacing the function $\pi_{t_it_j}$ appropriately, the analysis directly follows Section \ref{SecEffects}. Specifically, we define a function $\pi^L$ on $\{0,I\} \times \{0,I\}$, which is identical with $\pi$ except that it takes into account licensing payments when only one firm is successful.
The only difference between $\pi^L$ and $\pi$
is that $\pi^L_{I\mathit{0}} = \max \lbrace\pi_{I\mathit{0}},2\pi_{II}+\Delta-\pi_{\mathit{0}I}\rbrace$. This function captures profits as a function of technology level, but taking into account possible gains from licensing. Using this modified profit function, we derive thresholds $\theta^L_1$ and $\theta^L_2$ by replacing $\pi$ with $\pi^L$ in the definitions of $\theta_1$ and $\theta_2$. Crucially, whereas $\theta^L_2 = \theta_2$, $ \theta^L_1 \geq \theta_1$, reflecting the potential gains from licensing.

When $2\pi_{II}+\Delta-\pi_{\mathit{0}I}<\pi_{I\mathit{0}}$, the equilibrium under R\&D competition is exactly the same as in Lemma \ref{prop_equilibrium}, because licensing never occurs in this case. When $2\pi_{II}+\Delta-\pi_{\mathit{0}I}\geq\pi_{I\mathit{0}}$, licensing increases the innovation probability in any equilibrium to $\theta_1^L \geq \theta_1$, as the opportunity to license increases the incentives to explore further projects.

For the comparison with the RJV, we replace the budget threshold  $\bar{B}(\rho)$ and the interest threshold  $\bar{\rho}$ with thresholds  $\bar{B}^L(\rho)$ and $\bar{\rho}^L$ that are based on $\pi^L$ rather than $\pi$, leading to the following modification of Proposition \ref{prop_RJVEffect}.

\begin{proposition}[Comparison of R\&D competition with licensing and RJV]\label{prop_RJVEffect_licensing} $ $
\begin{enumerate}[label=(\roman*)]
\item Suppose $2\pi_{II}+\Delta-\pi_{\mathit{0}I} \geq \pi_{I\mathit{0}}$. Then: \newline (a) The innovation probability is strictly larger under the RJV than under competition if and only if $B > \bar{B}^L(\rho)$ and $\rho > \bar{\rho}^L$.\newline
(b)  If the formation of the RJV strictly increases the innovation probability, then it  weakly decreases total R\&D spending.
\item Suppose $2\pi_{II}+\Delta-\pi_{\mathit{0}I} < \pi_{I\mathit{0}}$. Then the effect of an RJV on the innovation probability is the same as in the absence of a licensing possibility.
\end{enumerate}
\end{proposition}

In case $(i)$, firms want to license the innovation. In case $(ii)$, they do not. Importantly, the conditions under which the RJV leads to a higher innovation probability are more rigid than without licensing. This is obvious in the case of soft competition, in which Proposition \ref{prop_RJVEffect}(i) states that an RJV is \textit{always} preferable to R\&D competition, while Proposition \ref{prop_RJVEffect_licensing}(i) requires that $B > \bar{B}^L(\rho)$ and $\rho > \bar{\rho}^L$. When competition is not soft, the conditions under which an RJV increases the innovation probability are also more restrictive with licensing than without, since $\bar{B}^L(\rho)\geq \bar{B}(\rho)$ and $\bar{\rho}^L \geq \bar{\rho}$ whenever Proposition \ref{prop_RJVEffect_licensing}(i) applies. The difference arises because licensing increases innovation incentives under R\&D competition, so that there is less to gain from an RJV. Moreover, an RJV that increases the innovation probability weakly decreases total R\&D spending, because it invests weakly less than the available budget while both firms invest strictly more than their budget under R\&D competition. 

To put the results into perspective, we can think of ex-post licensing and RJVs as imperfect substitutes for sharing the fruits of R\&D. Nonetheless, the above results show that even when ex-post licensing is possible, an RJV may still lead to a higher innovation probability than R\&D competition if financial constraints are sufficiently tight.

\subsection{Multiple firms}
We extend our model by allowing for more than two competing firms. With multiple firms, there are many conceivable ways in which RJVs could be formed, including industry-wide RJVs as well as several competing RJVs. We analyze two illustrative cases. First, we consider a market with three firms that can form an industry-wide RJV. Second, we consider the case of four firms that form two competing RJVs.
The analysis is very similar to the benchmark model with two firms. Therefore, we defer details to the Appendix \ref{SecAppMultiplefirms}.

\subsubsection{Industry-wide RJV}
We extend the analysis to the case of three firms, which can form one RJV. Suitably adjusting Assumptions \ref{ass_r} and \ref{ass_budget}, the analysis and results are analogous to the benchmark model with two firms. 
The only notable difference is that the R\&D competition game now has multiple equilibria in triple cut-off strategies characterized by the three critical values $\theta_3\leq\theta_2\leq\theta_1$.\footnote{All three firms invest below $\theta_3$, two between $\theta_3$ and $\theta_2$, one between $\theta_2$ and $\theta_1$, and none above $\theta_1$.} However, the innovation probability in any equilibrium is still given by $\theta_1$, the most expensive project in which a single firm can profitably invest relying on external resources. The analysis of the RJV when all firms participate and the resulting comparison between R\&D competition and cooperation is qualitatively unchanged. Therefore, we find similar results to Proposition \ref{prop_RJVEffect}: When competition is not too intense, the innovation probability is higher in the RJV; otherwise, this conclusion requires the budget and the external financing costs to be high enough. In the latter case, total R\&D-spending in the RJV is lower than under competition.

\subsubsection{Multiple RJVs}
Next, we consider the formation of multiple RJVs. We consider a market with four firms that form two symmetric RJVs, each with two firms. Therefore, R\&D cooperation does not eliminate competition in the innovation stage entirely, but reduces the number of competing agents. Hence, even with an RJV cheap projects are still duplicated. We assume that the budget of an RJV is sufficiently large that it never borrows in equilibrium. Otherwise, the analysis of two competing RJVs turns out to be similar to the R\&D competition regime in the baseline model. Analogously to Proposition \ref{prop_RJVEffect}, we find: When competition is relatively soft, then the innovation probability is higher with two RJVs than with R\&D competition without additional conditions. 
Under relatively moderate or intense competition, cooperation on R\&D increases the innovation probability only if the budget and the interest rate are sufficiently high. In this case, total R\&D-spending with two RJVs is lower than when four firms invest individually.

\section{Conclusion}\label{SecConclusion}

This paper provides a novel theory of RJVs for financially constrained firms who can choose the set of research projects that they will pursue. RJVs allow firms to share their R\&D budget and to coordinate their R\&D investment decisions, while maintaining product market competition.

We find that, if product market competition is sufficiently soft, the RJV will increase the probability of an innovation even when there are no financial constraints. As product market competition increases, a positive innovation effect of the RJV requires that the external funding conditions are sufficiently bad and the budget of the RJV is sufficiently large. In the latter case, the RJV reduces research costs by avoiding duplication -- this shows that the relation between R\&D spending and R\&D success probability need not be positive. Moreover, any RJV that increases the innovation probability also increases expected consumer welfare.

Importantly, the conditions under which the RJV increases the probability of a successful innovation and the conditions under which it is profitable for the participants often coincide; in particular, for soft or intermediate competition, firms always want to form RJVs if they increase the innovation probability. This increases consumer welfare under mild conditions. Nonetheless, we also identify situations under which firms find it profitable to form an innovation-reducing RJV merely because they can coordinate on reducing R\&D costs, which is in line with concerns of policy makers.

We obtain qualitatively similar results on the effects of mergers on innovation. More interestingly, we find conditions under which a merger does not lead to a higher innovation probability than an RJV. In such situations, even if the merger has pro-competitive effects on innovation relative to the benchmark of R\&D competition, the merger should be prohibited because, contrary to the alternative of an RJV, it results in an adverse effect on product market competition.

\newpage
\appendix

\section{Proofs}

\subsection{Proof of Lemma \ref{prop_equilibrium}}

We will first prove an intermediate result.

\begin{lemma}\label{dominated_r}
Any strategy $r_i$ such that $\int_0^1r_i(\theta) C(\theta)d \theta \leq B$ is dominated.
\end{lemma}

\begin{proof}
If $\int_0^1r_i(\theta) C(\theta)d \theta \leq B$, then by Assumption \ref{ass_budget} there exists a set $\Theta' \subseteq [0,\theta_2)$ of positive measure, such that $r_i(\theta)=0$ for all $\theta \in \Theta'$. Consider a strategy $r'_i$, where $r'_i(\theta)=1$ for all $\theta \in [0, \theta_2)$ and $r'_i(\theta)=r_i(\theta)$ otherwise. We will show that $\mathbb{E}\Pi_i(r_i',r_{j})>\mathbb{E}\Pi_i(r_i, r_{j})$ for any strategy of the opponent $r_{j}$.

Noting that the strategy $r_i'$ requires external financing (while $r_i$ does not), and taking into account that
$r'_i(\theta)=r_i(\theta)$ for all $\theta > \theta_2$ then 
\begin{align*}
\mathbb{E}\Pi_i(r_i', r_{j}) - \mathbb{E}\Pi_i(r_i, r_{j}) = \\
\int_0^{\theta_2}  (r_i'(\theta)-r_i(\theta))\Big[(1-r_{j}(\theta))\left( \pi_{I\mathit{0}} -\pi_{\mathit{0}\mathit{0}}\right) + r_{j}(\theta)\left(\pi_{II} -\pi_{\mathit{0}I}\right) \Big] d\theta\\
-(1+\rho)\int_0^1 r_i'(\theta)C(\theta) d\theta + \rho B +\int_0^1  r_i(\theta)C(\theta)d \theta \\
\geq \int_0^{\theta_2}  (r_i'(\theta)-r_i(\theta))\Big[(1-r_{j}(\theta))\left( \pi_{I\mathit{0}} -\pi_{\mathit{0}\mathit{0}}\right) + r_{j}(\theta)\left(\pi_{II} -\pi_{\mathit{0}I}\right) \Big] d\theta\\
-(1+\rho)\int_0^1 r_i'(\theta)C(\theta) d\theta + (1+ \rho)\int_0^1  r_i(\theta)C(\theta)d \theta  \\ 
= \int_0^{\theta_2}  (r_i'(\theta)-r_i(\theta))\Big[(1-r_{j}(\theta))\left( \pi_{I\mathit{0}} -\pi_{\mathit{0}\mathit{0}}\right) + r_{j}(\theta)\left(\pi_{II} -\pi_{\mathit{0}I}\right) \Big] d\theta\\
-(1+\rho)\int_0^{\theta_2} (r_i'(\theta)-r_i(\theta))C(\theta) d\theta  \\ 
\geq \int_0^{\theta_2}  (r_i'(\theta)-r_i(\theta))\Big[\left(\pi_{II} -\pi_{\mathit{0}I}\right) -(1+\rho)C(\theta)  \Big] d\theta\\
> 0.
\end{align*}
The first inequality follows from the assumption that $\int_0^1r_i(\theta) C(\theta)d \theta \leq B$, the second from the fact that $\pi_{I\mathit{0}} -\pi_{\mathit{0}\mathit{0}} \geq  \pi_{II} -\pi_{\mathit{0}I}$ and $r_i'(\theta)-r_i(\theta) \geq 0$ and the last from the fact that $\pi_{II}-\pi_{\mathit{0}I} > (1+\rho)C(\theta)$ for all $\theta < \theta_2$ and $r_i'(\theta)>r_i(\theta)$ on the set of positive measure $\Theta'$.
\end{proof}

\noindent \emph{Proof of (i):} Take any strategy $r_{j}$ which corresponds to one of the equilibrium strategies given in Lemma \ref{prop_equilibrium}. Note that for any fixed $r_{j}$, the equilibrium candidate strategy $r_i$ is uniquely determined. Suppose $(r_i, r_{j})$ does not constitute an equilibrium. Then, there exists a strategy $r'_i$ such that $\mathbb{E}\Pi_i(r_i', r_{j})>\mathbb{E}\Pi_i(r_i, r_{j})$. By Assumption \ref{ass_budget}, all equilibrium candidates satisfy $\int_0^1r_i(\theta) C(\theta)d \theta > B$. Moreover,
by Lemma \ref{dominated_r} we can focus on strategies such that $\int_0^1r'_i(\theta) C(\theta)d \theta > B$ is satisfied. 

Denote the expected total payoff of project $\theta$, conditional on it being correct, as $v_i(\theta,r_i,r_{j})$. Then there exists a set $\Theta' \subseteq [0,1)$ with positive measure such that $v_i(\theta,r'_i,r_{j}) > v_i(\theta,r_i,r_{j})$ for all $\theta\in \Theta'$, or more explicitly: 
\begin{align}\label{proof_1_orig_ineqality}
&(1-r_{j}(\theta))[r'_i(\theta)\pi_{I\mathit{0}}+(1-r'_i(\theta))\pi_{\mathit{0}\mathit{0}}] +r_{j}(\theta)[r'_i(\theta)\pi_{II}+(1-r'_i(\theta))\pi_{\mathit{0}I}]  \nonumber\\ &-(1+\rho)C(\theta) r'_i(\theta) > (1-r_{j}(\theta))[r_i(\theta)\pi_{I\mathit{0}}+(1-r_i(\theta))\pi_{\mathit{0}\mathit{0}}]  \nonumber\\ & +r_{j}(\theta)[r_i(\theta)\pi_{II}+(1-r_i(\theta))\pi_{\mathit{0}I}]-(1+\rho)C(\theta)  r_i(\theta).
\end{align}
If $\theta<\theta_2$ then $r_{j}(\theta)=1$ so this inequality simplifies to
\begin{equation}\label{proof_1a}
    \begin{split}
 r'_{i}(\theta)(\pi_{II} -\pi_{\mathit{0}I}-(1+\rho)C(\theta)) > r_i(\theta)(\pi_{II} -\pi_{\mathit{0}I}-(1+\rho)C(\theta))  
    \end{split}
\end{equation} 
Since for $\theta<\theta_2$ we have $\pi_{II} -\pi_{\mathit{0}I}-(1+\rho)C(\theta)>0$ and $r_i(\theta)=1$, this would imply $r'_{i}(\theta)  > 1 $ which is a contradiction.\newline

If $\theta>\theta_1$ then $r_{j}(\theta)=0$ so inequality \eqref{proof_1_orig_ineqality} simplifies to
\begin{align}\label{proof_2a}
 r'_i(\theta)[\pi_{I\mathit{0}}-\pi_{\mathit{0}\mathit{0}} -(1+\rho)C(\theta)]  >   r_i(\theta)[\pi_{I\mathit{0}}-\pi_{\mathit{0}\mathit{0}} -(1+\rho)C(\theta)].
\end{align}
Since for $\theta>\theta_1$ we have $\pi_{I\mathit{0}}-\pi_{\mathit{0}\mathit{0}} -(1+\rho)C(\theta)<0$ and $r_i(\theta)=0$ this would imply $r'_{i}(\theta)  < 0 $ which is a contradiction. 

Next, consider $\theta \in (\theta_2 , \theta_1)$. This case only arises if $\theta_2 < \theta_1$, which immediately implies $\pi_{I\mathit{0}}+\pi_{\mathit{0}I}-\pi_{II}-\pi_{\mathit{0}\mathit{0}}>0$. If $r_{j}(\theta)=1$ then, as before, inequality \eqref{proof_1_orig_ineqality} simplifies to \eqref{proof_1a}. However, now $\pi_{II}-\pi_{\mathit{0}I}-(1+\rho)C(\theta)<0$ and, for the candidate equilibrium, $r_i(\theta)=0$. \eqref{proof_1a} would thus require that $r'_{i}(\theta)  < 0 $, which is a contradiction. Similarly if $r_{j}(\theta)=0$ the inequality \eqref{proof_1_orig_ineqality} simplifies to \eqref{proof_2a}, but $\theta<\theta_1$ implies   $\pi_{I\mathit{0}}-\pi_{\mathit{0}\mathit{0}} -(1+\rho)C(\theta)>0$ and, for the candidate equilibrium, $r_i(\theta)=1$. \eqref{proof_2a} would thus require that  $r'_{i}(\theta) > 1 $, which is a contradiction. 

\vspace*{5mm}
\noindent \emph{Proof of (ii):} Suppose there exist two strategies, $r_i$ and $r_{j}$, which constitute an equilibrium, and a set of positive measure $I \subseteq [0,1)$, such that  $r_i$ is different from the strategies characterized in the Lemma at all points of the set $I$. By Lemma \ref{dominated_r} we can focus on strategies such that the budget is binding. 
Let $I_1 = I \cap (0,\theta_2)$, $I_2 = I \cap (\theta_2, \theta_1)$ and $I_3 = I \cap (\theta_1, 1)$. Note that at least one of the sets  $I_1$, $I_2$, or $I_3$ has positive measure.

Define
\begin{align*}
\Gamma_i(\theta,r_{j})= &\pi_{I\mathit{0}}-\pi_{\mathit{0}\mathit{0}}-(1+\rho)C(\theta)\\
&-r_{j}(\theta)(\pi_{I\mathit{0}}+\pi_{\mathit{0}I}-\pi_{II}-\pi_{\mathit{0}\mathit{0}}).
\end{align*}
We can express $v_i(\theta,r_i,r_{j})$, the expected total payoff of project $\theta$, conditional on it being correct, as 
\begin{align*}
v_i(\theta,r_i,r_{j})= r_i(\theta)\Gamma_i(\theta,r_{j})+(1-r_{j}(\theta)) \pi_{\mathit{0}\mathit{0}}+r_{j}(\theta)\pi_{\mathit{0}I}.
\end{align*}
Since $\pi_{I\mathit{0}}+\pi_{\mathit{0}I}-\pi_{II}-\pi_{\mathit{0}\mathit{0}} \geq 0$, $\Gamma_i(\theta,r_{j})$ is decreasing in $r_{j}(\theta)$.

Assume first that $I_1$ has positive measure. Then $r_i(\theta) = 0$ for all $\theta \in I_1$. Since $C(\theta)$ is strictly increasing and $(1+\rho)C(\theta_2)=\pi_{II}-\pi_{\mathit{0}I}$, then $\Gamma_i(\theta,r_{j})>0$ for any $r_{j}$. Thus, the best response of firm $i$ is $r_i(\theta)=1$ for all $\theta \in I_1$, which is a contradiction.

Next, assume $I_3$ has positive measure. Then $r_i(\theta) =1$ for all $\theta \in I_3$. But, analogously to before, $\Gamma_i(\theta,r_{j})<0$ for any $r_{j}$. Thus, the best response of firm $i$ is $r_i(\theta)=0$. A contradiction.

Finally, assume $I_2$ has positive measure, which implies that $r_i(\theta)=r_j(\theta)$ for all $\theta\in I_2$. Suppose first that $r_i(\theta) =0$ on a set of positive measure $I_2' \subseteq I_2$. Observe that $\Gamma_{j}(\theta,r_{i})>0$ for all $\theta \in I_2'$. Since this is an equilibrium, it must be that $r_{j}(\theta)=1$ for all $\theta \in I_2'$. A contradiction. Next, suppose that $r_i(\theta) =1$ on a set of positive measure $I_2'' \subseteq I_2$. Observe that $\Gamma_{j}(\theta,r_{i})<0$ for all $\theta \in I_2''$. Analogously to the argument above, it must be that $r_{j}(\theta)=0$ for all $\theta \in I_2''$, a contradiction. Thus, it cannot be that $r_i(\theta)=r_j(\theta)$ for all $\theta\in I_2$. \qed

\subsection{Proof of Lemma \ref{prop_RJVInvestment}}

We can rewrite the expected total payoff of the RJV as
\begin{align*}
    \mathbb{E}\Pi_v(r_v)  = 2&\pi_{\mathit{0}\mathit{0}} + 2\left(\pi_{II}-\pi_{\mathit{0}\mathit{0}}\right) \int_0^1  r_v(\theta) d\theta  \\
&  - \int_0^1 r_v(\theta)C(\theta) d\theta  -\rho\max\left\{\int_0^1 r_v(\theta)C(\theta) d\theta - 2B,0 \right\}
\end{align*}
where the probability that the RJV discovers the innovation is given by $\int_0^1  r_v(\theta) d\theta$ while $\int_0^1 r_v(\theta)C(\theta) d\theta + \rho\max\left\{\int_0^1 r_v(\theta)C(\theta) d\theta - 2B,0 \right\}$ captures total innovation costs.

Since research projects only differ with respect to investment costs and these costs are increasing in $\theta$, for any fixed probability of innovation $\hat{\theta}$, the RJV optimally chooses a cut-off strategy to obtain this probability: It sets $r_v(\theta)=1$ for $\theta<\hat{\theta}$ and $r_v(\theta)=0$ otherwise, so that $\int_0^1  r_v(\theta)C(\theta) d\theta$=$\int_0^{\hat{\theta}} C(\theta) d\theta$.

The RJV's optimal portfolio can be obtained by maximizing  
\begin{align*}
    \mathbb{E}\hat{\Pi}_v(\hat{\theta})  = 2\pi_{\mathit{0}\mathit{0}} + 2\left(\pi_{II}-\pi_{\mathit{0}\mathit{0}}\right)\hat{\theta} - \int_0^{\hat{\theta}} C(\theta) d\theta  -\rho\max\left\{\int_0^{\hat{\theta}} C(\theta) d\theta - 2B,0 \right\}.
\end{align*}

Note that

\begin{equation*}
    \dfrac{\partial \mathbb{E}\hat{\Pi}_v}{\partial \hat{\theta}} =
    \begin{cases}
    2\left(\pi_{II}-\pi_{\mathit{0}\mathit{0}}\right) - C(\hat{\theta}) & \text{for } \hat{\theta} < \theta^{B} \\
    2\left(\pi_{II}-\pi_{\mathit{0}\mathit{0}}\right) - (1+\rho)C(\hat{\theta}) & \text{for } \hat{\theta} > \theta^{B}.
    \end{cases}
\end{equation*}

Now consider the three cases from the proposition (i.e., whether $\theta^{B}<\theta^{\rho}$,  $\theta^{B} \in [\theta^{\rho} , \theta^u]$, or $\theta^{B} > \theta^u$). First, if  $\theta^{B}<\theta^{\rho}$ then 
\begin{equation*}
    \dfrac{\partial \mathbb{E}\hat{\Pi}_v}{\partial \hat{\theta}} = 
    \begin{cases}
    2\left(\pi_{II}-\pi_{\mathit{0}\mathit{0}}\right) - C(\hat{\theta})>0 & \text{for } \hat{\theta} < \theta^{B} \\
    2\left(\pi_{II}-\pi_{\mathit{0}\mathit{0}}\right) - (1+\rho)C(\hat{\theta}) > 0 & \text{for } \hat{\theta} \in (\theta^{B}, \theta^{\rho}) \\
    2\left(\pi_{II}-\pi_{\mathit{0}\mathit{0}}\right) - (1+\rho)C(\hat{\theta}) < 0 & \text{for } \hat{\theta} \in (\theta^{\rho},1).
    \end{cases}
\end{equation*}
Thus, $ \hat{\theta} = \theta^{\rho}$ maximizes the expected return of the RJV's portfolio. Second, if $\theta^{B}\in [\theta^{\rho}, \theta^u] $ then
\begin{equation*}
    \dfrac{\partial \mathbb{E}\hat{\Pi}_v}{\partial \hat{\theta}} = 
    \begin{cases}
    2\left(\pi_{II}-\pi_{\mathit{0}\mathit{0}}\right) - C(\hat{\theta})>0 & \text{for } \hat{\theta} < \theta^{B} \\
    2\left(\pi_{II}-\pi_{\mathit{0}\mathit{0}}\right) - (1+\rho)C(\hat{\theta}) < 0 & \text{for } \hat{\theta} >\theta^{B},
    \end{cases}
\end{equation*}
so that $ \hat{\theta} = \theta^{B}$ maximizes the expected return of the RJV's portfolio. Third, if $\theta^{B} > \theta^u$ then 
\begin{equation*}
    \dfrac{\partial \mathbb{E}\hat{\Pi}_v}{\partial \hat{\theta}} = 
    \begin{cases}
    2\left(\pi_{II}-\pi_{\mathit{0}\mathit{0}}\right) - C(\hat{\theta})>0 & \text{for } \hat{\theta} < \theta^{u} \\
    2\left(\pi_{II}-\pi_{\mathit{0}\mathit{0}}\right) - C(\hat{\theta})<0 & \text{for } \hat{\theta} \in (\theta^{u}, \theta^{B}) \\
    2\left(\pi_{II}-\pi_{\mathit{0}\mathit{0}}\right) - (1+\rho)C(\hat{\theta}) < 0 & \text{for } \hat{\theta} \in (\theta^{B},1).
    \end{cases}
\end{equation*} Thus, $ \hat{\theta} = \theta^u$ maximizes the expected return of the RJV's portfolio.
\qed

\subsection{Proof of Proposition \ref{prop_RJVEffect}}

First, we provide a lemma distinguishing the two parts of Proposition \ref{prop_RJVEffect}.

\begin{lemma}\label{lemma_theta_rho_theta_1}
$2\pi_{II} > \pi_{I\mathit{0}} + \pi_{\mathit{0}\mathit{0}} \Leftrightarrow \theta^\rho > \theta_1$.
\end{lemma}
\begin{proof}
\begin{align*}
2\pi_{II} &> \pi_{I\mathit{0}} + \pi_{\mathit{0}\mathit{0}} \\
2 (\pi_{II} - \pi_{\mathit{0}\mathit{0}} ) &> \pi_{I\mathit{0}} - \pi_{\mathit{0}\mathit{0}} \\
(1+\rho)C(\theta^{\rho}) & > (1+\rho)C(\theta_1) \\
\theta^\rho &>\theta_1.
\end{align*}
\end{proof}

(i) By Lemma \ref{lemma_theta_rho_theta_1}, $2\pi_{II}>\pi_{I\mathit{0}}+\pi_{\mathit{0}\mathit{0}}$ implies $\theta^\rho > \theta_1$. By Lemma \ref{prop_RJVInvestment}, the probability that the RJV innovates is at least $\theta^{\rho}$. By Lemma \ref{prop_equilibrium}, the probability of innovation under competition is $\theta_1$. Therefore, the probability that the innovation will be discovered is strictly larger under the RJV than under competition. 

(ii) To prove part (a), we first provide an auxiliary result (Lemma \ref{lemma_theta_u_theta_1}). Using this lemma, we separately show that ``if'' part follows from Lemma \ref{lemma_Prop3i} below and ``only if'' part from Lemma \ref{lemma_Prop3ii} below.

\begin{lemma}\label{lemma_theta_u_theta_1}
Suppose $2\pi_{II} \leq \pi_{I\mathit{0}} + \pi_{\mathit{0}\mathit{0}}$. Then
$\rho > \bar{\rho} \Leftrightarrow  \theta^u > \theta_1$.
\end{lemma}
\begin{proof}
First suppose that $\bar{\rho}<\infty$. Then
\begin{align*}
\rho &> \bar{\rho} = \dfrac{\pi_{I\mathit{0}}-2\pi_{II}+\pi_{\mathit{0}\mathit{0}}}{2(\pi_{II}-\pi_{\mathit{0}\mathit{0}})}\\
2\rho(\pi_{II}-\pi_{\mathit{0}\mathit{0}})&>\pi_{I\mathit{0}}-2\pi_{II}+\pi_{\mathit{0}\mathit{0}}\\
2(1+\rho)(\pi_{II}-\pi_{\mathit{0}\mathit{0}})&>\pi_{I\mathit{0}}-\pi_{\mathit{0}\mathit{0}}\\
2(\pi_{II}-\pi_{\mathit{0}\mathit{0}})&>\dfrac{\pi_{I\mathit{0}}-\pi_{\mathit{0}\mathit{0}}}{1+\rho}\\
C(\theta^u) &> C(\theta_1) \\
\theta^u &> \theta_1.
\end{align*}
Next suppose $\bar{\rho}=\infty$. Then, clearly  
$\rho<\bar{\rho}$. Hence, the statement of the lemma holds if and only if $\theta^u  \leq \theta_1$ or $2(\pi_{II} -\pi_{\mathit{0}\mathit{0}}) 
\leq \frac{\pi_{I\mathit{0}} - \pi_{\mathit{0}\mathit{0}}}{1+\rho}$. As $\bar{\rho}=\infty$ implies $\pi_{II}=\pi_{\mathit{0}\mathit{0}}$, this requirement holds. 
\end{proof}

\begin{lemma}\label{lemma_Prop3i}
Suppose $2\pi_{II} \leq \pi_{I\mathit{0}} + \pi_{\mathit{0}\mathit{0}}$. If $B > \bar{B}(\rho)$ and $\rho > \bar{\rho}$, then the probability that the innovation will be discovered is strictly larger under the RJV than under competition.
\end{lemma}
\begin{proof} 
$B >\int_0^{\theta_1} C(\theta)d\theta/2$ implies  $\int_0^{\theta^{B}} C(\theta)d\theta >\int_0^{\theta_1} C(\theta)d\theta$ and therefore $\theta_1 < \theta^{B} $. Furthermore, by Lemma \ref{lemma_theta_rho_theta_1}, $\theta^\rho \leq \theta_1$  so that $\theta^\rho < \theta^{B}$. Then, either $\theta^{B}\in(\theta^\rho, \theta^u)$ or $\theta^{B} \geq \theta^u $. If $\theta^{B}\in(\theta^\rho, \theta^u)$, then the RJV invests in all projects in the set $(0, \theta^{B})$  and discovers the innovation with probability $\theta^{B}$. 
Without the RJV, in any equilibrium,  the firms invest in projects in the set $(0, \theta_1)$ and the innovation is discovered with probability $\theta_1$. Since $\theta^{B} > \theta_1$ it immediately follows that the probability of innovation strictly increases under the RJV.

Next, suppose $\theta^{B} \geq \theta^u $. Then, the RJV invests in all projects in the set $(0, \theta^u )$  and discovers the innovation with probability $\theta^u$. Since $\rho > \bar{\rho}$ implies $\theta^u > \theta_1$ by Lemma \ref{lemma_theta_u_theta_1}, it follows that the probability of innovation strictly increases under the RJV.
 \end{proof}

\begin{lemma}\label{lemma_Prop3ii}
Suppose $2\pi_{II} \leq \pi_{I\mathit{0}} + \pi_{\mathit{0}\mathit{0}}$. If the probability that the innovation will be discovered is strictly larger under the RJV than under competition, then  $B > \bar{B}(\rho)$ and $\rho > \bar{\rho}$.
\end{lemma}
\begin{proof}
As $2\pi_{II} \leq \pi_{I\mathit{0}} + \pi_{\mathit{0}\mathit{0}}$, Lemma \ref{lemma_theta_rho_theta_1} implies $\theta^\rho \leq \theta_1$. 
Hence, if the probability that the innovation will be discovered is strictly larger under the RJV than under competition, then $\theta^{B}>\theta^\rho$ by Lemma \ref{prop_RJVInvestment}. 
Therefore, either $\theta^{B}\in(\theta^\rho, \theta^u)$ or $\theta^{B} \geq \theta^u$. 
If $\theta^{B}\in(\theta^\rho, \theta^u)$, then,  by Lemma \ref{prop_RJVInvestment}, the increase in the probability of discovering the innovation under the RJV implies $\theta^{B} > \theta_1$, so that  $\theta^u>\theta^{B} > \theta_1$. 
If $\theta^{B} \geq \theta^u$, then the increase in the probability of discovering the innovation under RJV implies $\theta^u > \theta_1$, so that  $\theta^{B} \geq \theta^u> \theta_1$. In either case, both $\theta^u > \theta_1$ and $\theta^{B} > \theta_1$.

Note that $\theta^{B} > \theta_1$ implies 
\begin{align*}
    \int_0^{\theta^{B}} C(\theta)d\theta > \int_0^{\theta_1} C(\theta)d\theta.
\end{align*}
It follows immediately that $B >\int_0^{\theta_1} C(\theta)d\theta/2 = \bar{B}(\rho)$. Furthermore, $\theta^u > \theta_1$ implies, by Lemma \ref{lemma_theta_u_theta_1}, that $\rho > \bar{\rho}$.
\end{proof}

Finally, we prove part (b) of (ii).
With moderate or intense competition, $2\pi_{II}\leq\pi_{I\mathit{0}}+\pi_{\mathit{0}\mathit{0}}$. If the formation of the RJV strictly increases the probability of discovering the innovation, then, by Lemma \ref{lemma_Prop3ii}, $B > \bar{B}(\rho)$ and $\rho > \bar{\rho}$.
Using the same argument as in the proof of Lemma \ref{lemma_Prop3ii}, we have that either $\theta^{B}\in(\theta^\rho, \theta^u)$ or $\theta^{B} \geq \theta^u $. If $\theta^{B} \in (\theta^\rho,\theta^u)$, then by Lemma \ref{prop_RJVInvestment}, the total costs of the RJV are
\begin{equation}
    \int_0^{\theta^{B}}C(\theta)d\theta = 2B.
\end{equation}
If $\theta^{B} \geq \theta^u$, then by Lemma \ref{prop_RJVInvestment}, the total cost of the RJV are
\begin{equation}
    \int_0^{\theta^u}C(\theta)d\theta \leq 2B.
\end{equation}
By Lemma \ref{prop_equilibrium}, the total costs for equilibrium strategies  $r^*_i$ and $r^*_j$ under competition are
\begin{align*}
   (1+\rho) \int_0^1 [r^*_i(\theta)+r^*_j(\theta)]C(\theta)&d\theta-2\rho B \\
    &=(1+\rho)\left[2\int_0^{\theta_2}C(\theta)d\theta + \int_{\theta_2}^{\theta_1}[r^*_i(\theta)+r^*_j(\theta)]C(\theta) d\theta \right]-2\rho B \\
    &>   (1+\rho)\left[2B + \int_{\theta_2}^{\theta_1}[r^*_i(\theta)+r^*_j(\theta)]C(\theta)d\theta \right]-2\rho B \\
    &= 2B + (1+\rho) \int_{\theta_2}^{\theta_1}[r^*_i(\theta)+r^*_j(\theta)]C(\theta)d\theta \\
    &\geq 2B,
\end{align*}
where the first inequality follows from Assumption \ref{ass_budget}, and the second inequality from $\theta_1 \geq \theta_2$ and $r^*_i(\theta)+r^*_j(\theta) \geq 0$ for any $\theta$.

It immediately follows that the total cost under competition is weakly larger than the total cost under RJV, which proves the proposition.\qed

\subsection{Proof of Proposition \ref{Prop_CS}}

Denote with $\mathcal{P}^{com}_{II}$ the probability that both firms discover the innovation under competition and with $\mathcal{P}^{com}_{I\mathit{0}}$ the probability that a single firm discovers the innovation under competition. Analogously, let $\mathcal{P}^{rjv}_{II}$ be the probability that the innovation is discovered under the RJV. 

The expected consumer surplus is strictly higher under RJV than under competition if
\begin{align}
    \mathcal{P}^{rjv}_{II} CS_{II} + \left[ 1 - \mathcal{P}^{rjv}_{II} \right] CS_{\mathit{0}\mathit{0}} > \nonumber \\  \mathcal{P}^{com}_{II}CS_{II} + \mathcal{P}^{com}_{I\mathit{0}} CS_{I\mathit{0}} + \left[1 - \mathcal{P}^{com}_{II} - \mathcal{P}^{com}_{I\mathit{0}} \right]CS_{\mathit{0}\mathit{0}}. \label{eq_prop_cs1}
\end{align}

We proceed to show that this holds under the assumptions of the proposition. First, observe that by Assumption \ref{ass_CS}, $CS_{II}>CS_{I\mathit{0}}$, so that 
\begin{align}
    \left[ \mathcal{P}^{com}_{II} + \mathcal{P}^{com}_{I\mathit{0}} \right] CS_{II} + \left[1 - \mathcal{P}^{com}_{II} - \mathcal{P}^{com}_{I\mathit{0}} \right]CS_{\mathit{0}\mathit{0}}  \geq \nonumber \\  
    \mathcal{P}^{com}_{II}CS_{II} + \mathcal{P}^{com}_{I\mathit{0}} CS_{I\mathit{0}} + \left[1 - \mathcal{P}^{com}_{II} - \mathcal{P}^{com}_{I\mathit{0}} \right]CS_{\mathit{0}\mathit{0}}. \label{eq_prop_cs2}
\end{align}

Second, since the RJV increases the probability of innovation, $\mathcal{P}^{rjv}_{II} > \mathcal{P}^{com}_{II} + \mathcal{P}^{com}_{I\mathit{0}}$ and since by Assumption \ref{ass_CS}, $CS_{II}>CS_{\mathit{0}\mathit{0}}$, it must be that 
\begin{align}
    \mathcal{P}^{rjv}_{II} CS_{II} + \left[1 - \mathcal{P}^{rjv}_{II} \right]CS_{\mathit{0}\mathit{0}} > \nonumber \\
    \left[ \mathcal{P}^{com}_{II} + \mathcal{P}^{com}_{I\mathit{0}} \right] CS_{II} + \left[1 - \mathcal{P}^{com}_{II} - \mathcal{P}^{com}_{I\mathit{0}} \right]CS_{\mathit{0}\mathit{0}} . \label{eq_prop_cs3}
\end{align}

Finally, observe that combining inequalities \eqref{eq_prop_cs3}  and \eqref{eq_prop_cs2} gives inequality \eqref{eq_prop_cs1}, which completes the proof. \qed

\subsection{Profitability of RJVs: Proof of Proposition \ref{PropICNew}}
Using (\ref{WeakIC}), the RJV  strictly increases gross profits if and only if%
\begin{gather}
2\theta^{\ast}\pi_{II}  +2\left(  1-\theta^{\ast}\right)
\pi_{\mathit{0}\mathit{0}}  >\label{ProfDiffA2}\\
2\theta_{2}\pi_{II}  +\left(  \theta_{1}-\theta_{2}\right)
\left(  \pi_{I\mathit{0}}  +\pi_{\mathit{0}I}  \right)  +2\left(
1-\theta_{1}\right)  \pi_{\mathit{0}\mathit{0}} \text{,} \nonumber
\end{gather}
which can be rewritten as
\begin{gather}
2\left(  \theta_{1}-\theta_{2}\right)  \pi_{II}  +2\left(
\theta^{\ast}-\theta_{1}\right)  \pi_{II}  >\label{ProfDiffB2}%
\\
\left(  \theta_{1}-\theta_{2}\right)  \left(  \pi_{I\mathit{0}}
+\pi_{\mathit{0}I}  \right)  +2\left(  \theta^{\ast}-\theta_{1}\right)
\pi_{\mathit{0}\mathit{0}}  \text{.}\nonumber
\end{gather}

We use this condition to prove Proposition \ref{PropICNew}(i) to (iii) in turn, except for the statement in (ii) that, with moderate competition profitable RJVs need not increase innovation, which we will deal with separately in Proposition \ref{PropICBad} below.

(i) By Lemma \ref{lemma_theta_rho_theta_1}, soft competition ($2\pi_{II}>\pi_{I\mathit{0}}+\pi_{\mathit{0}\mathit{0}}$), implies $\theta^\rho> \theta_1$ and thus, by Lemma \ref{prop_RJVInvestment}, $\theta^*> \theta_1$. 
Further, observe that $2\pi_{II}>\pi_{I\mathit{0}}+\pi_{\mathit{0}\mathit{0}}$ implies $\pi_{II}>\pi_{\mathit{0}\mathit{0}}$, because $\pi_{II}=\pi_{\mathit{0}\mathit{0}}$ would imply $\pi_{II}>\pi_{I\mathit{0}}$, which contradicts Assumption \ref{ass_r}(iii). Thus, under soft competition, inequality \eqref{ProfDiffB2} holds and the RJV strictly increases gross profit.

If $\theta^{\rho} \leq \theta^{B}$, Lemma \ref{prop_RJVInvestment}
implies that the RJV spends exactly its budget or less;\ hence an RJV does not
increase R\&D expenditure and, as it strictly increases gross profits, it must also strictly increase net profits. If instead $\theta^{\rho} > \theta^{B}$, then Lemma \ref{prop_RJVInvestment}
implies that $\theta^*=\theta^{\rho}>\theta_1$. Thus, the RJV strictly increases the probability of innovation. By revealed preference, the RJVs profit must be at least as high as if it had chosen $\theta^*=\theta_1$. Even this choice would lead to higher net profits than R\&D competition: First, it saves the R\&D costs of duplication; second, for those values of $\theta$ where total gross profits under the RJV differ from those under R\&D competition ($\theta \in (\theta_2,\theta_1)$), total gross profits under the RJV are strictly higher than under R\&D competition, as  $2\pi_{II}  >\pi_{I\mathit{0}}
+\pi_{\mathit{0}\mathit{0}}\geq \pi_{I\mathit{0}}
+\pi_{\mathit{0}I}$.

(ii) By Proposition \ref{prop_equilibrium}, moderate competition and $B>\bar{B}(\rho)$ and $\rho>\overline{\rho}$ imply that the RJV weakly reduces cost
and that $\theta^{\ast}>\theta_{1}$. 
Further, $\rho>\bar{\rho}$ implies that $\pi_{II}>\pi_{\mathit{0}\mathit{0}}$. Otherwise, we would have $\bar{\rho}=\infty$ if $\pi_{II}=\pi_{\mathit{0}\mathit{0}}$, which contradicts $\rho>\bar{\rho}$. Hence, since $2\pi_{II}\geq\pi_{I\mathit{0}}+\pi_{\mathit{0}I}$ under moderate competition, together with $\pi_{II}>\pi_{\mathit{0}\mathit{0}}$ and $\theta^{\ast}>\theta_{1}$, inequality \eqref{ProfDiffB2} holds and the RJV strictly increases gross profits.  As it does not increase costs, it also increases net profits.

(iii)
Suppose first that $\theta^{\ast}>\theta_{1}$.
Rearranging (\ref{ProfDiffA2}), the requirement that the expected gross profit
difference is strictly positive becomes
\begin{equation}
\left(  \theta^{\ast}-\theta_{1}\right)  \left(  2\pi_{II}
-2\pi_{\mathit{0}\mathit{0}}  \right) > \left(  \theta_{1}-\theta_{2}\right)
\left(  \pi_{I\mathit{0}}  +\pi_{\mathit{0}I}  -2\pi_{II}  \right) \nonumber
\end{equation}
As $2\pi_{II}  <\pi_{I\mathit{0}}  +\pi_{\mathit{0}I} $ (intense competition) implies  $\theta_1>\theta_2$ and the restriction on $\Psi$ can only hold if $\pi_{II}  >\pi_{\mathit{0}\mathit{0}} $, we can rearrange again to get%
\[
\frac{\theta^{\ast}-\theta_{1}}{\theta_{1}-\theta_{2}} > \frac{\pi_{I\mathit{0}}  +\pi_{\mathit{0}I}  -2\pi_{II}  }{2\pi_{I,I}  -2\pi_{\mathit{0}\mathit{0}}  }=\Psi\text{.}
\]
Thus, provided $\theta^{\ast} > \theta_{1}$, $\frac{\theta^{\ast}-\theta_{1}
}{\theta_{1}-\theta_{2}}>\Psi$ is equivalent with the requirement that the RJV strictly
increases expected gross profits. But $\pi_{I\mathit{0}}  +\pi_{\mathit{0}I}  -2\pi_{II} >  0$ implies $\Psi > 0$. Using $\theta
_{1}-\theta_{2} > 0,$ $\frac{\theta^{\ast}-\theta_{1}}{\theta_{1}-\theta_{2}
} > \Psi$ thus implies $\theta^{\ast}>\theta_{1}+\Psi\left(  \theta_{1}
-\theta_{2}\right) >\theta_{1}$.

Further, $2\pi_{II}
<\pi_{I\mathit{0}}  +\pi_{\mathit{0}I} \leq\pi_{I\mathit{0}}  +\pi_{\mathit{0}\mathit{0}} $ implies that $\theta_1>\theta^{\rho}$ and $\Psi>0$. Hence, $\frac{\min\{\theta^{B},\theta^{u}\}-\theta_{1}}%
{\theta_{1}-\theta_{2}} > \Psi$ implies $\theta^{B} \geq \theta^{\rho}$. Therefore, using Lemma \ref {prop_RJVInvestment}, we obtain that $\theta^{\ast}=\min\{\theta^{B},\theta^{u}\}$, so
that the RJV is not spending more than its budget and hence not more than the
individual firms.
Therefore, the RJV strictly increases gross profit and (as it does not increase costs) net profits.

Next, we finish the proof by dealing deal with the possibility that an RJV may be profitable in spite of reducing innovation incentives (see Proposition \ref{PropICNew}(ii)). We sharpen the statement as follows.

\begin{proposition}[Profitable innovation-reducing RJV]\label{PropICBad} $ $ \newline
    Suppose that the following conditions hold:
    \begin{enumerate}[label=(\roman*)]
   \item $2\pi_{II} -(\pi_{I\mathit{0}}  +\pi_{\mathit{0}\mathit{0}}) =0$.
    \item $B \leq \bar{B}(\rho)$ or $\rho \leq \bar{\rho}$.
    \item $\pi_{II} > \pi_{\mathit{0}I}$.
   \end{enumerate}
    Then there exists some $\hat{\pi}_{I\mathit{0}}>\pi_{I\mathit{0}}$ such that for all $\pi_{I\mathit{0}}' \in ( \pi_{I\mathit{0}}, \hat{\pi}_{I\mathit{0}}) $ and keeping other parameters fixed, the RJV is profitable, but reduces the innovation. 
\end{proposition}
\begin{proof}
   We first show that when $2\pi_{II} -(\pi_{I\mathit{0}}  +\pi_{\mathit{0}\mathit{0}}) =0$, the RJV leaves the innovation probability unaffected. To see this, first note that, together with condition (ii), Proposition \ref{prop_equilibrium} implies that the innovation probability is not strictly higher in the RJV than under R\&D competition. Next, $2\pi_{II} -(\pi_{I\mathit{0}}  +\pi_{\mathit{0}\mathit{0}}) =0$ implies that $\theta_{1}=\theta^{\rho}$. As $\theta^* \geq \theta^{\rho}$ by Lemma \ref{prop_RJVInvestment}, we obtain $\theta^* \geq \theta_1$, so that the RJV does not have a negative effect on the innovation probability either. All told, there is no effect of the RJV on the innovation probability.

Next, still assuming that $2\pi_{II} -(\pi_{I\mathit{0}}  +\pi_{\mathit{0}\mathit{0}}) =0$, total gross profits are weakly higher in the RJV than under competition for $\theta \in (\theta_2,\theta_1)$ because $2\pi_{II}\geq\pi_{I\mathit{0}}  +\pi_{\mathit{0}I}$. As gross profits are the same with and without RJV for the remaining realizations of $\theta$, expected total gross profits are at least weakly higher with the RJV than without. Note that by condition (iii), $\theta_2>0$, so that the costs with the RJV are strictly lower than the total costs with R\&D competition, with the difference being equal to $\int_0^{\theta_2}C(\theta)d\theta$.

Finally, observe that a ceteris paribus increase from $\pi_{I\mathit{0}}$ to $\pi_{I\mathit{0}}'$ does not affect $\theta_2$ nor $\theta^*$ but increases $\theta_1$ to some $\theta'_1 > \theta^*$. By continuity, there exists some $\hat{\pi}_{I\mathit{0}}$ such that for all $\pi_{I\mathit{0}}' \in (\pi_{I\mathit{0}}, \hat{\pi}_{I\mathit{0}}) $, the change in total gross profits under R\&D competition is smaller than $\int_0^{\theta_2}C(\theta)d\theta$. For all such $\pi_{I\mathit{0}}'$, the RJV is profitable but it decreases the innovation probability from $\theta'_1$ to $\theta^*$.
 
\end{proof}

Note that conditions guarantee that this case arises when competition is moderate, but parameters are close to the soft competition regime. Then, under condition (ii) in Proposition \ref{PropICBad}, an RJV woudd reduce innovation slightly, but without major adverse effects on gross profits. The cost-reducing effect of an RJV will then suffice to make it profitable.

\subsection{Equilibrium Selection via Risk Dominance}
\label{SecAppRisk}

Consider a variant of the original model, where the firms have different borrowing costs. Firm $h$ has a borrowing cost of $\rho_{h}$, while firm $\ell$ has a borrowing cost of $\rho_{\ell} = \rho_{h} - \epsilon$, where $\epsilon > 0$ is arbitrarily small. Define the following cutoffs $\theta_i^k$:
\begin{align*}
    (1+\rho_h)C(\theta_1^h) &= \pi_{I0} - \pi_{00}, \\
    (1+\rho_h)C(\theta_2^h) &= \pi_{II} - \pi_{0I}, \\
    (1+\rho_{\ell})C(\theta_1^{\ell}) &= \pi_{I0} - \pi_{00}, \\
    (1+\rho_{\ell})C(\theta_2^{\ell}) &= \pi_{II} - \pi_{0I}.
\end{align*}
There always exists $\epsilon$ small enough such that $\theta_2^{\ell} < \theta_1^h$. Then, the following inequality holds:
\begin{align*}
    \theta_2^h < \theta_2^{\ell} < \theta_1^h < \theta_1^{\ell}.
\end{align*}

We maintain Assumption \ref{ass_budget}, with $\theta_2^h$ substituting $\theta_2$. Using arguments which are by now familiar, we can show that in any equilibrium:
\begin{enumerate}[label=(\roman*)]
    \item both firms invest in $[0, \theta_2^h)$,
    \item firm $\ell$ invests in $[\theta_2^h, \theta_2^{\ell})$,
    \item either firm $\ell$ or firm $h$ invests in each $\theta$ in $[\theta_2^{\ell}, \theta_1^h)$,
    \item firm $\ell$ invests in all $\theta$ in $[\theta_1^h, \theta_1^{\ell})$,
    \item and no firm invests in $[\theta_1^{\ell}, 1)$.
\end{enumerate}

Denote the strategy $(r^d_{\ell}, r^d_h)$ as the following equilibrium strategy:
\[
r^d_{\ell}(\theta) = \begin{cases}
1 & \text{for all } \theta \in [0, \theta_1^{\ell}) \\
0 & \text{otherwise}
\end{cases}
\]
\[
r^d_h(\theta) = \begin{cases}
1 & \text{for all } \theta \in [0, \theta_2^{h}) \\
0 & \text{otherwise}.
\end{cases}
\]
Denote with $(r_{\ell}^*, r_h^*)$ any other equilibrium strategy. Let \( D \subseteq [\theta_2^{\ell}, \theta_1^{h}] \) be the set of projects $\theta$ on which $(r_{\ell}^*, r_h^*)$ differs from $(r^d_{\ell}, r^d_h)$. Let \( \mathcal{M}^D = \int_D dj >0\) be the probability that the successful project is in the set $D$. 

We are interested in the (infinite) set of 2x2 games where firm \( \ell \) can choose between \( r_{\ell}^d \)  and \( r_{\ell}^* \), and firm \( h \) can choose between \( r_h^d \) and \( r_h^* \).
\begin{table}[htb!]
    \centering
    \setlength{\extrarowheight}{2pt}
    \begin{tabular}{cc|c|c|}
      & \multicolumn{1}{c}{} & \multicolumn{2}{c}{Firm $h$}\\
      & \multicolumn{1}{c}{} & \multicolumn{1}{c}{$r^*_h$}  & \multicolumn{1}{c}{$r^d_h$} \\\cline{3-4}
      \multirow{2}*{Firm $\ell$}  & $r^*_{\ell}$ & $\mathbb{E}\Pi_\ell(r_{\ell}^*, r_h^*), \mathbb{E}\Pi_h(r_h^*,r_{\ell}^*)$ & $\mathbb{E}\Pi_\ell(r_{\ell}^*, r_h^d), \mathbb{E}\Pi_h(r_h^d, r_\ell^*)$ \\\cline{3-4}
      & $r^d_{\ell}$ & $\mathbb{E}\Pi_\ell(r_{\ell}^d, r_h^*), \mathbb{E}\Pi_h(r_h^*, r_{\ell}^d)$ & $\mathbb{E}\Pi_\ell(r_{\ell}^d, r_h^d), \mathbb{E}\Pi_h(r_h^d, r_\ell^d)$ \\\cline{3-4}
    \end{tabular}
    \caption{Payoff matrix of the 2x2 game.}
    \label{payoff_matrix_2x2_game}
  \end{table}

The expected equilibrium payoffs are as follows:
\begin{align*}
    \mathbb{E}\Pi_{\ell}(r_{\ell}^d, r_h^d) &= \theta_2^{h} \pi_{II} + \left( \theta_1^{\ell} - \theta_2^{h} \right) \pi_{I0} + \left( 1 - \theta_1^{\ell} \right) \pi_{00} \\
    &\quad - \left( 1 + \rho_\ell \right) \int_0^{\theta_1^{{\ell}}} C(j) dj + \rho_{\ell} B \\
    \mathbb{E}\Pi_{h}(r_h^d,r_{\ell}^d) &= \theta_2^{h} \pi_{II} + \left( \theta_1^{\ell} - \theta_2^{h} \right) \pi_{0I} + \left( 1 - \theta_1^{\ell} \right) \pi_{00} \\
    &\quad - \left( 1 + \rho_h \right) \int_0^{\theta_2^{{\ell}}} C(j) dj + \rho_{h} B \\
    \mathbb{E}\Pi_{\ell}(r_{\ell}^*, r_h^*) &= \theta_2^{h} \pi_{II} + \left( \theta_1^{\ell} - \theta_2^{h} - \mathcal{M}^D \right) \pi_{I0} + \mathcal{M}^D \pi_{0I} + \left( 1 - \theta_1^{\ell} \right) \pi_{00} \\
    &\quad - \left( 1 + \rho_\ell \right) \int_{[0,\theta_1^{\ell}) \setminus \mathcal{M}^D} C(j) dj + \rho_{\ell} B\\
    \mathbb{E}\Pi_{h}(r_h^*,r_{\ell}^*) &= \theta_2^{h} \pi_{II} + \left( \theta_1^{\ell} - \theta_2^{h} - \mathcal{M}^D\right) \pi_{0I} + \mathcal{M}^D \pi_{I0} + \left( 1 - \theta_1^{\ell} \right) \pi_{00} \\
    &\quad - \left( 1 + \rho_h \right) \int_{[0,\theta_2^{h}) \cup \mathcal{M}^D} C(j) dj + \rho_{h} B
\end{align*}

Non-equilibrium payoffs are as follows:
\begin{align*}
    \mathbb{E}\Pi_{\ell}(r_{\ell}^d, r_h^*) &= \theta_2^{h} \pi_{II} + \left( \theta_1^{\ell} - \theta_2^{h} - \mathcal{M}^D\right) \pi_{I0} + \mathcal{M}^D \pi_{II} + \left( 1 - \theta_1^{\ell} \right) \pi_{00} \\
    &\quad - \left( 1 + \rho_\ell \right) \int_0^{\theta_1^{{\ell}}} C(j) dj + \rho_{\ell} B \\
    \mathbb{E}\Pi_{h}(r_h^*,r_{\ell}^d) &= \theta_2^{h} \pi_{II} + \left( \theta_1^{\ell} - \theta_2^{h} - \mathcal{M}^D \right) \pi_{0I} + \mathcal{M}^D \pi_{II} + \left( 1 - \theta_1^{\ell} \right) \pi_{00} \\
    &\quad - \left( 1 + \rho_h \right) \int_{[0,\theta_2^{h}) \cup \mathcal{M}^D} C(j) dj + \rho_{h} B \\
    \mathbb{E}\Pi_{\ell}(r_{\ell}^*, r_h^d) &= \theta_2^{h} \pi_{II} + \left( \theta_1^{\ell} - \theta_2^{h} - \mathcal{M}^D \right) \pi_{I0} + \mathcal{M}^D \pi_{00} + \left( 1 - \theta_1^{\ell} \right) \pi_{00} \\
    &\quad - \left( 1 + \rho_\ell \right) \int_{[0,\theta_1^{\ell}) \setminus \mathcal{M}^D} C(j) dj + \rho_{\ell} B\\
    \mathbb{E}\Pi_{h}(r_h^d,r_{\ell}^*) &= \theta_2^{h} \pi_{II} + \left( \theta_1^{\ell} - \theta_2^{h} - \mathcal{M}^D\right) \pi_{0I} + \mathcal{M}^D \pi_{00} + \left( 1 - \theta_1^{\ell} \right) \pi_{00} \\
    &\quad - \left( 1 + \rho_h \right) \int_0^{\theta_2^{h}} C(j) dj + \rho_{h} B
\end{align*}

Let a game $\mathcal{G}$ be any 2x2 game in which the actions of the firm $i \in \{\ell,h\}$ are $\{r_{i}^d,r_{i}^*\}$, where $(r_{\ell}^d,r_{h}^d)$ the equilibrium specified above and $(r_{\ell}^*,r_{h}^*)$ is any other equilibrium of the original game. The payoffs are as specified in Table \ref*{payoff_matrix_2x2_game}.  For any such game, we prove the following result.

\begin{proposition}[Risk dominance]
    In any game $\mathcal{G}$, the strategy profile $(r_{\ell}^d,r_{h}^d)$ is risk-dominant.    
\end{proposition}

\begin{proof}
For notational simplicity and consistency with \cite{harsanyi1988general}, we will denote the payoffs from the table above as follows:

\begin{table}[htb!]
    \centering
    \setlength{\extrarowheight}{2pt}
    \begin{tabular}{cc|c|c|}
      & \multicolumn{1}{c}{} & \multicolumn{2}{c}{Firm $h$}\\
      & \multicolumn{1}{c}{} & \multicolumn{1}{c}{$r^*_h$}  & \multicolumn{1}{c}{$r^d_h$} \\\cline{3-4}
      \multirow{2}*{Firm $\ell$}  & $r^*_{\ell}$ & $a_{11}, b_{11}$ & $a_{12}, b_{12}$ \\\cline{3-4}
      & $r^d_{\ell}$ & $a_{21}, b_{21}$ & $a_{22}, b_{22}$ \\\cline{3-4}
    \end{tabular}
  \end{table}

Using the notation from \cite{harsanyi1988general}, we have:
\begin{align*}
u_1 &= a_{11} - a_{21} = \mathcal{M}^D (\pi_{0I} -  \pi_{II}) + (1 + \rho_{\ell}) \int_{\mathcal{M}^D} C(j) dj > 0 \\
u_2 &= b_{11} - b_{12} = \mathcal{M}^D (\pi_{I0} - \pi_{00}) - (1 + \rho_h) \int_{\mathcal{M}^D} C(j) dj > 0 \\
v_1 &= a_{22} - a_{12} = \mathcal{M}^D (\pi_{I0} - \pi_{00}) - (1 + \rho_{\ell}) \int_{\mathcal{M}^D} C(j) dj > 0 \\
v_2 &= b_{22} - b_{21} = \mathcal{M}^D (\pi_{0I} - \pi_{II}) + (1 + \rho_h) \int_{\mathcal{M}^D} C(j) dj > 0
\end{align*}

By Theorem 3.9.1 in \cite{harsanyi1988general}, the strategy profile $(r_{\ell}^d,r_{h}^d)$ will be risk-dominant if $v_1v_2 \geq u_1 u_2$.

This is equivalent to showing that:
\begin{align*}
    \left[\mathcal{M}^D (\pi_{I0} - \pi_{00}) - (1 + \rho_{\ell}) \int_{\mathcal{M}^D} C(j) dj\right]\left[ \mathcal{M}^D (\pi_{0I} - \pi_{II}) + (1 + \rho_h) \int_{\mathcal{M}^D} C(j) dj \right] \geq \\
    \left[\mathcal{M}^D (\pi_{0I} -  \pi_{II}) + (1 + \rho_{\ell}) \int_{\mathcal{M}^D} C(j) dj \right]\left[ \mathcal{M}^D (\pi_{I0} - \pi_{00}) - (1 + \rho_h) \int_{\mathcal{M}^D} C(j) dj\right]
\end{align*}
$\iff$
\begin{align*}
    \left(\mathcal{M}^D\right)^2 (\pi_{I0} - \pi_{00})(\pi_{0I} - \pi_{II}) + \mathcal{M}^D (\pi_{I0} - \pi_{00})(1 + \rho_h) \int_{\mathcal{M}^D} C(j) dj \\
    - \mathcal{M}^D (\pi_{0I} - \pi_{II})(1 + \rho_{\ell}) \int_{\mathcal{M}^D} C(j) dj - (1 + \rho_{\ell}) (1 + \rho_{h})\left( \int_{\mathcal{M}^D} C(j) dj\right)^2 \geq \\
    \left(\mathcal{M}^D\right)^2 (\pi_{I0} - \pi_{00})(\pi_{0I} - \pi_{II}) - \mathcal{M}^D (\pi_{0I} - \pi_{II})(1 + \rho_h) \int_{\mathcal{M}^D} C(j) dj \\
    + \mathcal{M}^D (\pi_{I0} - \pi_{00})(1 + \rho_{\ell}) \int_{\mathcal{M}^D} C(j) dj - (1 + \rho_{\ell}) (1 + \rho_{h})\left( \int_{\mathcal{M}^D} C(j) dj\right)^2
\end{align*}
$\iff$
\begin{align*}
    (\pi_{I0} - \pi_{00})(1 + \rho_h) 
    -  (\pi_{0I} - \pi_{II})(1 + \rho_{\ell})   \geq \\
    -  (\pi_{0I} - \pi_{II})(1 + \rho_h) 
    + (\pi_{I0} - \pi_{00})(1 + \rho_{\ell})
\end{align*}
$\iff$
\begin{align*}
    (1 + \rho_h)(\pi_{I0} - \pi_{00}+\pi_{0I} - \pi_{II}) \geq
    (1 + \rho_{\ell})(\pi_{I0} - \pi_{00}+\pi_{0I} - \pi_{II})
\end{align*}
which always holds since $(\pi_{I0} - \pi_{00}+\pi_{0I} - \pi_{II})\geq0$ by Assumption 1(iv).
\end{proof}

 \subsection{Examples}\label{SecAppExamples}
 
 \subsubsection{Linear Cournot Competition}

We now sketch the details for the Cournot example of Section \ref{SecCournot}. Using the
notation $\alpha=a-c$, it is straightforward to show that, under the
assumption that $\alpha>I$ an equilibrium with positive outputs and profits
exists for both firms, so that the innovation is non-drastic. The equilibrium
profits are given as $\pi_{I\mathit{0}}  =\frac{1}{9}\frac{\left(
\alpha+2I\right)  ^{2}}{b}$, $\pi_{II}  =\frac{1}{9}\frac{\left(
\alpha+I\right)  ^{2}}{b}$, $\pi_{\mathit{0}\mathit{0}}  =\frac{1}{9}\frac
{\alpha^{2}}{b}$, $\pi_{\mathit{0}I}  =\frac{1}{9}\frac{\left(
\alpha-I\right)  ^{2}}{b}$. These expressions imply that, whenever $\alpha>I$,
Assumption 1 holds, as well as the stricter condition that competition is not soft
required by Proposition 1(ii). Next, Corollary 1 follows
directly from inserting these profit expressions in the term $\overline{\rho}
$. \ Furthermore, the boundary between intense and moderate competition, given by
$2\pi_{II}  =\pi_{I\mathit{0}}  +\pi_{\mathit{0}I}  $,
can be calculated as $\alpha=3I/2$.

\subsubsection{Differentiated Price Competition}

We now add further details for the case of price competition with inverse demand $p_{i}=1-q_{i}-bq_{j}$
for $b\in\left[  0,1\right)  $. We assume that cost differences are not too
large $\bigl( c_{i} < \frac{2-b-b^{2}+bc_{j}}{2-b^{2}}\bigr)$ for $i \in \{1,2\}, j \neq i$. Then standard calculations show that both equilibrium outputs are positive, with equilibrium profit
\begin{equation}\label{ProfitDPC}
\pi_{i}=\frac{\left(  2-b-b^{2}-\left(  2-b^{2}\right)  c_{i}+bc_{j}\right)
^{2}}{\left(  4-b^{2}\right)  ^{2}\left(  1-b^{2}\right)  }.
\end{equation}

Inserting appropriate values for $c_1$ and $c_2$ gives
\begin{align*}
\pi_{I\mathit{0}}   &  =\frac{\left(  2-b-b^{2}-\left(  2-b^{2}\right)
\left(  c-I\right)  +bc\right)  ^{2}}{\left(  4-b^{2}\right)  ^{2}\left(
1-b^{2}\right)  }\\
\pi_{II}   &  =\frac{\left(  2-b-b^{2}-\left(  2-b^{2}\right)
\left(  c-I\right)  +b\left(  c-I\right)  \right)  ^{2}}{\left(
4-b^{2}\right)  ^{2}\left(  1-b^{2}\right)  }\\
\pi_{\mathit{0}\mathit{0}}   &  =\frac{\left(  2-b-b^{2}-\left(  2-b^{2}\right)
c+bc\right)  ^{2}}{\left(  4-b^{2}\right)  ^{2}\left(  1-b^{2}\right)  }\\
\pi_{\mathit{0}I}   &  =\frac{\left(  2-b-b^{2}-\left(  2-b^{2}\right)
c+b\left(  c-I\right)  \right)  ^{2}}{\left(  4-b^{2}\right)  ^{2}\left(
1-b^{2}\right)  }%
\end{align*}

The requirement that all profits be non-negative $\bigl( c_{i}\leq
\frac{2-b-b^{2}+bc_{j}}{2-b^{2}} \bigr)$ is most demanding when $c_{i}=c$ and
$c_{j}=c-I$, in which case it can be guaranteed by assuming $\ $
\[
I<\frac{b^{2}c-2c-b+bc-b^{2}+2}{b}%
\]
In Figure \ref{fig:Bertrand}, we set $c=0.5$ and hence
\begin{equation}
I<\frac{0.5}{b}\left(  2-b^{2}-b\right)\nonumber
\end{equation}
The assumptions of the paper can easily be verified in this case. We also find expressions for the regions plotted in Figure \ref{fig:Bertrand}. After some rearrangements, the condition that $\pi_{I\mathit{0}}  +\pi_{\mathit{0}\mathit{0}}
 \leq 2\pi_{II} $ becomes
\begin{equation}
\left(  4-8b-2b^{2}+4b^{3}+b^{4}\right)  I\geq\\
8c+12b-4b^{2}c+6b^{3}c+2b^{4}c-12bc+4b^{2}-6b^{3}-2b^{4}-8\nonumber
\end{equation}
For $c=0.5$, this simplifies to%
\begin{equation}
4I-8bI-2b^{2}I+4b^{3}I+b^{4}I\geq-b^{4}-3b^{3}+2b^{2}+6b-4\nonumber
\end{equation}
The condition $\pi_{I\mathit{0}}  +\Pi\left(
0,I\right)  <2\pi_{II} $ becomes
\begin{equation}
4I-4b-8bI-3b^{2}I+4b^{3}I+b^{4}I-3b^{2}+2b^{3}+b^{4}+4\geq0\text{.}\nonumber
\end{equation}

\subsection{Mergers}

In this section, we first provide formal statements of the informal claims in the main text; thereafter, we state and prove the central result comparing mergers and RJVs.

\subsubsection{Optimal R\&D portfolio of Merged Entity}\label{SecAppMergerPortfolio}

We first describe the investment behavior of the merged entity in a similar way as for the RJV (see Lemma \ref{prop_RJVInvestment}).

\begin{lemma}\label{prop_mergerInvestment}
The merged entity chooses a single cut-off strategy with 
\begin{align*}
    \hat{\theta}= \begin{cases} 
    \theta^\rho_m  \quad &\text{if} \quad \theta^{B} < \theta^\rho_m \\
     \theta^{B} \quad &\text{if} \quad \theta^{B} \in [\theta^\rho_m,\theta^u_m] \\
     \theta^u_m  \quad &\text{if} \quad \theta^{B}>\theta^u_m.
    \end{cases}
\end{align*}
\end{lemma}

\begin{proof}
The proof is entirely analogous to the proof of Lemma \ref{prop_RJVInvestment}. We merely have to replace $\pi_{\mathit{0}\mathit{0}}$ with $\pi_{\mathit{0}}$ and $\pi_{II}$ with $\pi_I$ in the profit expressions and adjust the critical values.
\end{proof}

Next, we adapt Proposition \ref{prop_RJVEffect} to the case of mergers.

\begin{proposition}[Comparison of R\&D-Competition and Mergers]\label{prop_CompMergers} $ $
\begin{enumerate}[label=(\roman*)]
\item Suppose $\pi_I -\pi_{\mathit{0}}>\pi_{I\mathit{0}}-\pi_{II}$. Then the innovation probability is strictly larger after the merger than under R\&D competition.
\item Suppose $\pi_I-\pi_{\mathit{0}} \leq \pi_{I\mathit{0}}- \pi_{II}$. Then: 
\newline (a) The innovation probability is strictly larger after the merger than in any equilibrium under competition if and only if $B > \bar{B}(\rho)$ and $\rho > \bar{\rho}_m$.\newline
(b) If the merger strictly increases the innovation probability, then it  weakly decreases total R\&D spending.
\end{enumerate}
\end{proposition}

The proof is analagous to the proof of Proposition \ref{prop_RJVEffect}.

\subsubsection{Proof of Proposition \ref{prop_comparison_m-rjv}: Comparing Mergers and RJVs}\label{SecCompMergerRJV}

We require an auxiliary result, the proof of which is obvious.

\begin{lemma} \label{lemma_order_m-rjv}
$2(\pi_{II}-\pi_{\mathit{0}\mathit{0}}) \gtreqqless \pi_I - \pi_{\mathit{0}} \Leftrightarrow \theta^u \gtreqqless \theta^u_m \Leftrightarrow \theta^\rho \gtreqqless \theta^\rho_m$.
\end{lemma}

Next, we prove the two statements of the proposition in turn.

$(i)$ Suppose that $2(\pi_{II}-\pi_{\mathit{0}\mathit{0}})>\pi_I - \pi_{\mathit{0}}$. By Lemma \ref{lemma_order_m-rjv}, $\theta^\rho_m < \theta^\rho$ and $\theta^u_m < \theta^u $. This implies that there are two possible orderings of critical values:
\begin{align}
    \theta^\rho_m < \theta^\rho &\leq \theta^u_m < \theta^u  \label{ineq_critorder1}\\
    \theta^\rho_m < \theta^u_m & < \theta^\rho  < \theta^u. \label{ineq_critorder2}
\end{align}

Suppose first that ordering \eqref{ineq_critorder1} holds.
If $\theta^{B}<\theta^\rho_m$, then the RJV invests in the set $(0,\theta^\rho)$ and the merged firm in $(0,\theta^\rho_m)$. If $\theta^{B}\in[\theta^\rho_m,\theta^\rho)$, then the RJV invests in the set $(0,\theta^\rho)$ and the merged firm in $(0,\theta^{B})$. Hence, since $\theta^\rho>\theta^\rho_m$, it follows that the RJV invests in a larger set than the merged firm whenever $\theta^{B}<\theta^\rho$. If $\theta^{B}\in[\theta^\rho, \theta^u_m]$, then both invest in the identical set $(0,\theta^{B})$. If $\theta^{B}\in(\theta^u_m,\theta^u)$, then the RJV invests in the set $(0,\theta^{B})$ and the merged firm in $(0,\theta^u_m)$. If $\theta^{B}\geq\theta^u$, then the RJV invests in the set $(0, \theta^u)$, whereas the merged firm still invests in $(0,\theta^u_m)$. Hence, since $\theta^u>\theta^u_m$, it follows that the RJV invests in a larger set than the merged firm whenever $\theta^{B}>\theta^u_m$.  

Now suppose that ordering \eqref{ineq_critorder2} holds.
The analysis for $\theta^{B} < \theta^\rho_m $ and $\theta^{B} \geq \theta^u$ is unchanged. If $\theta^{B}\in[\theta^\rho_m,\theta^u_m)$, then the RJV invests in the set $(0,\theta^\rho)$ and the merged firm in $(0,\theta^{B})$. If $\theta^{B}\in[\theta^u_m,\theta^\rho)$, then the RJV still invests in the set $(0,\theta^\rho)$, and the merged firm invests in $(0,\theta^u_m)$. If $\theta^{B}\in[\theta^\rho, \theta^u)$, then the RJV invests in the set $(0,\theta^{B})$ and the merged firm in $(0,\theta^u_m)$.
Hence, whenever ordering \eqref{ineq_critorder2} holds, the RJV invests in a larger set than the merged firm.

Next, suppose that $2(\pi_{II}-\pi_{\mathit{0}\mathit{0}})=\pi_I - \pi_{\mathit{0}}$. By  Lemma \ref{lemma_order_m-rjv}, $\theta^u = \theta^u_m$ and $\theta^\rho = \theta^\rho_m$. If $\theta^{B}<\theta^\rho=\theta^\rho_m$, then both the RJV and the merged firm invest in $(0,\theta^\rho)$. If $\theta^{B}\in[\theta^\rho, \theta^u)$, then both invest in the set $(0,\theta^{B})$. If $\theta^{B} \geq \theta^u=\theta^u_m$, then both the RJV and the merged firm invest in the set $(0,\theta^u)$. Hence, for any $\theta^{B}$ both the RJV and the merged firm invest in the same set of research projects.

$(ii)$ Suppose that $2(\pi_{II}-\pi_{\mathit{0}\mathit{0}})<\pi_I - \pi_{\mathit{0}}$. By  Lemma \ref{lemma_order_m-rjv}, $\theta^u < \theta^u_m$ and $\theta^\rho < \theta^\rho_m$. This implies that there are two possible orderings of critical values:
\begin{align}
    \theta^\rho < \theta^\rho_m &\leq \theta^u < \theta^u_m  \label{ineq_critorder3}\\
    \theta^\rho < \theta^u & < \theta^\rho_m  < \theta^u_m. \label{ineq_critorder4}
\end{align} 

Suppose first that ordering \eqref{ineq_critorder3} holds. 
If $\theta^{B}<\theta^\rho$, then the merged firm invests in the set $(0,\theta^\rho_m)$ and the RJV in $(0,\theta^\rho)$. If $\theta^{B}\in[\theta^\rho,\theta^\rho_m)$, then the merged firm invests in the set $(0,\theta^\rho_m)$ and the RJV in $(0, \theta^{B})$.  Hence, if $\theta^{B}<\theta^\rho_m$ the merged firm invests in a larger set than the RJV. If $\theta^{B}\in[\theta^\rho_m,\theta^u]$, then both invest in the identical set $(0,\theta^{B})$. If $\theta^{B}\in(\theta^u,\theta^u_m)$, then the merged firm invests in the set $(0, \theta^{B})$ and the RJV in $(0,\theta^u)$. If $\theta^{B} \geq \theta^u_m$, then the merged firm invests in the set $(0,\theta^u_m)$ and the RJV still in $(0,\theta^u)$. Hence, whenever $\theta^{B}>\theta^u$ the merged firm invests in a larger set than the RJV. 

Now suppose that ordering \eqref{ineq_critorder4} holds and consider again different values that $\theta^{B}$ can take. The analysis for $\theta^{B} < \theta^\rho $ and $\theta^{B} \geq \theta^u_m$ is unchanged. If $\theta^{B}\in[\theta^\rho,\theta^u)$, then the merged firm invests in the set $(0,\theta^\rho_m)$ and the RJV in $(0, \theta^{B})$. If $\theta^{B}\in[\theta^u,\theta^\rho_m)$, then the merged firm invests in the set $(0,\theta^\rho_m)$ and the RJV in $(0, \theta^u)$. If $\theta^{B}\in[\theta^\rho_m, \theta^u_m)$, then the merged firm invests in the set $(0,\theta^{B})$ and the RJV in $(0, \theta^u)$. Hence, whenever the ordering \eqref{ineq_critorder4} holds, the  merged firm invests in a larger set than the RJV. \qed

\subsection{Spillovers}

\subsubsection{No Financial Constraints}\label{SecAppSpillNFC}
As a benchmark, we now consider a model without financial constraints.
Instead, we allow for spillovers. Specifically, if a firm has invested successfully in a project and the rival has not, then with probability $\sigma\in[0,1]$ the rival will obtain access to the innovation. The expected total payoff of firm $i$, given the strategy of firm $j$ is then
\begin{align*}
\mathbb{E}\Pi_i(r_i, &r_{j})= \int_0^1  (1-r_{j}(\theta))\left[r_i(\theta)((1-\sigma)\pi_{I\mathit{0}}+\sigma\pi_{II}) +(1-r_i(\theta))\pi_{\mathit{0}\mathit{0}}\right] d\theta  \; + \\
&  \int_0^1  r_{j}(\theta)\left[(r_i(\theta)+\sigma(1-r_i(\theta)))\pi_{II} +(1-\sigma)(1-r_i(\theta))\pi_{\mathit{0}I}\right] d\theta  -\\&
\int_0^1 r_i(\theta)C(\theta) d\theta.
\end{align*}
Compared to the expected total payoff with financial constraints, firms do not have additional costs from borrowing. Moreover, there is now the possibility that a firm obtains the innovation without innovating itself.
The equilibrium characterization for R\&D competition closely follows the previous analysis. We first implicitly define critical projects similar to those defined previously.
\begin{align*}
C(\theta^{nc}_1) &= (1-\sigma)\pi_{I\mathit{0}}+\sigma\pi_{II} -\pi_{\mathit{0}\mathit{0}}\\
C(\theta^{nc}_2) &= (1-\sigma )(\pi_{II} -\pi_{\mathit{0}I})
\end{align*}
The intuition for $\theta^{nc}_1$ and $\theta^{nc}_2$ is analogous to the one for $\theta_1$ and $\theta_2$, taking into account different payoffs due to potential spillovers. It is straightforward to show that $\theta^{nc}_1 \geq \theta^{nc}_2$.

\begin{lemma}[Investment strategies under competition with spillovers]\label{prop_equilibrium_spill} $ $\\
\noindent (i) The research competition game has multiple equilibria. A profile of double cut-off strategies $(r^*_i,r^*_j)$ is an equilibrium if it satisfies (a) $\theta_L=\theta^{nc}_2$ and $\theta_H=\theta^{nc}_1$  and (b) for each $ \theta \in (\theta^{nc}_2, \theta^{nc}_1)$ either:
\begin{align*}
    r^*_i(\theta)&=1 \text{ and } r^*_{j}(\theta)=0 \text{ or }\\
    r^*_i(\theta)&=0 \text{ and } r^*_{j}(\theta)=1.
\end{align*}
(ii) No other equilibria of the research-competition game exist.
\end{lemma}
Next, we consider the case with RJVs. The analysis is simpler than in the case with financial constraints. The increase in joint profit from a successful innovation is $2\pi_{II}-2\pi_{\mathit{0}\mathit{0}}$. Hence, the RJV invests in all projects up to 
$\theta^{u}$, 
and it does not invest in the remaining ones. We can now prove Proposition \ref{PropCompSpill}.

\subparagraph{Proof of Proposition \ref{PropCompSpill}}
When competition is soft, the argument follows as in the case without spillovers (without relying on Assumption \ref{ass_budget}). When competition is not soft, we need to show that the condition in the proposition is equivalent with the requirement that 
 $\theta^{nc}_1$ < $\theta^{u}$. This follows from simple rearrangements:
\begin{align*}
    \sigma &> 1 - \frac{\pi_{II}-\pi_{\mathit{0}\mathit{0}}}{\pi_{I\mathit{0}}-\pi_{II}}\\
    \sigma(\pi_{I\mathit{0}}-\pi_{II}) &> \pi_{I\mathit{0}}-2\pi_{II} + \pi_{\mathit{0}\mathit{0}}\\
    2\pi_{II}-2\pi_{\mathit{0}\mathit{0}}&>(1-\sigma)\pi_{I\mathit{0}}+\sigma\pi_{II}-\pi_{\mathit{0}\mathit{0}}\\
    C(\theta^{u})&>C(\theta^{nc}_1).
\end{align*}

\subsubsection{Financial constraints}\label{SecAppSpillFC}

We now augment the model with spillovers with financial constraints. The analysis of Sections \ref{SecModel} and \ref{SecEffects} carries over directly if we replace the function $\pi$ with $\Tilde{\pi}$ defined as follows
\begin{align*}
\Tilde{\pi}_{\mathit{0}\mathit{0}} &\equiv \pi_{\mathit{0}\mathit{0}}   \\
\Tilde{\pi}_{II} &\equiv \pi_{II}   \\
\Tilde{\pi}_{I\mathit{0}} &\equiv (1-\sigma)\pi_{I\mathit{0}} + \sigma \pi_{II}  \\
\Tilde{\pi}_{\mathit{0}I} &\equiv (1-\sigma)\pi_{\mathit{0}I} + \sigma \pi_{II}  
\end{align*}

Replacing $\pi$ with $\Tilde{\pi}$, we obtain new expressions for expected profits, $\mathbb{E}\Tilde{\Pi}_i(r_i, r_{j})$, for critical values $\Tilde{\theta}_1, \Tilde{\theta}_2$, etc. We replace Assumption \ref{ass_budget} with 
\begin{assumption}\label{ass_budget_spillFC}
$B < \int_0^{{\Tilde{\theta}_2}}C(\theta) d\theta$.
\end{assumption}
It is straightforward to show that $\Tilde{\theta}_1 \geq \Tilde{\theta}_2$. Moreover, if $\theta_1>0$, then $\theta_1>\Tilde{\theta}_1$ for all $\sigma>0$. Hence, spillover reduces the incentives to invest because rivals could also benefit from the innovation. The following result follows directly from replacing $\pi$ with $\Tilde{\pi}$ in Lemma \ref{prop_equilibrium} and then inserting the above definitions for $\Tilde{\pi}$.

\begin{lemma}
[Investment strategies under competition with spillovers and financial constraints]\label{prop_equilibriumSpillFC} $ $\\
\noindent (i) The research competition game has multiple equilibria. A profile of double-cut off strategies $(r^*_i,r^*_j)$ is an equilibrium if it satisfies (a) $\theta_L=\Tilde{\theta}_2$ and $\theta_H=\Tilde{\theta}_1$ and (b) for each $ \theta \in (\Tilde{\theta}_2, \Tilde{\theta}_1)$ either:
\begin{align*}
 r^*_i(\theta)&=1 \text{ and } r^*_{j}(\theta)=0 \text{ or }\\
    r^*_i(\theta)&=0 \text{ and } r^*_{j}(\theta)=1.
\end{align*}
(ii) No other pure-strategy equilibria of the research-competition game exist.
\end{lemma}

Suppose now that the two firms form an RJV. Since the firms will share a successful innovation, spillovers do not affect innovation behavior under cooperative R\&D. Therefore, an RJV still has the critical projects $\theta^\rho$ and $\theta^u$ and invests according to Lemma \ref{prop_RJVInvestment}. 

For the comparison between RJV and R\&D competition, we replace $\theta_1$ and $\pi$ in the definitions of  $\bar{\rho}$ and $\bar{B}(\rho)$ with $\Tilde{\theta}_1$ and $\Tilde{{\pi}}$ to obtain:
\begin{align*}
\Tilde{{B}}(\rho) & = \dfrac{\int_0^{\Tilde{\theta}_1} C(\theta)d\theta}{2}\\
    \Tilde{\rho} & =
    \begin{cases}
    \dfrac{\Tilde{\pi}_{I\mathit{0}}-\Tilde{\pi}_{II}-(\Tilde{\pi}_{II}-\Tilde{\pi}_{\mathit{0}\mathit{0}})}{2(\Tilde{\pi}_{II}-\Tilde{\pi}_{\mathit{0}\mathit{0}})}, & \text{for } \Tilde{\pi}_{II} > \Tilde{\pi}_{\mathit{0}\mathit{0}} \\
    \infty, & \text{for }\Tilde{\pi}_{II} = \Tilde{\pi}_{\mathit{0}\mathit{0}}.
  \end{cases} 
\end{align*}

It is straightforward to show that $\bar{\rho}>\Tilde{\rho}$ and $\bar{B}(\rho)> \Tilde{{B}}(\rho)$. Replacing $\pi$ with $\Tilde{\pi}$ in Proposition \ref{prop_RJVEffect} and then inserting the values for $\pi_{t_it_j}$ into the definitions of $\Tilde{\pi}_{t_i,t_j}$ immediately shows under which circumstances an RJV increases innovation with spillovers.

\begin{proposition}[Comparison of competition and RJV with spillovers]\label{prop_RJVEffect_spill} $ $
\begin{enumerate}[label=(\roman*)]
\item Suppose $2\pi_{II}-2\pi_{\mathit{0}\mathit{0}} > \pi_{I\mathit{0}}-\pi_{\mathit{0}\mathit{0}}-\sigma(\pi_{I\mathit{0}}-\pi_{II})$. Then the innovation probability is strictly larger under the RJV than under R\&D competition.
\item Suppose $2\pi_{II}-2\pi_{\mathit{0}\mathit{0}} < \pi_{I\mathit{0}}-\pi_{\mathit{0}\mathit{0}}-\sigma(\pi_{I\mathit{0}}-\pi_{II})$. Then: \newline (a) The innovation probability is strictly larger under the RJV than under competition if and only if $B > \Tilde{B}(\rho)$ and $\rho > \Tilde{\rho}$.\newline
(b) If the formation of the RJV strictly increases the innovation probability, then it  weakly decreases total R\&D spending.
\end{enumerate}
\end{proposition}
 Moreover, the conditions on the budget and interest rate that an RJV increases the probability of innovation are weaker with higher spillovers (see Section \ref{SecAppProofSpillFC}).

\subsubsection{Proof of Proposition \ref{prop_RJV_better} }\label{SecAppProofSpillFC}

First, we note two auxiliary results which are analogous to Lemmas \ref{lemma_theta_rho_theta_1} and \ref{lemma_theta_u_theta_1}, replacing $\pi$ with $\Tilde{\pi} $ and $\theta_1$ with $\Tilde{\theta}_1 $.

\begin{lemma}\label{lemma_theta_rho_theta_1_spill}
$\pi_{II}>(1-\sigma)(\pi_{I\mathit{0}}-\pi_{II})+\pi_{\mathit{0}\mathit{0}} \Leftrightarrow \theta^\rho > \Tilde{\theta}_1$.
\end{lemma}

\begin{lemma}\label{lemma_theta_u_theta_1_spill}
$\rho > \Tilde{\rho} \Leftrightarrow  \theta^u > \Tilde{\theta}_1$. 
\end{lemma}

Next, we provide a useful monotonicity result:
\begin{lemma}\label{lemma_theta_1_sigma_rho}
$\Tilde{\theta}_1$ is a weakly decreasing function of $\sigma$ and $\rho$. $\Tilde{B}$ and $\Tilde{\rho}$ are weakly decreasing in $\sigma$.
\end{lemma}
\begin{proof}
Suppose $\sigma'\geq\sigma$ and $\rho'\geq\rho$. Then
\begin{align*}
\frac{\pi_{I\mathit{0}}-\sigma(\pi_{I\mathit{0}}-\pi_{II}) -\pi_{\mathit{0}\mathit{0}}}{1+\rho} &\geq \frac{\pi_{I\mathit{0}}-\sigma'(\pi_{I\mathit{0}}-\pi_{II}) -\pi_{\mathit{0}\mathit{0}}}{1+\rho'}\\
\frac{(1-\sigma)\pi_{I\mathit{0}}+\sigma\pi_{II} -\pi_{\mathit{0}\mathit{0}}}{1+\rho}&\geq\frac{(1-\sigma')\pi_{I\mathit{0}}+\sigma'\pi_{II} -\pi_{\mathit{0}\mathit{0}}}{1+\rho'}
\end{align*}
where the inequality holds since $\pi_{I\mathit{0}}\geq\pi_{II}$. The first result immediately follows.
Next, since $C(\theta)$ is a strictly increasing function, $\Tilde{B} = \dfrac{\int_0^{\Tilde{\theta}_1} C(\theta)d\theta}{2}$ must also be weakly decreasing in $\sigma$.
The interest rate cut-off value $\Tilde{\rho}$ is decreasing in $\sigma$, since 
\begin{align*}
    \dfrac{\partial \Tilde{\rho}}{\partial\sigma}= \frac{\pi_{II}-\pi_{I\mathit{0}}}{2(\pi_{II}-\pi_{\mathit{0}\mathit{0}})}\leq 0,
\end{align*}
if $\pi_{II}-\pi_{\mathit{0}\mathit{0}}>0$ and zero otherwise.
\end{proof}

To prove Proposition \ref{prop_RJV_better}, suppose first that we have weak competition in the sense that
$\pi_{II}>(1-\sigma)(\pi_{I\mathit{0}}-\pi_{II})+\pi_{\mathit{0}\mathit{0}}$. By Lemma \ref{lemma_theta_rho_theta_1_spill}, this implies $\theta^\rho > \Tilde{\theta}_1$. By Lemma \ref{prop_RJVInvestment}, the probability that the RJV innovates is at least $\theta^{\rho}$. By Lemma \ref{prop_equilibriumSpillFC}, the probability of innovation under R\&D competition is $\Tilde{\theta}_1$, where this expression is decreasing in $\sigma$ by Lemma \ref{lemma_theta_1_sigma_rho}. 
Now, suppose  $\pi_{II}\leq(1-\sigma)(\pi_{I\mathit{0}}-\pi_{II})+\pi_{\mathit{0}\mathit{0}}$. 
By Proposition \ref{prop_RJVEffect_spill}(ii)(a), a strictly larger innovation probability under the RJV implies $B > \Tilde{B}(\rho)$ and $\rho > \Tilde{\rho}$. 
Further, arguing as in the proof of Lemma \ref{lemma_Prop3ii}, if the innovation probability is strictly larger under the RJV than under R\&D competition, then $\theta^{B}>\theta^\rho$. 
If $\theta^{B}\in(\theta^\rho, \theta^u)$, the RJV invests in all projects in the set $(0,\theta^{B})$ and discovers the innovation with probability $\theta^{B}$. Since $B> \Tilde{B}(\rho)=\int_0^{\Tilde{\theta}_1} C(\theta)d\theta/2$, we have $\int_0^{\theta^{B}} C(\theta)d\theta >\int_0^{\Tilde{\theta}_1} C(\theta)d\theta$, which implies $\theta^{B}>\Tilde{\theta}_1$. Without the RJV, in any equilibrium, the firms invest in projects in the set $(0, \Tilde{\theta}_1)$ and the innovation is discovered with probability  $\Tilde{\theta}_1$. As $\Tilde{\theta}_1$ is weakly decreasing in $\rho$ and weakly decreasing in $\sigma$ by Lemma \ref{lemma_theta_1_sigma_rho}, it immediately follows that, if the probability of innovation is strictly larger under the RJV for any $\sigma$ and $\rho$, then this is also true for any $\sigma'\geq\sigma$ and $\rho'\geq\rho$.
If $\theta^{B}\geq \theta^u $, then the RJV invests in all projects in the set $(0, \theta^u )$  and discovers the innovation with probability $\theta^u$. By Lemma \ref{lemma_theta_u_theta_1_spill}, $\rho > \Tilde{\rho}$ implies $\theta^u >\Tilde{\theta}_1$. It immediately follows that, if the probability of innovation is strictly larger under the RJV for any $\sigma$ and $\rho$, then this is also true for $\sigma'\geq\sigma$ and $\rho'\geq\rho$.
\qed

\subsubsection{Licensing}\label{SecAppLicense}

We now add some more details to the licensing model sketched in Section \ref{sec_license}, where a successful innovator chooses a two-part tariff licensing contract $(L,\eta)$ at which the unsuccessful innovator can use the innovation. The buyer accepts any contract that yields at least the outside option of $\pi_{\mathit{0}I}$. In equilibrium, the innovator extracts all rents and sets a fixed fee $L$ such that the unsuccessful firm earns $\pi_{\mathit{0}I}$. Therefore, the single innovator receives the total market surplus net of the outside option, $2\pi_{II}+\Delta-\pi_{\mathit{0}I}$.
We spell out the profit function $\pi^L$ as 

\begin{align*}
\pi^L_{00} &\equiv \pi_{\mathit{0}\mathit{0}}   \\
\pi^L_{II} &\equiv \pi_{II}   \\
\pi^L_{I0} &= \max \{\pi_{I\mathit{0}},2\pi_{II}+\Delta-\pi_{\mathit{0}I}\} \\
\pi^L_{0I} &= \pi_{\mathit{0}I}.
\end{align*}

Replacing $\pi$ with $\pi^L$, we obtain a new expression for expected profits, $\mathbb{E}\Pi^L_i(r_i, r_{j})$, and critical values $\theta^L_1, \theta^L_2$, etc. We maintain Assumption \ref{ass_budget}.
It is straightforward to show that $\theta^L_1 \geq \theta_1$ and $\theta^L_2=\theta_2$. The next result follows directly from replacing $\pi$ with $\pi^L$ in Lemma \ref{prop_equilibrium} and then inserting the above definitions for $\pi^L$.

\begin{lemma}
[Investment strategies under competition with licensing]\label{prop_equilibrium_license} $ $\\
Suppose that $2\pi_{II}+\Delta-\pi_{\mathit{0}I} \geq \pi_{I\mathit{0}}$. Then:\\
\noindent (i) The research competition game with licensing has multiple equilibria. A profile of double-cut off strategies $(r^*_i,r^*_j)$ is an equilibrium if it satisfies (a) $\theta_L=\theta_2$ and $\theta_H=\theta^L_1$ and (b) for each $ \theta \in (\theta_2, \theta^L_1)$ either:
\begin{align*}
    r^*_i(\theta)&=1 \text{ and } r^*_{j}(\theta)=0 \text{ or }\\
    r^*_i(\theta)&=0 \text{ and } r^*_{j}(\theta)=1.
\end{align*}
(ii) No other equilibria of the research-competition game exist.
\end{lemma}

The analysis of the RJV is unchanged; it invests according to Lemma \ref{prop_RJVInvestment}. 

Define the budget threshold $\bar{B}^L(\rho)$ and the interest threshold  $\bar{\rho}^L$ as 
\begin{align*}
    \bar{\rho}^L & =
    \begin{cases}
    \dfrac{\pi_{\mathit{0}\mathit{0}}+\Delta-\pi_{\mathit{0}I}}{2(\pi_{II}-\pi_{\mathit{0}\mathit{0}})}, & \text{for } \pi_{II} > \pi_{\mathit{0}\mathit{0}} \\
    \infty, & \text{for } \pi_{II} = \pi_{\mathit{0}\mathit{0}}.
  \end{cases}\\
  \bar{B}^L(\rho) & = \dfrac{\int_0^{\theta^L_1} C(\theta)d\theta}{2} 
\end{align*}

With this notation in place, it is straightforward to see how Proposition \ref{prop_RJVEffect_licensing} directly follows by reformulation of Proposition \ref{prop_RJVEffect} with $\pi$ replaced by $\pi^L$.

\subsection{Multiple firms} \label{SecAppMultiplefirms}
\subsubsection{Industry-wide RJV}\label{SecAppIndustrywide}
We extend the model to three ex-ante symmetric firms $(i\in\{1,2,3\})$. The product market profits of firm $i$ are now given in the reduced form $\pi_{t_it_jt_k}$ for $j,k \neq i$, $ j \neq k$.
We suppose profits are symmetric: $\pi_{\mathit{0}I\mathit{0}}=\pi_{\mathit{0}\mathit{0}I}$ and $\pi_{II\mathit{0}}=\pi_{I\mathit{0}I}$. That is, only the number of successful rivals matters.
We adjust the assumptions on the product market profits accordingly.
\begin{assumption}[Regularity of market profit functions]\label{ass_r_3firms}
$ $
\begin{enumerate}[label=(\roman*)]
\item Profits are non-negative: $\pi_{t_it_jt_k} \geq 0$ for all $t_i$, $t_j$ and $t_k$.
\item Innovation increases profits: $\pi_{III} \geq \pi_{\mathit{0}\mathit{0}\mathit{0}}$.
\item Competitor innovation reduces profits: $\pi_{\mathit{0}\mathit{0}\mathit{0}} \geq \pi_{\mathit{0}I\mathit{0}} \geq \pi_{\mathit{0}II}$.
\item Competitor innovations reduce the value of own innovations:\\
$\pi_{I\mathit{0}\mathit{0}} - \pi_{\mathit{0}\mathit{0}\mathit{0}}\geq \pi_{II\mathit{0}}- \pi_{\mathit{0}I\mathit{0}}  \geq \pi_{III}-\pi_{\mathit{0}II}$.
\end{enumerate}
\end{assumption}

We obtain cut-off values $\theta_3 \leq\theta_2 \leq \theta_1$ from 
\begin{align*}
(1+\rho)C(\theta_1)&=\pi_{I\mathit{0}\mathit{0}}-\pi_{\mathit{0}\mathit{0}\mathit{0}} \\
(1+\rho)C(\theta_2)&=\pi_{II\mathit{0}} - \pi_{\mathit{0}I\mathit{0}}\\
(1+\rho)C(\theta_3)&=\pi_{III} - \pi_{\mathit{0}II}.
\end{align*}

After appropriately modifying Assumption \ref{ass_budget}, we find that all equilibria have a triple cut-off structure with all firms investing in $[0,\theta_3)$, two firms in $(\theta_3,\theta_2)$, one firm in $(\theta_2,\theta_1)$ and no firm investing in $(\theta_1,1)$

Now we suppose that all three firms form an RJV.
Let $\theta^{B}$ be defined as the solution to $\int_0^{\theta^{B}} C(\theta) d\theta=3B$ if $\int_0^1 C(\theta) d\theta>3B$ and $\theta^{B}=1$ otherwise.
Next, let $\theta^u$ and $\theta^\rho$ be the solutions to the following equations
\begin{align*}
(1+\rho)C(\theta^\rho)&=3(\pi_{III}-\pi_{\mathit{0}\mathit{0}\mathit{0}}) \\
C(\theta^u)&=3(\pi_{III}-\pi_{\mathit{0}\mathit{0}\mathit{0}}).
\end{align*}

Using these cut-off values, the RJV follows a cut-off strategy as in Lemma \ref{prop_RJVInvestment}. Defining soft competition by the requirement that  $3\pi_{III}>\pi_{I\mathit{0}\mathit{0}}+2\pi_{\mathit{0}\mathit{0}\mathit{0}}$ and adjusting the budget cut-off value $\bar{B}$ and the interest rate cut-off value $\bar{\rho}$ appropriately, we finally obtain conditions under which an RJV increases the probability of innovation, which are analogous to those in  Proposition \ref{prop_RJVEffect}.

\subsubsection{Multiple RJVs}

Next, we extend the model to four ex-ante symmetric firms $(i\in\{1,2,3,4\})$. We write the product market profits of firm $i$ facing competitors $j$, $k$ and $l$ as $\pi_{t_it_jt_kt_{\ell}}$. As in the previous subsection, we assume that profits depend only on the own technology and the number of competitors with the new technology, not on their identity. Further, we impose the regularity conditions that profits are non-negative, weakly increasing in own innovation and that the positive effect of own innovation decreases in the number of competitors with access to the new technology.

We again adjust Assumption \ref{ass_budget} so that firms want to borrow externally under R\&D competition. Unsurprisingly, it turns out that, under R\&D competition these equilibria have four cut-off values $\theta_4 \leq\theta_3 \leq\theta_2 \leq \theta_1$, defined in the by now familiar way.

Now we suppose that two RJVs are formed, each consisting of two firms. Thus, instead of four firms, we have two competing RJVs $\{v_1,v_2\}$, each with budget $2B$. Let $\theta^{B}$ be defined as the solution to $\int_0^{\theta^{B}} C(\theta) d\theta=2B$ if $\int_0^1 C(\theta) d\theta>2B$ and $\theta^{B}=1$ otherwise. To find the cutoff-values, consider the equations
\begin{align*}
(1+\rho)C(\theta_1^{\rho})&=2(\pi_{II\mathit{0}\mathit{0}}-\pi_{\mathit{0}\mathit{0}\mathit{0}\mathit{0}}) \\
(1+\rho)C(\theta_2^{\rho})&=2(\pi_{IIII} - \pi_{\mathit{0}\mathit{0}II})\\
C(\theta_1^u)&=2(\pi_{II\mathit{0}\mathit{0}}-\pi_{\mathit{0}\mathit{0}\mathit{0}\mathit{0}})\\
C(\theta_2^u)&=2(\pi_{IIII} - \pi_{\mathit{0}\mathit{0}II}).
\end{align*}
The interpretation is the same as with one RJV. We restrict our analysis to the case in which the budget of an RJV is sufficiently large such that no RJV borrows in equilibrium. 
\begin{assumption}\label{ass_budget_4RJV}
$2B > \int_0^{\theta_1^{\rho}}C(\theta) d\theta$.
\end{assumption}
The assumptions imply $\theta^{B}>\theta_1^{\rho}\geq\theta_2^{\rho}$.
How $\theta^{B}$ relates to the two values $\theta_2^u\leq\theta_1^u$ will determine the  optimal portfolio of an RJV. The research competition game turns out to have multiple equilibria with double cut-offs (and no other equilibria). 

The proof follows a similar structure as in Lemma \ref{prop_equilibrium}, but we have to distinguish between the three cases $\theta^{B}<\theta_2^u$, $\theta^{B}\in[\theta_2^u,\theta_1^u)$ and $\theta_1^u\leq\theta^{B}$.
Further, Assumption \ref{ass_budget_4RJV} implies $(1+\rho)C(\theta)>2(\pi_{II\mathit{0}\mathit{0}}-\pi_{\mathit{0}\mathit{0}\mathit{0}\mathit{0}})$ for any $\theta>\theta^{B}$. Thus, it is never optimal to invest more than the available budget for an RJV.
Defining soft competition by the requirement that  $\pi_{II\mathit{0}\mathit{0}}-\pi_{\mathit{0}\mathit{0}\mathit{0}\mathit{0}}>
\pi_{I\mathit{0}\mathit{0}\mathit{0}}-\pi_{II\mathit{0}\mathit{0}}$ and adjusting the budget cut-off value $\bar{B}$ and the interest rate cut-off value $\bar{\rho}$ appropriately, we finally obtain conditions under which the formation of two RJVs increases the probability of innovation, which are analogous to those in  Proposition \ref{prop_RJVEffect}.

\newpage
\subsection{Sources for RJV Examples}\label{appendix_sources}

In the Introduction, we mentioned several actual RJVs. More information about these ventures can be found at the following links, which are listed in the order in which the RJV appeared in text. All links were last accessed on June 28, 2022 and are archived on \href{https://web.archive.org}{https://web.archive.org}.

\begin{itemize}
    \item \href{https://www.cbo.gov/publication/57126}{https://www.cbo.gov/publication/57126}
   \item \href{https://www.bdo.co.uk/en-gb/news/2021/top-20-global-carmakers-spend-another-71-7bn-on-r-and-d-as-electric-vehicle-rollout-gathers-pace}{https://www.bdo.co.uk/en-gb/news/2021/top-20-global-carmakers-spend-another-\\71-7bn-on-r-and-d-as-electric-vehicle-rollout-gathers-pace}
    \item \href{https://group-media.mercedes-benz.com/marsMediaSite/de/instance/ko.xhtml?oid=42917172}{https://group-media.mercedes-benz.com/marsMediaSite/de/instance/ko.xhtml?oid\\=42917172}
    \item \href{https://www.saftbatteries.com/media-resources/press-releases/psa-a-total-automotive-cells-company}{https://www.saftbatteries.com/media-resources/press-releases/psa-a-total-\\automotive-cells-company}
    \item \href{https://www.bp.com/en/global/corporate/news-and-insights/press-releases/paving-the-way-for-sustainable-mobility-bp-bmw-daimler-announce-bp-third-shareholder-of-dcs.html}{https://www.bp.com/en/global/corporate/news-and-insights/press-releases/\\paving-the-way-for-sustainable-mobility-bp-bmw-daimler-announce-bp-third\\-shareholder-of-dcs.html}
    \item \href{https://www.reuters.com/business/autos-transportation/renault-nissan-mitsubishi-alliance-say-deepen-cooperations-ev-production-2022-01-27/}{https://www.reuters.com/business/autos-transportation/renault-nissan-mitsubishi-\\alliance-say-deepen-cooperations-ev-production-2022-01-27/}
    \item \href{https://www.basf.com/global/de/media/news-releases/2021/05/p-21-215.html}{https://www.basf.com/global/de/media/news-releases/2021/05/p-21-215.html}
    \item \href{https://www.volvogroup.com/en/news-and-media/news/2020/apr/news-3640568.html}{https://www.volvogroup.com/en/news-and-media/news/2020/apr/news-\\3640568.html}
    \item\href{https://www.washingtonpost.com/technology/2022/05/18/solid-state-batteries-electric-vehicles-race/}{https://www.washingtonpost.com/technology/2022/05/18/solid-state-batteries-\\electric-vehicles-race/}
    \item \href{https://www.forbes.com/sites/greggardner/2020/02/03/toyota-and-panasonic-launch-joint-ev-battery-venture/}{https://www.forbes.com/sites/greggardner/2020/02/03/toyota-and-panasonic-\\launch-joint-ev-battery-venture/}
    \item \href{https://www.volkswagen-newsroom.com/en/press-releases/enel-x-and-volkswagen-team-up-for-electric-mobility-in-italy-7315}{https://www.volkswagen-newsroom.com/en/press-releases/enel-x-and-volkswagen-\\team-up-for-electric-mobility-in-italy-7315}
    \item \href{https://www.media.stellantis.com/em-en/fca-archive/press/fiat-chrysler-automobiles-and-engie-eps-plan-to-join-forces-in-a-jv-creating-a-leading-company-in-the-e-mobility-sector}{https://www.media.stellantis.com/em-en/fca-archive/press/fiat-chrysler-\\automobiles-and-engie-eps-plan-to-join-forces-in-a-jv-creating-a-leading-company-\\in-the-e-mobility-sector}
\end{itemize}

\newpage
\bibliographystyle{ecta}
\bibliography{references_RJV}

\begin{thebibliography}{58}
\newcommand{\enquote}[1]{``#1''}
\expandafter\ifx\csname natexlab\endcsname\relax\def\natexlab#1{#1}\fi

\bibitem[\protect\citeauthoryear{Amir, Liu, Machowska, and Resende}{Amir
  et~al.}{2019}]{amir2019spillovers}
\textsc{Amir, R., H.~Liu, D.~Machowska, and J.~Resende} (2019):
  \enquote{Spillovers, subsidies, and second-best socially optimal R\&D,}
  \emph{Journal of Public Economic Theory}, 21, 1200--1220.

\bibitem[\protect\citeauthoryear{Aschhoff and Schmidt}{Aschhoff and
  Schmidt}{2008}]{aschhoff2008empirical}
\textsc{Aschhoff, B. and T.~Schmidt} (2008): \enquote{Empirical evidence on the
  success of R\&D cooperation -- happy together?} \emph{Review of Industrial
  Organization}, 33, 41--62.

\bibitem[\protect\citeauthoryear{Bagwell and Staiger}{Bagwell and
  Staiger}{1994}]{bagwell1994sensitivity}
\textsc{Bagwell, K. and R.~W. Staiger} (1994): \enquote{The sensitivity of
  strategic and corrective R\&D policy in oligopolistic industries,}
  \emph{Journal of International Economics}, 36, 133--150.

\bibitem[\protect\citeauthoryear{Bardey, Jullien, and Lozachmeur}{Bardey
  et~al.}{2016}]{bardey2016}
\textsc{Bardey, D., B.~Jullien, and J.-M. Lozachmeur} (2016): \enquote{Health
  insurance and diversity of treatment,} \emph{Journal of Health Economics},
  47, 50--63.

\bibitem[\protect\citeauthoryear{Bavly, Heller, and Schreiber}{Bavly
  et~al.}{2022}]{bavly2020}
\textsc{Bavly, G., Y.~Heller, and A.~Schreiber} (2022): \enquote{Social welfare
  in search games with asymmetric information,} \emph{Journal of Economic
  Theory}, 202, 105462.

\bibitem[\protect\citeauthoryear{Becker and Dietz}{Becker and
  Dietz}{2004}]{becker2004r}
\textsc{Becker, W. and J.~Dietz} (2004): \enquote{R\&D cooperation and
  innovation activities of firms -- evidence for the German manufacturing
  industry,} \emph{Research Policy}, 33, 209--223.

\bibitem[\protect\citeauthoryear{Boone}{Boone}{2008{\natexlab{a}}}]{boone2008competition}
\textsc{Boone, J.} (2008{\natexlab{a}}): \enquote{Competition: Theoretical
  parameterizations and empirical measures,} \emph{Journal of Institutional and
  Theoretical Economics (JITE)/Zeitschrift f{\"u}r die gesamte
  Staatswissenschaft}, 587--611.

\bibitem[\protect\citeauthoryear{Boone}{Boone}{2008{\natexlab{b}}}]{boone2008new}
---\hspace{-.1pt}---\hspace{-.1pt}--- (2008{\natexlab{b}}): \enquote{A new way
  to measure competition,} \emph{The Economic Journal}, 118, 1245--1261.

\bibitem[\protect\citeauthoryear{Bourreau, Jullien, and Lefouili}{Bourreau
  et~al.}{2021}]{bourreau2018b}
\textsc{Bourreau, M., B.~Jullien, and Y.~Lefouili} (2021): \enquote{Mergers and
  demand-enhancing innovation,} \emph{TSE Working Papers No. 18-907}.

\bibitem[\protect\citeauthoryear{Brunner, Letina, and Schmutzler}{Brunner
  et~al.}{2022}]{brunner2022research}
\textsc{Brunner, P., I.~Letina, and A.~Schmutzler} (2022): \enquote{Research
  joint ventures: The role of financial constraints,} \emph{arXiv preprint
  arXiv:2207.04856v1}.

\bibitem[\protect\citeauthoryear{Bryan and Lemus}{Bryan and
  Lemus}{2017}]{bryan2017}
\textsc{Bryan, K.~A. and J.~Lemus} (2017): \enquote{The direction of
  innovation,} \emph{Journal of Economic Theory}, 172, 247--272.

\bibitem[\protect\citeauthoryear{Caggese}{Caggese}{2019}]{caggese2019financing}
\textsc{Caggese, A.} (2019): \enquote{Financing constraints, radical versus
  incremental innovation, and aggregate productivity,} \emph{American Economic
  Journal: Macroeconomics}, 11, 275--309.

\bibitem[\protect\citeauthoryear{Caloghirou, Ioannides, and
  Vonortas}{Caloghirou et~al.}{2003}]{caloghirou2003research}
\textsc{Caloghirou, Y., S.~Ioannides, and N.~S. Vonortas} (2003):
  \enquote{Research joint ventures,} \emph{Journal of Economic Surveys}, 17,
  541--570.

\bibitem[\protect\citeauthoryear{Cassiman}{Cassiman}{2000}]{cassiman2000research}
\textsc{Cassiman, B.} (2000): \enquote{Research joint ventures and optimal R\&D
  policy with asymmetric information,} \emph{International Journal of
  Industrial Organization}, 18, 283--314.

\bibitem[\protect\citeauthoryear{Cassiman and Veugelers}{Cassiman and
  Veugelers}{2002}]{cassiman2002r}
\textsc{Cassiman, B. and R.~Veugelers} (2002): \enquote{R\&D cooperation and
  spillovers: some empirical evidence from Belgium,} \emph{American Economic
  Review}, 92, 1169--1184.

\bibitem[\protect\citeauthoryear{Choi}{Choi}{1993}]{choi1993cooperative}
\textsc{Choi, J.~P.} (1993): \enquote{Cooperative R\&D with product market
  competition,} \emph{International Journal of Industrial Organization}, 11,
  553--571.

\bibitem[\protect\citeauthoryear{Czarnitzki and Hottenrott}{Czarnitzki and
  Hottenrott}{2011}]{czarnitzki2011r}
\textsc{Czarnitzki, D. and H.~Hottenrott} (2011): \enquote{R\&D investment and
  financing constraints of small and medium-sized firms,} \emph{Small Business
  Economics}, 36, 65--83.

\bibitem[\protect\citeauthoryear{d'Aspremont and Jacquemin}{d'Aspremont and
  Jacquemin}{1988}]{d1988cooperative}
\textsc{d'Aspremont, C. and A.~Jacquemin} (1988): \enquote{Cooperative and
  noncooperative R\&D in duopoly with spillovers,} \emph{American Economic
  Review}, 78, 1133--1137.

\bibitem[\protect\citeauthoryear{Denicol{\`o} and Polo}{Denicol{\`o} and
  Polo}{2018}]{denicolo2018a}
\textsc{Denicol{\`o}, V. and M.~Polo} (2018): \enquote{Duplicative research,
  mergers and innovation,} \emph{Economics Letters}, 166, 56--59.

\bibitem[\protect\citeauthoryear{Duso, R{\"o}ller, and Seldeslachts}{Duso
  et~al.}{2014}]{duso2014collusion}
\textsc{Duso, T., L.-H. R{\"o}ller, and J.~Seldeslachts} (2014):
  \enquote{Collusion through joint R\&D: An empirical assessment,} \emph{Review
  of Economics and Statistics}, 96, 349--370.

\bibitem[\protect\citeauthoryear{Falvey, Poyago-Theotoky, and
  Teerasuwannajak}{Falvey et~al.}{2013}]{falvey2013coordination}
\textsc{Falvey, R., J.~Poyago-Theotoky, and K.~T. Teerasuwannajak} (2013):
  \enquote{Coordination costs and research joint ventures,} \emph{Economic
  Modelling}, 33, 965--976.

\bibitem[\protect\citeauthoryear{Farrell and Shapiro}{Farrell and
  Shapiro}{2000}]{farrell2000scale}
\textsc{Farrell, J. and C.~Shapiro} (2000): \enquote{Scale economies and
  synergies in horizontal merger analysis,} \emph{Antitrust LJ}, 68, 685.

\bibitem[\protect\citeauthoryear{Fauli-Oller and Sandonis}{Fauli-Oller and
  Sandonis}{2003}]{fauli2003merge}
\textsc{Fauli-Oller, R. and J.~Sandonis} (2003): \enquote{To merge or to
  license: implications for competition policy,} \emph{International Journal of
  Industrial Organization}, 21, 655--672.

\bibitem[\protect\citeauthoryear{Federico, Langus, and Valletti}{Federico
  et~al.}{2017}]{federico2017}
\textsc{Federico, G., G.~Langus, and T.~Valletti} (2017): \enquote{A simple
  model of mergers and innovation,} \emph{Economics Letters}, 157, 136--140.

\bibitem[\protect\citeauthoryear{Federico, Langus, and Valletti}{Federico
  et~al.}{2018}]{federico2018}
---\hspace{-.1pt}---\hspace{-.1pt}--- (2018): \enquote{Horizontal mergers and
  product innovation,} \emph{International Journal of Industrial Organization},
  59, 1--23.

\bibitem[\protect\citeauthoryear{Fumagalli, Motta, and Tarantino}{Fumagalli
  et~al.}{2022}]{fumagalli2022}
\textsc{Fumagalli, C., M.~Motta, and E.~Tarantino} (2022): \enquote{Shelving or
  developing? The acquisition of potential competitors under financial
  constraints,} \emph{CSEW Working Paper no. 637}.

\bibitem[\protect\citeauthoryear{Gilbert}{Gilbert}{2019}]{Gilbert2019}
\textsc{Gilbert, R.} (2019): \enquote{Competition, mergers, and R\&D
  diversity,} \emph{Review of Industrial Organization}, 465--484.

\bibitem[\protect\citeauthoryear{Goyal and Moraga-Gonzalez}{Goyal and
  Moraga-Gonzalez}{2001}]{goyal2001r}
\textsc{Goyal, S. and J.~L. Moraga-Gonzalez} (2001): \enquote{R\&d networks,}
  \emph{Rand Journal of Economics}, 686--707.

\bibitem[\protect\citeauthoryear{Grossman and Shapiro}{Grossman and
  Shapiro}{1986}]{grossman1986research}
\textsc{Grossman, G.~M. and C.~Shapiro} (1986): \enquote{Research joint
  ventures: an antitrust analysis,} \emph{JL Econ. \& Org.}, 2, 315.

\bibitem[\protect\citeauthoryear{Hall and Lerner}{Hall and
  Lerner}{2010}]{hall2010financing}
\textsc{Hall, B.~H. and J.~Lerner} (2010): \enquote{The financing of R\&D and
  innovation,} in \emph{Handbook of the Economics of Innovation}, Elsevier,
  vol.~1, 609--639.

\bibitem[\protect\citeauthoryear{Harsanyi and Selten}{Harsanyi and
  Selten}{1988}]{harsanyi1988general}
\textsc{Harsanyi, J.~C. and R.~Selten} (1988): \emph{A General Theory of
  Equilibrium Selection in Games}, Cambridge, Massachusetts: MIT Press Books.

\bibitem[\protect\citeauthoryear{Howell}{Howell}{2017}]{howell2017financing}
\textsc{Howell, S.~T.} (2017): \enquote{Financing innovation: Evidence from
  R\&D grants,} \emph{American Economic Review}, 107, 1136--64.

\bibitem[\protect\citeauthoryear{Jacquemin}{Jacquemin}{1988}]{jacquemin1988cooperative}
\textsc{Jacquemin, A.} (1988): \enquote{Cooperative agreements in R\&D and
  European antitrust policy,} \emph{European Economic Review}, 32, 551--560.

\bibitem[\protect\citeauthoryear{Kamien, Muller, and Zang}{Kamien
  et~al.}{1992}]{kamien1992research}
\textsc{Kamien, M.~I., E.~Muller, and I.~Zang} (1992): \enquote{Research joint
  ventures and R\&D cartels,} \emph{American Economic Review}, 1293--1306.

\bibitem[\protect\citeauthoryear{Kamien and Zang}{Kamien and
  Zang}{2000}]{kamien2000meet}
\textsc{Kamien, M.~I. and I.~Zang} (2000): \enquote{Meet me halfway: Research
  joint ventures and absorptive capacity,} \emph{International Journal of
  Industrial Organization}, 18, 995--1012.

\bibitem[\protect\citeauthoryear{Katz}{Katz}{1986}]{katz1986analysis}
\textsc{Katz, M.~L.} (1986): \enquote{An analysis of cooperative research and
  development,} \emph{RAND Journal of Economics}, 527--543.

\bibitem[\protect\citeauthoryear{Katz and Shapiro}{Katz and
  Shapiro}{1985}]{katzshapiro1985}
\textsc{Katz, M.~L. and C.~Shapiro} (1985): \enquote{On the licensing of
  innovations,} \emph{RAND Journal of Economics}, 16, 504--520.

\bibitem[\protect\citeauthoryear{Kerr and Nanda}{Kerr and
  Nanda}{2015}]{kerr2015financing}
\textsc{Kerr, W.~R. and R.~Nanda} (2015): \enquote{Financing innovation,}
  \emph{Annual Review of Financial Economics}, 7, 445--462.

\bibitem[\protect\citeauthoryear{Krieger, Li, and Papanikolaou}{Krieger
  et~al.}{2022}]{krieger2017developing}
\textsc{Krieger, J., D.~Li, and D.~Papanikolaou} (2022): \enquote{Missing
  novelty in drug development,} \emph{Review of Financial Studies}, 35,
  636--679.

\bibitem[\protect\citeauthoryear{Leahy and Neary}{Leahy and
  Neary}{1997}]{leahy1997public}
\textsc{Leahy, D. and J.~P. Neary} (1997): \enquote{Public policy towards R\&D
  in oligopolistic industries,} \emph{American Economic Review}, 642--662.

\bibitem[\protect\citeauthoryear{Letina}{Letina}{2016}]{letina2016}
\textsc{Letina, I.} (2016): \enquote{The road not taken: competition and the
  R\&D portfolio,} \emph{RAND Journal of Economics}, 47, 433--460.

\bibitem[\protect\citeauthoryear{Letina and Schmutzler}{Letina and
  Schmutzler}{2019}]{letina2019inducing}
\textsc{Letina, I. and A.~Schmutzler} (2019): \enquote{Inducing variety: A
  theory of innovation contests,} \emph{International Economic Review}, 60,
  1757--1780.

\bibitem[\protect\citeauthoryear{Letina, Schmutzler, and Seibel}{Letina
  et~al.}{forthcoming}]{letina2020start}
\textsc{Letina, I., A.~Schmutzler, and R.~Seibel} (forthcoming):
  \enquote{Killer acquisitions and beyond: policy effects on innovation
  strategies,} \emph{International Economic Review}.

\bibitem[\protect\citeauthoryear{Link}{Link}{1998}]{link1998case}
\textsc{Link, A.~N.} (1998): \enquote{Case study of R\&D efficiency in an ATP
  joint venture,} \emph{The Journal of Technology Transfer}, 23, 43--51.

\bibitem[\protect\citeauthoryear{Mancusi and Vezzulli}{Mancusi and
  Vezzulli}{2014}]{mancusi2014r}
\textsc{Mancusi, M.~L. and A.~Vezzulli} (2014): \enquote{R\&D and credit
  rationing in SMEs,} \emph{Economic Inquiry}, 52, 1153--1172.

\bibitem[\protect\citeauthoryear{Martin}{Martin}{1996}]{martin1996r}
\textsc{Martin, S.} (1996): \enquote{R\&D joint ventures and tacit product
  market collusion,} \emph{European Journal of Political Economy}, 11,
  733--741.

\bibitem[\protect\citeauthoryear{Miyagiwa}{Miyagiwa}{2009}]{miyagiwa2009collusion}
\textsc{Miyagiwa, K.} (2009): \enquote{Collusion and research joint ventures,}
  \emph{Journal of Industrial Economics}, 57, 768--784.

\bibitem[\protect\citeauthoryear{Mohnen, Palm, Van Der~Loeff, and
  Tiwari}{Mohnen et~al.}{2008}]{mohnen2008financial}
\textsc{Mohnen, P., F.~C. Palm, S.~S. Van Der~Loeff, and A.~Tiwari} (2008):
  \enquote{Financial constraints and other obstacles: Are they a threat to
  innovation activity?} \emph{De Economist}, 156, 201--214.

\bibitem[\protect\citeauthoryear{Moraga-Gonz{\'a}lez, Motchenkova, and
  Nevrekar}{Moraga-Gonz{\'a}lez et~al.}{2022}]{moraga2022mergers}
\textsc{Moraga-Gonz{\'a}lez, J.~L., E.~Motchenkova, and S.~Nevrekar} (2022):
  \enquote{Mergers and innovation portfolios,} \emph{The RAND Journal of
  Economics}, 53, 641--677.

\bibitem[\protect\citeauthoryear{Motta and Tarantino}{Motta and
  Tarantino}{2021}]{motta2021effect}
\textsc{Motta, M. and E.~Tarantino} (2021): \enquote{The effect of horizontal
  mergers, when firms compete in prices and investments,} \emph{International
  Journal of Industrial Organization}, 78, 102774.

\bibitem[\protect\citeauthoryear{R\"oller, Siebert, and Tombak}{R\"oller
  et~al.}{2007}]{Roeller2007}
\textsc{R\"oller, L.-H., R.~Siebert, and M.~Tombak, Mikhel} (2007):
  \enquote{Why firms form (or do not form) RJVs,} \emph{Economic Journal}, 117,
  1122--1144.

\bibitem[\protect\citeauthoryear{Savignac}{Savignac}{2008}]{savignac2008impact}
\textsc{Savignac, F.} (2008): \enquote{Impact of financial constraints on
  innovation: What can be learned from a direct measure?} \emph{Economics of
  Innovation and New Technology}, 17, 553--569.

\bibitem[\protect\citeauthoryear{Schmutzler}{Schmutzler}{2013}]{schmutzler2013competition}
\textsc{Schmutzler, A.} (2013): \enquote{Competition and investment: A unified
  approach,} \emph{International Journal of Industrial Organization}, 31,
  477--487.

\bibitem[\protect\citeauthoryear{Shapiro}{Shapiro}{1985}]{shapiro1985}
\textsc{Shapiro, C.} (1985): \enquote{Patent licensing and R\&D rivalry,}
  \emph{American Economic Review, Papers and Proceedings}, 75, 25--30.

\bibitem[\protect\citeauthoryear{Sovinsky}{Sovinsky}{2022}]{helland2019research}
\textsc{Sovinsky, M.} (2022): \enquote{Do research joint ventures serve a
  collusive function?} \emph{Journal of the European Economic Association}, 20,
  430--475.

\bibitem[\protect\citeauthoryear{Suzumura}{Suzumura}{1992}]{suzumura1992cooperative}
\textsc{Suzumura, K.} (1992): \enquote{Cooperative and noncooperative R\&D in
  an oligopoly with spillovers,} \emph{American Economic Review}, 1307--1320.

\bibitem[\protect\citeauthoryear{Veugelers}{Veugelers}{1997}]{veugelers1997internal}
\textsc{Veugelers, R.} (1997): \enquote{Internal R\&D expenditures and external
  technology sourcing,} \emph{Research Policy}, 26, 303--315.

\bibitem[\protect\citeauthoryear{Vilasuso and Frascatore}{Vilasuso and
  Frascatore}{2000}]{vilasuso2000public}
\textsc{Vilasuso, J. and M.~R. Frascatore} (2000): \enquote{Public policy and
  R\&D when research joint ventures are costly,} \emph{Canadian Journal of
  Economics/Revue canadienne d'{\'e}conomique}, 33, 818--839.

\end{thebibliography}

\end{document}